\documentclass[12pt]{article}

\usepackage[english]{babel}
\usepackage[utf8]{inputenc}
\usepackage{amsmath}
\usepackage{graphicx}
\usepackage{subfigure}
\usepackage{amssymb}
\usepackage{amsthm}
\usepackage{bm}
\usepackage{tikz-cd}
\usepackage{mathrsfs}
\usepackage[colorinlistoftodos]{todonotes}
\usepackage{enumitem}
\usepackage{yfonts}
\usepackage{ dsfont }
\usepackage{geometry}
\usepackage{setspace}
\usepackage{tikz}
\usetikzlibrary{positioning}
\usetikzlibrary{decorations.pathreplacing}
\usetikzlibrary{decorations,arrows}
\usetikzlibrary{decorations.markings}
\usetikzlibrary{patterns}

\numberwithin{equation}{section}
\title{On the bifurcation of a Dirac point in a photonic waveguide without band gap opening}

\author{Jiayu Qiu\thanks{Department of Mathematics, 
 HKUST,  Clear Water Bay, Kowloon, Hong Kong SAR, China.
    \tt jqiuaj@connect.ust.hk.} \,\,\, 
 Hai Zhang
  \thanks{Department of Mathematics, 
 HKUST,  Clear Water Bay, Kowloon, Hong Kong SAR, China.
    {\tt{haizhang@ust.hk}}. H. Zhang was partially supported by Hong Kong RGC grant GRF 16307024.}}

\date{}

\newtheorem{theorem}{Theorem}[section]
\newtheorem{lemma}[theorem]{Lemma}

\newtheorem{corollary}[theorem]{Corollary}

\newtheorem{remark}[theorem]{Remark}
\newtheorem{proposition}[theorem]{Proposition}
\newtheorem{assumption}[theorem]{Assumption}
\allowdisplaybreaks[4]

\begin{document}

\maketitle

\begin{abstract}
Recent progress in topological insulators and topological phases of matter has motivated new methods for the localization of waves in photonic structures.  Especially, it is established that a Dirac point of a periodic structure can bifurcate into in-gap eigenvalues if the periodic structure is perturbed differently on the two sides of an interface and if a common band gap can be opened for the two perturbed periodic structures near the Dirac point.  The associated eigenmodes are localized near the interface and decay exponentially away from it. This paper addresses the less-known situation when the perturbation only lifts the degeneracy of the Dirac point without opening a band gap. Using a two-dimensional waveguide model, we constructed a wave mode bifurcated from a Dirac point of a periodic waveguide. We show that when the constructed mode couples to an analytically continued Floquet-Bloch mode near the Dirac energy, its eigenvalue acquires a strictly negative imaginary part, making the mode resonant. On the other hand, when the coupling vanishes, the imaginary part of the eigenvalue turns to zero, and the constructed mode becomes an interface mode that decays exponentially away from the interface.  The developed method can be extended to other settings, thus providing a clear answer to the problem concerning the bifurcation of Dirac points. 

\medskip


\end{abstract}

\section{Introduction}
The study of localized waves in photonic systems has gained significant interest due to its potential applications in designing new optical devices \cite{sheng1990,qiu2011,soukoulis2012,Ozawa_2019}. One specific type of localized wave is the interface mode whose energy is concentrated near the interface between two media. The classic way to create a localized wave mode or interface mode is to introduce a defect or random defects in a band-gapped photonic crystal; see \cite{figotin1996, figotin1997, figotin1997_jsp, figotin1998, figotin1998-1, Ammari2018-SIAMAP, Ammari2019linedefect}. The recent developments in the field of topological insulators and topological photonics/phononics\cite{Hasan10, Ozawa_2019, Khanikaev2013} have opened new avenues for creating localized modes and interface modes. One such approach is to start with a periodic structure that has a Dirac point in its spectral band where two dispersion curves/surfaces intersect linearly or conically \cite{fefferman12, Berkolaiko14, Ammari-F-H-L-Y-2020, LZ2021, Li23}, and then perturb it differently on the two sides of an interface to open up a common band gap and produce a \textit{band inversion}; that is, the two perturbed media switch their Bloch eigenspaces near the Dirac point, as specified in Remark \ref{rmk_gap_open_band_inversion}. As a consequence of the band inversion, an in-gap eigenvalue associated with the interface mode can be bifurcated from the Dirac point. It should be noted that a necessary condition for this approach of constructing interface modes is \textit{the opening of a band gap}. The condition that ensures this is referred to as the spectral no-fold condition, which is defined and discussed in \cite{fefferman2017topologically, fefferman2016honeycomb_edge}. Roughly speaking, the no-fold condition assumes the energy level of the Dirac point intersects with the two adjacent bands only at the Dirac point. It ensures the ``local'' band gap opened by the lifting of degeneracy at the Dirac point is indeed a ``global'' band gap; see the discussion under Figure \ref{fig_perturbed_band_structure}. Under the spectral no-fold condition, the in-gap eigenvalue bifurcated from Dirac points has been rigorously analyzed in various settings, including one-dimensional Schr\"{o}dinger operators \cite{fefferman2017topologically}, one-dimensional photonic structure \cite{LZ2021}, two-dimensional Schr\"{o}dinger operators \cite{fefferman2016honeycomb_edge,Drouot2019TheBC,drouot2020edge,cao2024edge}, and two-dimensional photonic structures \cite{lee2019elliptic,qiu2023mathematical,li2024interface}.

An intriguing question is whether an interface mode still can be created if the perturbations only lift the degeneracy of the Dirac point locally without opening a global band gap. In this scenario, it is conjectured that the mode bifurcated from the Dirac point will couple with other Floquet-Bloch modes of energy near the Dirac point energy level but with quasi-momentum away from the Dirac point, and then become a resonant mode \cite{fefferman2016honeycomb_edge,drouot2020edge}. To date, a definitive answer to this conjecture is not yet available. In this paper, we provide a solution to this conjecture by explicitly constructing a mode bifurcated from a Dirac point in a two-dimensional waveguide without the band gap opening condition. More precisely, we proved that even though the gap is not opened, there is indeed an eigenvalue $\lambda^{\star}$ bifurcated from the Dirac energy level with $\Im(\lambda^{\star})\leq 0$ (Theorem \ref{thm_main result}). We provide a clear characterization of the eigenmode $u^{\star}$ associated with $\lambda^{\star}$, especially on its coupling with the Bloch modes of quasi-momentum away from the Dirac point (Remark \ref{rmk_radiation_cond}). When the coupling is present, the mode $u^{\star}$ is indeed a resonant mode with unbounded $L^2$ norm. On the other hand, when the coupling vanishes, $u^{\star}$ is an interface mode which decays exponentially at infinity and $\lambda^{\star}$ is an embedded eigenvalue in the continuous spectrum (in such cases, $u^{\star}$ is also called a bound state in the continuum (BIC) \cite{hsu2016bic}).

The rest of the paper is organized in the following way. 
In Section 1.1, we provide a detailed setup of the problem and present our main results. A sketch of proof to our main result is exhibited in Section 1.2, which covers the main idea of this paper. Section 2 is a list of some preliminaries, where we briefly review the Floquet-Bloch theory for periodic differential operators and introduce the Green functions for periodic waveguide structures. In Section 3, we present the asymptotic expansions of Bloch eigenvalues and eigenfunctions near the energy level of the Dirac point; see Theorem \ref{thm_local_gap_open} and \ref{thm_perturbed_fold_dispersion}. These results demonstrate that a ``local'' band gap is opened near the Dirac point upon applying appropriate perturbations (but not a ``global'' band gap due to the failure of no-fold condition ), and the band inversion emerges for different perturbations. This provides the basis of bifurcation near the Dirac point. Finally, in Section 4, we construct a mode $u^{\star}$ bifurcated from the Dirac point using an analytic continuation of the Green function and the layer potential technique. Its associated eigenvalue $\lambda^{\star}$ has a non-positive imaginary part. Thanks to the construction of the analytic continuation of the Green function, the asymptotic behaviour of $u^{\star}$ is clearly analyzed, which enables us to characterize its bounded or resonant nature.

\subsection{Problem setup and main results}
We consider the propagation of a time-harmonic scalar wave in a two-dimensional periodic photonic waveguide $\Omega\subset \mathbf{R}^2$ (see Figure \ref{fig_waveguide}) at frequency $\sqrt{\lambda}$
\begin{equation} \label{eq_waveeq}
\left\{
\begin{aligned}
&-\frac{1}{n^2}\Delta u-\lambda u=0,\quad x\in \Omega ,\\
&\nabla u (x)\cdot \bm{n}_x=0 ,\quad x\in \partial \Omega ,
\end{aligned}
\right.
\end{equation}
where $\bm{n}_x$ denotes the outward normal at $x\in\partial \Omega$, and $n=n(x)$ is the refractive index. We assume that
\begin{enumerate}
    \item The domain $\Omega$ is connected and open in $\mathbf{R}^2$ with the boundary $\partial \Omega$ being $C^2$. Moreover, it's strip-like in the sense that there exists a compact set $S\subset \mathbf{R}$ such that $\Omega\subset \mathbf{R}\times S$;
    \item $\Omega$ is periodic in the sense that $x+\bm{e}_1\in \Omega$ for all $x\in \Omega$ ($\bm{e}_1=(1,0)$);
    \item $n\in L^{\infty}(\Omega)$, $n(x+\bm{e}_1)=n(x)$, and $n(x)\geq c>0$ for some constant $c$.
\end{enumerate}
The primitive cell of $\Omega$ is denoted by $Y:=\Omega\cap ((0,1)\times \mathbf{R})$. To simplify the analysis, we assume that system \eqref{eq_waveeq} is reflection symmetric:

\begin{assumption} \label{assump_reflection}
(1) For any $(x_1,x_2)\in\Omega$,  $(-x_1,x_2)\in\Omega$;

(2) $n(x)=(\mathcal{P}n)(x)$ for all $x\in\Omega$, where $\mathcal{P}$ is the reflection operator defined as $(\mathcal{P}u)(x_1,x_2):=u(-x_1,x_2)$.
\end{assumption}


Note that \eqref{eq_waveeq} can be viewed as the eigenvalue problem of the Laplacian operator:
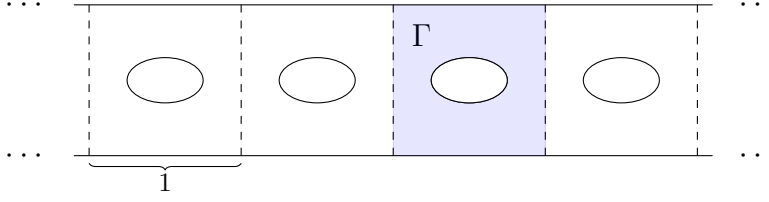
\begin{figure}
\centering
\begin{tikzpicture}
\draw (-0.2,1)--(8.2,1);
\draw (-0.2,-1)--(8.2,-1);

\node[right] at (8.4,1) {$\cdots$};
\node[right] at (8.4,-1) {$\cdots$};
\node[left] at (-0.4,1) {$\cdots$};
\node[left] at (-0.4,-1) {$\cdots$};
\draw[dashed] (0,-1)--(0,1);
\draw[dashed] (2,-1)--(2,1);
\draw[dashed] (4,-1)--(4,1);
\node[right] at (4.1,0.6) {$\Gamma$};
\draw[dashed] (6,-1)--(6,1);
\draw[dashed] (8,-1)--(8,1);

\draw (1,0) ellipse(0.5 and 0.3);
\draw (3,0) ellipse(0.5 and 0.3);
\draw (5,0) ellipse(0.5 and 0.3);
\draw (7,0) ellipse(0.5 and 0.3);

\draw[fill=blue,opacity=0.1] (4,-1) rectangle(6,1);
\draw[fill=white,opacity=1] (5,0) ellipse(0.5 and 0.3);

\draw[decorate,decoration={brace,mirror}] (0,-1.1) -- (2,-1.1);
\node[below,font=\fontsize{10}{10}\selectfont] at (1,-1.1) {$1$};
\end{tikzpicture}
\caption{A waveguide $\Omega$ with periodically arranged obstacles \cite{qiu2023mathematical}. Here $\Omega=\mathbf{R}\times [0,1]\backslash \cup_{n\in\mathbf{N}}D_n$, where $D_n$ is an array of periodically arranged obstacles that are centered at $(\frac{1}{2}+n,\frac{1}{2})$. $\Omega$ is periodic with its minimal period equal to $1$. The primitive cell $Y$ is filled with blue in the figure. The interface $\Gamma:=\Omega\cap (\{0\}\times \mathbf{R})$ is also marked in the figure.}
\label{fig_waveguide}
\end{figure}

\begin{equation*}
\mathcal{L}=-\frac{1}{n^2}\Delta: H_{b}^1(\Delta, \Omega)\subset L^2(\Omega)\to L^2(\Omega),
\end{equation*}
with
\begin{equation*}
H_{b}^1(\Delta, \Omega):=\{u\in H^1(\Omega):\Delta u\in L^2(\Omega),\, \nabla u(x)\cdot \bm{n_x}\big|_{\partial\Omega}=0\}.
\end{equation*}
By the Floquet-Bloch theory \cite{kuchment2016overview}, the spectrum of $\mathcal{L}$ admits the decomposition 
\[
\sigma(\mathcal{L})=\cup_{p\in [-\pi,\pi]}\sigma(\mathcal{L}(p)),
\]
where $\mathcal{L}(p)$ is the Floquet-Bloch component of $\mathcal{L}$ at the quasi-momentum $p\in [-\pi,\pi]$. In particular, $\mathcal{L}(p)$ can be analytically extended to $p\in \mathbf{C}$; then $\{\mathcal{L}(p)\}$ forms an analytic family of self-adjoint operators \cite{joly2016solutions}. It's known from the analytic perturbation theory that there exist analytic functions $\{\mu_n(p)\}_{n=1}^{\infty}$ such that $\sigma(\mathcal{L}(p))=\{\mu_n(p):n\geq 1\}$ for $p\in [-\pi,\pi]$. We call $\{\mu_n(p)\}_{n=1}^{\infty}$ the analytical labeling of the Floquet-Bloch eigenvalues. By the analytic perturbation theory, the dimension of the eigenspace associated with each $\mu_n(p)$ is constant for almost every $p$ \cite{kato2013perturbation}. In this paper, we assume without loss of generality they are one-dimensional a.e.:
\begin{assumption} \label{assump_general spectrum}
For each $n\geq 1$, the eigenspace associated with $\mu_n(p)$ is one-dimensional except for finitely many $p\in [-\pi,\pi]$.
\end{assumption}

With Assumption \ref{assump_general spectrum}, we denote the Floquet-Bloch eigenspace corresponding to $\mu_n(p)$ ($n\geq 1$, $p\in [-\pi,\pi]$) by $\text{span}\{v_n(x;p)\}$. Each of the Bloch modes $v_n(x;p)$ is analytically continued as a function of $p$ to a neighborhood of the real axis \cite{joly2016solutions}; the analytic domain of $v_n(x;p)$ will be specified in this paper if necessary. For ease of notation, the continued Bloch modes are still denoted as $v_n(x;p)$. The possible exceptional $p$'s in Assumption \ref{assump_general spectrum} usually occur at the intersection of two graphs of $\mu_n(p)$'s. Dirac points are among such intersection points.  We assume the existence of a Dirac point at $p=0$ and $\lambda=\lambda_*>0$. In particular, we assume that the following holds. 


\begin{assumption} \label{assump_dirac_points}
There exist $n_{*},m_{*}\in \mathbf{N}$, $q_{*}\in (0,\pi)$ and $\lambda_*>0$ such that

(1)the dispersion curves of $\lambda=\mu_{n_*}(\cdot)$ and $\lambda=\mu_{m_*}(\cdot)$ intersect with the energy level $\lambda=\lambda_*$ at $p=-q_*$ and $p=q_*$, and at the Dirac point $(p,\lambda)=(0,\lambda_*)$, i.e., $\lambda_*=\mu_{n_{*}}(0)=\mu_{m_{*}}(0)=\mu_{n_{*}}(-q_{*})=\mu_{m_{*}}(q_{*})$;


(2) the dispersion curves of $\mu_{n_{*}}(\cdot)$ and $\mu_{m_{*}}(\cdot)$ do not intersect with those of other $\mu_n(\cdot)$'s. More precisely, for any $p\in (-\pi,\pi]$ and $n\in \mathbf{N}\backslash \{n_*,m_*\}$, $\mu_{n}(p)\neq \mu_{n_{*}}(p)$, $\mu_{n}(p)\neq \mu_{m_{*}}(p)$; for any $p\neq (-\pi,\pi]\backslash \{0\}$, $\mu_{n_{*}}(p)\neq \mu_{m_{*}}(p)$;

(3)$\mu_{n_*}^\prime(0)>0$, $\mu_{n_*}^\prime(-q_*)<0$.
\end{assumption}

A band structure as described in Assumption \ref{assump_general spectrum} and \ref{assump_dirac_points} is depicted in Figure \ref{fig_unperturbed_band_structure_a}.

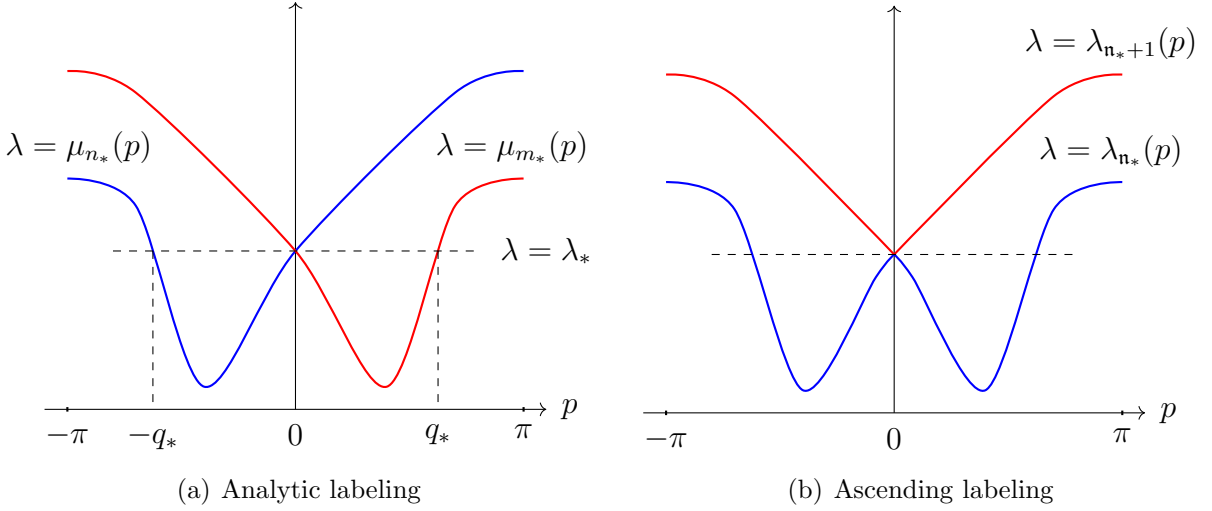
\begin{figure}
\centering
\subfigure[Analytic labeling]{
\label{fig_unperturbed_band_structure_a}     
\begin{tikzpicture}[scale=0.3]
\draw[->] (-11,0)--(11,0);
\draw[->] (0,0)--(0,18);
\node[right] at (11.2,0) {$p$};
\node[below] at (0,-0.2) {$0$};
\draw[thick] (-10,-0.1)--(-10,0.1);
\node[below] at (-10,-0.2) {$-\pi$};
\draw[thick] (10,-0.1)--(10,0.1);
\node[below] at (10,-0.2) {$\pi$};

\draw[thick,blue] plot [smooth] coordinates {(-10,10.2) (-7,9) (-4,1) (0,7) (7,14) (10,14.95)};
\node[above] at (-9.5,10.4) {$\lambda=\mu_{n_*}(p)$};

\draw[thick,red] plot [smooth] coordinates {(10,10.2) (7,9) (4,1) (0,7) (-7,14) (-10,14.95)};
\node[above] at (9.5,10.4) {$\lambda=\mu_{m_*}(p)$};

\draw[dashed] (-8,7)--(8,7);
\node[right] at (8.5,7) {$\lambda=\lambda_*$};
\draw[dashed] (-6.25,7)--(-6.25,0);
\node[below] at (-6.25,-0.2) {$-q_*$};
\draw[dashed] (6.25,7)--(6.25,0);
\node[below] at (6.25,-0.2) {$q_*$};
\end{tikzpicture}
}     
\subfigure[Ascending labeling] { 
\label{fig_unperturbed_band_structure_b}
\begin{tikzpicture}[scale=0.3]
\draw[->] (-11,0)--(11,0);
\draw[->] (0,0)--(0,18);
\node[right] at (11.2,0) {$p$};
\node[below] at (0,-0.2) {$0$};
\draw[thick] (-10,-0.1)--(-10,0.1);
\node[below] at (-10,-0.2) {$-\pi$};
\draw[thick] (10,-0.1)--(10,0.1);
\node[below] at (10,-0.2) {$\pi$};

\draw[thick,blue] plot [smooth] coordinates {(-10,10.2) (-7,9) (-4,1) (-0.8,6) (0,7)};
\draw[thick,red] plot [smooth] coordinates {(0,7) (7,14) (10,14.95)};
\node[above] at (9.5,15.15) {$\lambda=\lambda_{\mathfrak{n}_*+1}(p)$};

\draw[thick,blue] plot [smooth] coordinates {(10,10.2) (7,9) (4,1) (0.8,6) (0,7)};
\draw[thick,red] plot [smooth] coordinates {(0,7) (-7,14) (-10,14.95)};
\node[above] at (9.5,10.4) {$\lambda=\lambda_{\mathfrak{n}_*}(p)$};

\draw[dashed] (-8,7)--(8,7);
\end{tikzpicture}
}    
\caption{Band structure of $\mathcal{L}$ near the Dirac point $(0,\lambda_*)$. (a) $\lambda=\mu_{n_*}(p)$ is plotted in blue while $\lambda=\mu_{m_*}(p)$ is plotted in red; they are smooth on $[-\pi,\pi]$ and constitute the band structure of $\mathcal{L}$ near the Dirac point. Moreover, $\mu_{n_*}(p)$ and $\mu_{m_*}(p)$ intersect with the energy level $\lambda=\lambda_*$ at $p=-q_*$ and $p=q_*$, respectively. The existence of those extra intersection points breaks the spectral no-fold condition. (b) $\lambda=\lambda_{\mathfrak{n}_*}(p)$ is plotted in blue while $\lambda=\lambda_{\mathfrak{n}_*+1}(p)$ is plotted in red; they are not differentiable at $p=0$.}
\label{fig_unperturbed_band_structure}
\end{figure}

\begin{remark}
The conditions in Assumption \ref{assump_dirac_points} (1) break the so-called spectral no-fold condition \cite{fefferman2017topologically, fefferman2016honeycomb_edge}. 
As a result, a band gap cannot be opened under small perturbations (see Figure \ref{fig_perturbed_band_structure}). The main focus of this paper is the bifurcation of the Dirac point in such a ``no-gap'' case.
\end{remark}

\begin{remark}
The condition in Assumption \ref{assump_dirac_points} (2) is imposed to simplify the proof of Lemma \ref{lem_analytic_domain}. It is not essential for the validity of the main result (Theorem \ref{thm_main result}). Assumption \ref{assump_dirac_points} (3) ensures linear dispersions near the Dirac energy $\lambda_*$.
\end{remark}

Now we introduce perturbations to the system \eqref{eq_waveeq}. Consider a family of operators $\{\mathcal{L}_{\epsilon}=-\frac{1}{n_{\epsilon}^2}\Delta \}$ $(|\epsilon|\ll 1)$, where $n_\epsilon(x)$ satisfies the following properties. 

\begin{assumption} \label{assump_perturbation}
(1)\,The function $\epsilon\mapsto n_{\epsilon}(x)$ is $C^2$ for each fixed $x$; $n_{0}(x)=n(x)$;

(2)\,$n_{\epsilon}(x+\bm{e}_1)=n_{\epsilon}(x)$, $n_{\epsilon}(x)=(\mathcal{P}n_{\epsilon})(x)$ for all $|\epsilon|\ll 1$;

(3)\,Let $A(p):=-\frac{2}{n}\frac{\partial \left(n_{\epsilon}\right)}{\partial \epsilon}\big|_{\epsilon=0}\cdot \mathcal{L}(p)$. Then
\begin{equation*}
t_{*}:=\int_{Y}A(0)v_{m_{*}}(x;0)\cdot\overline{v_{n_{*}}(x;0)}n^2(x)dx\neq 0,
\end{equation*}
\begin{equation*}
\int_{Y}A(0)v_{m_{*}}(x;0)\cdot\overline{v_{m_{*}}(x;0)}n^2(x)dx=0,\quad
\int_{Y}A(0)v_{n_{*}}(x;0)\cdot\overline{v_{n_{*}}(x;0)}n^2(x)dx
=0.
\end{equation*}
\end{assumption}

\begin{remark}
Assumption \ref{assump_perturbation} (3) can be relaxed to the following one
\begin{equation} \label{eq_non_degenerate}
2|t_*|>\left|\int_{Y}A(0)v_{n_{*}}(x;0)\cdot\overline{v_{n_{*}}(x;0)}n^2(x)dx
+\int_{Y}A(0)v_{m_{*}}(x;0)\cdot\overline{v_{m_{*}}(x;0)}n^2(x)dx\right|,
\end{equation}
without affecting the main result of this paper. In fact, \eqref{eq_non_degenerate} ensures that $\mathcal{L}_{\epsilon}$ and $\mathcal{L}_{-\epsilon}$ attain a common ``local'' band gap near the Dirac point $(0, \lambda_*)$ for small $\epsilon$, and that they switch their Bloch eigenspaces therein (see the proof of Theorem \ref{thm_local_gap_open} and the remark followed), which are essential for the bifurcation of eigenvalue from the Dirac point. Nonetheless, we do not use Condition \eqref{eq_non_degenerate} for ease of presentation.
\end{remark}

We are concerned with the bifurcation of the Dirac point in the following joint structure
\begin{equation} \label{eq_joint_system}
\left\{
\begin{aligned}
&\mathcal{L}^{\star}_\epsilon u-\lambda u=0,\quad x\in \Omega ,\\
&\nabla u (x)\cdot \bm{n}_x=0 ,\quad x\in \partial \Omega ,
\end{aligned}
\right.
\end{equation}
where
\begin{equation*}
(\mathcal{L}^{\star}_\epsilon u)(x_1,x_2):=
\left\{
\begin{aligned}
&(\mathcal{L}_{\epsilon}u)(x_1,x_2),\quad x_1>0, \\
&(\mathcal{L}_{-\epsilon}u)(x_1,x_2),\quad x_1<0.
\end{aligned}
\right.
\end{equation*}


Let $\mathcal{I}_{\epsilon}:=\{\lambda\in\mathbf{C}:|\lambda-\lambda_*|<c_0|t_*|\epsilon\}$, where $c_0$ is any positive number such that $c_0<1$ (fixed throughout this paper). Our main result is stated below:
\begin{theorem} \label{thm_main result}
Under Assumptions \ref{assump_reflection},\ref{assump_general spectrum}, \ref{assump_dirac_points} and \ref{assump_perturbation}, there exists $\epsilon_0>0$ such that for any $|\epsilon|<\epsilon_0$, \eqref{eq_joint_system} has a solution $u^{\star}$ with $\lambda^{\star}\in \mathcal{I}_{\epsilon}$ and $\Im(\lambda^{\star})\leq 0$. In particular, when $\Im(\lambda^{\star})=0$, $\lambda^{\star}$ an embedded eigenvalue and $u^{\star}$ is an interface mode that decays exponentially away from the interface $x_1=0$;  
When $\Im(\lambda^{\star})<0$, $u^{\star}$ is a resonant mode with infinite $L^2$-norm. 
\end{theorem}

\begin{remark} \label{rmk_radiation_cond}
In the case of 
$\Im(\lambda^{\star})=0$ in Theorem \ref{thm_main result}, we have (by Proposition \ref{corollary_G_eps_decay_blow})
\[
\langle \phi^{\star},\overline{u_{\mathfrak{n}_*,\epsilon}(\cdot ;\overline{q_{+,\epsilon}(\lambda^{\star})})} \rangle = \langle \phi^{\star},\overline{u_{\mathfrak{n}_*,-\epsilon}(\cdot ;\overline{q_{-,-\epsilon}(\lambda^{\star})})} \rangle= 0,
\]
where $\phi^{\star}=\big(\frac{\partial u^{\star}}{\partial x_1}\big)\big|_{\Gamma}$ is the Neumann trace of $u^{\star}$ on the interface, $\langle\cdot,\cdot\rangle$ is the dual product between functions on the interface (see the notation in Section 1.4), and
$u_{\mathfrak{n}_*,\pm\epsilon}(\cdot ;\overline{q_{\pm, \pm \epsilon}(\lambda^{\star})})$ is the Bloch mode at the energy level $\lambda^{\star}$ with quasi-momentum $\overline{q_{\pm,\pm \epsilon}}$ (which is exactly $q_{\pm,\pm \epsilon}$ since we assume $\Im(\lambda^{\star})=0$) in the media corresponding to the operator $\mathcal{L}_{\pm\epsilon}$. Physically, it means that $u^{\star}$ is decoupled to the two outgoing Bloch modes of the joint structure.

On the other hand, in the case of $\Im(\lambda^{\star})<0$, the resonant mode $u^{\star}$ found in Theorem \ref{thm_main result} admits the following outgoing radiation condition (see Proposition \ref{prop_G_eps_cont_jump+radiation}):
\begin{equation} \label{eq_radiation_resonant_mode_1}
\lim_{x_1\to \infty}\Big|u^{\star}(x)
-\frac{-i\langle \phi^*(\cdot),\overline{u_{\mathfrak{n}_*,\epsilon}(\cdot ;\overline{q_{+,\epsilon}(\lambda^{\star})})} \rangle}{\lambda^{\prime}_{\mathfrak{n}_*,\epsilon}(q_{+,\epsilon}(\lambda^{\star}))}u_{\mathfrak{n}_*,\epsilon}(x ;q_{+,\epsilon}(\lambda^{\star}))\Big|=0,
\end{equation}
and
\begin{equation} \label{eq_radiation_resonant_mode_2}
\lim_{x_1\to -\infty}\Big|u^{\star}(x)
-\frac{-i\langle \phi^*(\cdot),\overline{u_{\mathfrak{n}_*,-\epsilon}(\cdot ;\overline{q_{-,-\epsilon}(\lambda^{\star})})} \rangle}{\lambda^{\prime}_{\mathfrak{n}_*,-\epsilon}(q_{-,-\epsilon}(\lambda^{\star}))}u_{\mathfrak{n}_*,-\epsilon}(x ;q_{-,-\epsilon}(\lambda^{\star}))\Big|=0,
\end{equation}
where $u_{\mathfrak{n}_*, \pm \epsilon}(\cdot ;\overline{q_{\pm, \pm \epsilon}(\lambda^{\star})})$ is the continued outgoing Bloch mode at the complex energy level $\lambda^{\star}$ with complex quasi-momentum $q_{\pm, \pm\epsilon}$ in the media corresponding to the operator $\mathcal{L}_{\pm\epsilon}$. Physically, it means that $u^{\star}$ is asymptotically proportional to the continued outgoing Bloch mode of $\mathcal{L}_{\epsilon}$ as $x_1 \to \infty$ and the continued Bloch mode of $\mathcal{L}_{-\epsilon}$ as $x_1 \to -\infty$. 
We remark that the continued outgoing Bloch modes are the analytical continuation of the conventional outgoing Bloch modes for real energies that are specified by the limiting absorption principle. Those continued outgoing Bloch modes specify the outgoing condition for the resonant mode. 


The radiation conditions \eqref{eq_radiation_resonant_mode_1} and \eqref{eq_radiation_resonant_mode_2} also lead to the following condition that ensures $u^{\star}$ in Theorem \ref{thm_main result} is a resonant mode: 
\begin{center}
either $\langle \phi^{\star},\overline{u_{\mathfrak{n}_*,\epsilon}(\cdot ;\overline{q_{+,\epsilon}(\lambda^{\star})})} \rangle\neq 0$ or $\langle \phi^{\star},\overline{u_{\mathfrak{n}_*,-\epsilon}(\cdot ;\overline{q_{-,-\epsilon}(\lambda^{\star})})} \rangle\neq 0$.
\end{center}
In other words, it must couple with one of the two continued outgoing Bloch modes in the joint structure at the complex energy level $\lambda^{\star}$.
\end{remark}



\subsection{Main idea of proof}

For the reader's convenience, we sketch the proof of Theorem \ref{thm_main result} here. 
Roughly speaking, the study of eigenvalues and or resonances of an operator $\mathcal{L}$ in an unbounded domain is most conveniently studied via its resolvent
$\mathcal{R}(z) \;=\; (\mathcal{L}-z)^{-1}$ (or its Green function),
which is holomorphic for \(\Im z>0\).  One then analytically continues \(\mathcal{R}(z)\) across the real axis. By the limiting absorbing principle, the limit $\mathcal{R}(\lambda + i0)$ on the real axis encodes the outgoing radiation condition. Poles with \(\Im z<0\) are called \emph{resonances}, while poles on the real axis coincide with eigenvalues (henceforth viewed simply as resonances of zero imaginary part).  
Expanding $\mathcal{R}(z)$ in a Laurent series about a pole $z_{*}$, the coefficient of $(z-z_{*})^{-1}$ gives the projection onto the space of resonant modes at $z_{*}$.

This framework applies readily to operators that are globally ``homogeneous'', such as a periodic operator with a compact defect \cite{gerard1990resonance,bentosela1976scattering}, since in that setting one can explicitly construct the analytic continuation of its resolvent. 
However, for the interface operator $\mathcal{L}_{\epsilon}^{\star}$ considered in Theorem \ref{thm_main result}, the corresponding media consists of two different periodic structures. This lack of global homogeneity prevents a direct implementation of the above approach.

Instead, we construct resonant modes near a chosen energy via layer potentials, a strategy akin to resolvent continuation. Specifically, we begin with the physical Green functions $G_{\pm\epsilon}(x,y;\lambda)$ \eqref{eq_G_eps_def} of the perturbed operator $\mathcal{L}_{\pm \varepsilon}$, obtained by the limiting-absorption principle. These Green functions satisfy the outgoing radiation condition, are analytic for \(\Im\lambda>0\), and continuous up to \(\Im\lambda\ge0\).  Localizing at the Dirac energy \(\lambda_*\), we analytically continue each \(G_{\pm\epsilon}\) in \(\lambda\) and denote the result by \(\tilde G_{\pm\epsilon}(x,y;\lambda)\). Using those continued Green functions, we seek a resonant mode which satisfies equation (1.3) in the following form
\begin{equation} \label{eq_sketch_1}
u(x;\lambda)=
\left\{
\begin{aligned}
&\int_{\Gamma}\tilde{G}_{\epsilon}(x,y;\lambda)\varphi(y)dy,\quad x_1>0, \\
&-\int_{\Gamma}\tilde{G}_{-\epsilon}(x,y;\lambda)\varphi(y)dy,\quad x_1<0  ,
\end{aligned}
\right.
\end{equation}
which by construction satisfies the outgoing radiation condition. Continuity across the interface \(\Gamma\) then reads as 
\begin{equation} \label{eq_sketch_2}
\mathbb{G}_{\epsilon}^{\Gamma}(\lambda)\varphi
:=\big(\tilde{\mathbb{G}}_{\epsilon}(\lambda)\varphi+\tilde{\mathbb{G}}_{-\epsilon}(\lambda)\varphi\big)\big|_{\Gamma}=0,
\end{equation}
where each $\tilde{\mathbb{G}}_{\pm\epsilon}(\lambda)$ is the integral operator on the interface $\Gamma$ associated with the Green function $\tilde{G}_{\pm\epsilon}(x,y;\lambda)$. 
In this way, resonances of the operator 
$\mathcal{L}_{\epsilon}^{\star}$ near $\lambda_*$
reduce to finding characteristic values of the operator-valued function $\mathbb{G}_{\epsilon}^{\Gamma}$ near $\lambda_*$. 
Solving \(\mathbb G_\epsilon^\Gamma(\lambda)\varphi=0\) with \(\Im\lambda<0\) (resp.\ \(\Im\lambda=0\)) yields a resonance (resp.\ an embedded interface mode), and the corresponding \(u\) has unbounded (resp. bounded) $L^2$ norm. See Section 4.2 for details.


Before proceeding to analyze \eqref{eq_sketch_2}, we summarize the main idea behind the analytic continuation of the Green function $G_{\pm\epsilon}(x,y;\lambda)$, which is the key to the above construction of resonant or interface modes:
\begin{itemize}
    \item We decompose 
\[
  G_{\pm\epsilon}(x,y;\lambda)
  = G_{\pm\epsilon}^{I}(x,y;\lambda)
  + G_{\pm\epsilon}^{II}(x,y;\lambda),
\]
where \(G_{\pm\epsilon}^{I}(x,y;\lambda)\) contains all contributions from energy bands that do not intersect the Dirac level \(\lambda_*\) (i.e.\ those with \(n\neq\mathfrak n_*\) in the Floquet expansion; see \eqref{eq_G_eps_I_def}), \(G_{\pm\epsilon}^{II}(x,y;\lambda)\) is the single‐band contribution corresponding to \(n=\mathfrak n_*\) (see \eqref{eq_G_eps_II_def}).

 \item $G_{\pm\epsilon}^{I}(x,y;\lambda)$ is analytic in a neighborhood of $\lambda=\lambda_*$, denoted as $\mathcal{I}_{\epsilon}$. Indeed, it is the integral kernel of the resolvent of the self‐adjoint operator obtained by projecting the Laplacian $\mathcal{L}_{\pm \epsilon}$ onto all bands with \(n\neq\mathfrak n_*\), and this operator has a spectral gap at \(\lambda_*\) (see Lemma 4.1(3)). 
    \item $G_{\pm\epsilon}^{II}(x,y;\lambda)$ has a nonzero jump as $\lambda$ sweeps across the interval $\mathcal{I}_{\epsilon}\cap \mathbf{R}$ because the band $n=\mathfrak{n}_{*}$ intersects with the Dirac energy level. Such singularity can be avoided by deforming the integral path $[-\pi,\pi]+0i$ in (4.2) of the momentum variable $p$. The key to this deformation is to avoid those $p$'s such that $\lambda_{\mathfrak{n}_{*}}(p)=\lambda$ ($\lambda\in \mathcal{I}_{\epsilon}$), which are called \textit{the pinching singularities of deformation} \cite{gerard1990resonance}. Based on a careful estimate of the location of pinching singularities, we design a continuous family of integral contours $\{C_{\epsilon}\}_{\epsilon}$ such that $C_{\epsilon}$ avoids all the pinching singularities for each $\epsilon$. As a consequence, by deforming the integral path in $G_{\pm\epsilon}^{II}(x,y;\lambda)$ to the new contour $C_{\pm \epsilon}$, we obtain an analytic function $\tilde{G}_{\pm\epsilon}(x,y;\lambda)$ of $\lambda$; see (4.3). Moreover, our construction of contour ensures that $\tilde{G}_{\pm\epsilon}(x,y;\lambda)=G_{\pm\epsilon}(x,y;\lambda)$ for $\lambda\in\mathbf{R}$. Hence, we obtain an analytic continuation of the physical Green function $G_{\pm\epsilon}(x,y;\lambda)$.
\end{itemize}
Next, we solve the characteristic value problem \eqref{eq_sketch_2} near the Dirac energy $\lambda_*$. By the generalized Rouché theorem, it's sufficient to study the characteristic value of the scaled limit $\mathbb{G}_{0}^{\Gamma}(h):=\lim_{\epsilon\to 0^+}\mathbb{G}_{\epsilon}^{\Gamma}(\lambda_*+\epsilon\cdot h)$. The convergence of $\lim_{\epsilon\to 0^+}\mathbb{G}_{\epsilon}^{\Gamma}(\lambda_*+\epsilon\cdot h)$ (Proposition 4.7) is analyzed again by decomposing \(\mathbb{G}_{\epsilon}^{\Gamma}(\lambda_{*}+\epsilon h)\)
into three band‐wise parts and treating each separately: 
\begin{itemize}

\item Far-energy bands  
    with \(n\neq\mathfrak n_*,\mathfrak n_*+1\) that are bounded away from the Dirac energy \(\lambda_*\) uniformly for $\epsilon$: Their total contribution
   \(\mathbb{G}_{\epsilon}^{\Gamma,\mathrm{evan}}(\lambda_*+\epsilon\cdot h) \) (see \eqref{eq_B_decom_1}) converges by standard regular perturbation theory (cf. \cite{kato2013perturbation})  to a limit \(\mathbb{T}^{\mathrm{evan}}\); see Lemma \ref{lem_app_B_1}.

\item Near-energy bands \(n\in\{\mathfrak n_*,\mathfrak n_*+1\}\) with momentum away from the Dirac point ($|p|>\epsilon^{1/3}$): 
Their contribution is set as
\(\mathbb{G}_{\epsilon}^{\Gamma,\mathrm{prop},1}(\lambda_*+\epsilon\cdot h)\) (see \eqref{eq_B_decom_2}).  Linear dispersion at the Dirac point of the unperturbed structure (Assumption \ref{assump_dirac_points} (3)) ensures that the real‐axis contour tends to a principal‐value integral, while small semicircles around \(p=0,\pm q_*\) produce purely imaginary limits (cf. \cite{joly2016solutions}); see Lemma \ref{lem_app_B_2}.

\item  Near-energy bands \(n\in\{\mathfrak n_*,\mathfrak n_*+1\}\) with momentum near the Dirac point ($|p|<\epsilon^{1/3}$):    
Their contribution is set as 
   \(\mathbb{G}_{\epsilon}^{\Gamma,\mathrm{prop},2}(\lambda_*+\epsilon\cdot h)\) (see \eqref{eq_B_decom_3}).  The perturbed dispersion
   \(\sqrt{p^2+\epsilon^2}\) near the Dirac point from Theorem \ref{thm_local_gap_open}, a consequence of the linear dispersion of the unperturbed operator $\mathcal{L}$ at the Dirac point, yields uniform convergence of this contribution; see Lemma \ref{lem_app_B_3}.
\end{itemize}

Together, these three estimates establish that
\begin{equation} \label{eq_sketch_3}
    \lim_{\epsilon\to 0^+}\mathbb{G}_{\epsilon}^{\Gamma}(\lambda_*+\epsilon\cdot h)=2\mathbb{T}+\beta(h)\mathbb{P}^{Dirac} .
\end{equation}
We shortly illustrate the meaning of this limiting operator, which is critical for studying the characteristic value problem. Using the limiting absorption principle, every solution at the Dirac energy level in the unperturbed structure can be decomposed into three parts: (i) two propagating modes away from the Dirac point: $v_{n_*}(x;-q_*)$ and $v_{m_*}(x;q_*)$ (see Figure \ref{fig_unperturbed_band_structure_a}), (ii) two propagating modes at the Dirac point: $v_{n_*}(x;0)$ and $v_{m_*}(x;0)$, and (iii) an evanescent mode that decays exponentially at infinity. 
Correspondingly, the operator $\mathbb{T}$ in \eqref{eq_sketch_3} encodes the contributions of parts (i) and (iii), while the projector $\mathbb{P}^{Dirac}$ encodes that of part (ii). Moreover, the function $\beta(h)$ encodes information of the band inversion near the Dirac point due to the perturbation.  
We prove the limiting operator \eqref{eq_sketch_3} has a unique characteristic value $h=0$. The generalized Rouché theorem then guarantees that the original characteristic‐value problem \eqref{eq_sketch_2} likewise has a unique solution in a neighborhood of the Dirac energy \(\lambda_*\). 


To conclude the proof of Theorem \ref{thm_main result}, we need to show that the solution $\lambda^{\star}$ to \eqref{eq_sketch_2} corresponds to an embedded eigenvalue or resonance as claimed in Theorem \ref{thm_main result}. The key is to analyze the asymptotic of the associated mode constructed in \eqref{eq_sketch_1} at infinity and its coupling with the propagating modes or extended propagating modes near $p=\pm q_*$. See Proposition \ref{prop_G_eps_cont_jump+radiation} and \ref{corollary_G_eps_decay_blow}.


\subsection{Notations}
Here we list notations that are used in the paper.

\subsubsection{Geometries}

Upper/lower half complex plane \noindent $\mathbf{C}_+=\{z \in \mathbf{C}: \Im(z) >0\}$, $\mathbf{C}_-=\{z \in \mathbf{C}: \Im(z) <0\}$;

\noindent $\Omega$: the domain of the waveguide (introduced in Section 1.1);

\noindent $\Omega^{+}:=\Omega\cap (\mathbf{R}^{+} \times \mathbf{R})$, \quad $\Omega^{-}:=\Omega\cap (\mathbf{R}^{-} \times \mathbf{R})$;

\noindent Interface $\Gamma:=\Omega\cap (\{0\}\times \mathbf{R})$;

\noindent $\Gamma^{+}:=\partial (\Omega^{+})$, \quad $\Gamma^{-}:=\partial (\Omega^{-})$;

\noindent Primitive cell $Y:=\Omega\cap ((0,1)\times \mathbf{R})$. 

\subsubsection{Function spaces}
\noindent $L^2(\Omega):=\{u(x):\|u\|_{L^2(\Omega)}<\infty\}$, where $\|\cdot\|_{L^2(\Omega)}$ is induced by the inner product $(u,v)_{L^2(\Omega)}:=\int_{\Omega}u\cdot\overline{v}$;

\noindent $H^m(\Omega):=\{u(x):\partial_\alpha u\in L^2(\Omega),\, |\alpha|\leq m\}$;

\noindent $L^2(Y):=\{u(x):\|u\|_{L^2(Y)}<\infty\}$, where $\|\cdot\|_{L^2(Y)}$ is induced by the inner product $(u,v)_{L^2(Y)}:=\int_{Y}u\cdot\overline{v}$;

\noindent $L^2(Y;n_{\epsilon}(x)):=\{u(x):\|u\|_{L^2(Y;n_{\epsilon}(x))}<\infty\}$, where $\|\cdot\|_{L^2(Y;n_{\epsilon}(x))}$ is induced by the inner product $(u,v)_{L^2(Y;n_{\epsilon}(x))}:=\int_{Y}n^2_{\epsilon}(x)u(x)\cdot\overline{v(x)}$ (the refractive index $n_{\epsilon}$ is introduced in Assumption \ref{assump_perturbation});

\noindent $H^m(Y):=\{u(x):\partial_\alpha u \in L^2(Y),\, |\alpha|\leq m\}$;

\noindent $H_{b}^1(\Delta, \Omega):=\{u\in H^1(\Omega):\Delta u\in L^2(\Omega),\, \nabla u(x)\cdot \bm{n_x}\big|_{\partial\Omega}=0\}$;

\noindent $L_{p}^2(\Omega):=\{u\in L_{loc}^2(\Omega):\, u(x+\bm{e}_1)=e^{ip}u(x)\}$ ($p\in\mathbf{C}$), which is equipped with $L^2(Y)-$norm;

\noindent $H^m_{p}(\Omega):=\{u\in H_{loc}^m(\Omega):\, (\partial_\alpha u)(x+\bm{e}_1)=e^{ip}(\partial_\alpha u)(x),\, 
|\alpha|\leq m\}$ ($p\in\mathbf{C}$), which is equipped with $H^m(Y)-$norm;

\noindent $H^m_{p,b}(\Omega):=\{u\in H_{p}^m(\Omega):\nabla u(x)\cdot \bm{n_x}\big|_{\partial\Omega}=0\}$;

\noindent $H_{p,b}^1(\Delta, \Omega):=\{u\in H^1_{p,b}(\Omega):\Delta u\in L^2_{loc}(\Omega)\}$;

\noindent $H^{\frac{1}{2}}(\Gamma):=\{u=U|_{\Gamma}:U\in H^{\frac{1}{2}}(\Gamma^{+})\}$, where $H^{\frac{1}{2}}(\Gamma^{+})$ is defined in the standard way;

\noindent $\tilde{H}^{-\frac{1}{2}}(\Gamma):=\{u=U|_{\Gamma}:U\in H^{-\frac{1}{2}}(\Gamma^{+})\text{ and }supp (U)\subset \overline{\Gamma}\}$. 

\subsection{Operators and others}

\noindent Equivalence $\sim$ between two functions: $u\sim v$ $\Longleftrightarrow$ $\exists \tau\in\mathbf{C}$ such that $|\tau|=1$ and $u=\tau\cdot v$;

\noindent Dual product $\langle\cdot,\cdot \rangle$ (between $\tilde{H}^{-\frac{1}{2}}(\Gamma)$ and $H^{\frac{1}{2}}(\Gamma)$): $\langle \varphi,\phi\rangle=\int_{\Gamma}\varphi\cdot \phi$, for $\varphi\in \tilde{H}^{-\frac{1}{2}}(\Gamma)$, $\phi\in H^{\frac{1}{2}}(\Gamma)$;

\noindent Reflection operator $\mathcal{P}:u(x_1,x_2)\mapsto u(-x_1,x_2)$;

\noindent Trace operator $\text{Tr}:H^1(Y)\to H^{\frac{1}{2}}(\Gamma)$, $u\mapsto u\big|_{\Gamma}$;

\noindent Extension operator $\mathcal{M}:=\text{Tr}^{*}$.

\section{Preliminaries}

\subsection{Floquet-Bloch theory and band structure near the Dirac points}
In this section, we briefly recall the Floquet-Bloch theory, which is used to characterize the spectrum of the operators $\{\mathcal{L}_{\epsilon}\}$. Here we only exhibit the result for $\mathcal{L}=\mathcal{L}_{0}$ for the ease of notation; the discussion of $\mathcal{L}_{\epsilon}$ ($\epsilon\neq 0$) is the same.

To study the spectrum $\sigma(\mathcal{L})$, it's equivalent to consider the following family of operators for $p\in [-\pi, \pi]$, 
$$
\mathcal{L}(p):H_{p,b}^1(\Delta,\Omega)\subset L^2_{p}(\Omega)\to L^{2}_{p}(\Omega),\quad \phi\to -\frac{1}{n^2}\Delta.
$$
The Floquet-Bloch theory indicates that
$$
\sigma(\mathcal{L})=\cup_{p\in [-\pi,\pi]}\sigma(\mathcal{L}(p)).
$$
For each $p\in [-\pi,\pi]$, $\sigma(\mathcal{L}(p))$ consists of a discrete set of real eigenvalues
\begin{equation*} 
0\leq \lambda_1(p)\leq\lambda_2(p)\leq\cdots\leq\lambda_n(p)\leq\cdots.
\end{equation*}
Since $\{\lambda_{n}(p)\}$ are labeled in the ascending order, we call it the \textit{ascending labeling} of the Floquet-Bloch eigenvalues.  It is clear that
\begin{equation*}
\forall p\in [-\pi,\pi],\quad \{\mu_n(p)\}_{n=1}^{\infty}=\{\lambda_n(p)\}_{n=1}^{\infty},
\end{equation*}
where $\{\mu_n(p)\}_{n=1}^{\infty}$ is the analytical labeling of the Floquet-Bloch eigenvalues introduced in Section 1.1. Each $\lambda_n(p)$ is piecewise smooth for $p\in [-\pi,\pi]$ \cite{kuchment2016overview}. The graph of $\lambda_n(p)$ is called the $n$-th dispersion curve. 
We denote by $u_{n}(x;p)$ the $L^2$-normalized eigenfunction associated with the eigenvalue $\lambda_n(p)$. $u_{n}(x;p)$ is called the $n$-th Floquet-Bloch mode at quasi-momentum $p$. 
The Bloch modes $\{u_{n}(x;p): n \geq 1\}$ form a basis of $L_p^2(\Omega)$.
\begin{proposition}
Assume \ref{assump_general spectrum} and \ref{assump_dirac_points} hold.  There exists an integer $\mathfrak{n}_*>0$ such that 
\begin{enumerate}
    \item $\lambda_*=\lambda_{\mathfrak{n}_*}(0)=\lambda_{\mathfrak{n}_*+1}(0)$;
    \item $\lambda_{\mathfrak{n}_*}(p)<\lambda_{\mathfrak{n}_*+1}(p)$ for all $p\in [-\pi,\pi]\backslash \{0\}$;
    \item $\lambda_{\mathfrak{n}_*-1}(p)<\lambda_{\mathfrak{n}_*}(p)$, $\lambda_{\mathfrak{n}_*+1}(p)<\lambda_{\mathfrak{n}_*+2}(p)$ for all $p\in [-\pi,\pi]$.
\end{enumerate}
\end{proposition}
The two dispersion curves $\lambda_*=\lambda_{\mathfrak{n}_*}(p)$ and $\lambda_*=\lambda_{\mathfrak{n}_*+1}(p)$ are depicted in Figure \ref{fig_unperturbed_band_structure_b}.

\subsection{The physical Green function and representation of solutions for the unperturbed structure}
We introduce the physical Green function $G(x, y; \lambda)$ for the unperturbed waveguide at $\lambda=\lambda_*$. It's the solution to the following equations:
\begin{equation*} 
    \left\{
    \begin{aligned}
        &(\frac{1}{n^2(x)}\Delta_x +\lambda_*)G(x,y;\lambda_*)=\frac{1}{n^2(x)}\delta(x-y),\quad x,y \in \Omega,\\
        &\nabla_x G(x,y;\lambda_*) (x)\cdot \bm{n}_x=0 ,\quad x\in \partial \Omega ,
    \end{aligned}
    \right.
\end{equation*}
obtained by the limiting absorption principle \cite{joly2016solutions},
\begin{equation*}
\begin{aligned}
G(x,y;\lambda_*)
&= \lim_{\eta \to 0^+}G(x,y;\lambda_* + i \eta)= \frac{1}{2\pi}\lim_{\eta \to 0^+}\int_0^{2\pi}\sum_{n\geq 1}\frac{v_{n}(x;p)\overline{v_{n}(y;p)}}{\lambda_*+i\eta-\mu_{n}(p)}dp.
\end{aligned}
\end{equation*}
For ease of notation, we denote
\begin{equation} \label{eq_K_n_def}
K_{n}(x,y;p,\lambda):=\frac{v_{n}(x;p)\overline{v_{n}(y;p)}}{\lambda-\mu_{n}(p)}.
\end{equation}
The radiation condition of $G(x,y;\lambda_*)$ is established in \cite{joly2016solutions}. We recall the following properties of $G(x,y;\lambda_*)$: first, $G(x,y;\lambda_*)$ admits the following decomposition
as a manifestation of the limiting absorption principle (cf. Remark 8 in \cite{joly2016solutions}):
\begin{equation} \label{eq_G_pv_decomposition}
\begin{aligned}
G(x,&y;\lambda_* )
=\frac{1}{2\pi}\int_{-\pi}^{\pi}\sum_{n\neq n_*,m_*}K_{n}(x,y;p,\lambda)dp \\
&+\Bigg(\frac{1}{2\pi}p.v.\int_{-\pi}^{\pi}K_{n_{*}}(x,y;p,\lambda)dp
+\frac{1}{2\pi}p.v.\int_{-\pi}^{\pi}K_{m_{*}}(x,y;p,\lambda)dp \\
&\qquad -\frac{i}{2}\frac{v_{n_{*}}(x;0)\overline{v_{n_{*}}(y;0)}}{|\mu_{n_{*}}^{\prime}(0)|}
-\frac{i}{2}\frac{v_{m_{*}}(x;0)\overline{v_{m_{*}}(y;0)}}{|\mu_{m_{*}}^{\prime}(0)|}\\
&\qquad -\frac{i}{2}\frac{v_{n_{*}}(x;-q_{*})\overline{v_{n_{*}}(y;-q_{*})}}{|\mu_{n_{*}}^{\prime}(-q_{*})|}
-\frac{i}{2}\frac{v_{m_{*}}(x;q_{*})\overline{v_{m_{*}}(y;q_{*})}}{|\mu_{m_{*}}^{\prime}(q_{*})|} \Bigg),
\end{aligned}
\end{equation}
where $p.v.$ means Cauchy's principal value. The emergence of the principal-value integral above is a consequence of the fact that the dispersion curves $\lambda=\mu_{n_{*}}(p)$ and $\lambda=\mu_{n_{*}}(p)$ intersect with $\lambda=\lambda_*$ linearly at $p=0,\pm q_*$.

Next, $G(x,y;\lambda_*)$ attains the following asymptotic at infinity (cf. Remark 9 in \cite{joly2016solutions}):
\begin{equation} \label{eq_G_right_decomposition}
G(x,y;\lambda_*)=G_0^+(x,y;\lambda_*)-i\cdot\frac{v_{n_{*}}(x;0)\overline{v_{n_{*}}(y;0)}}{|\mu_{n_{*}}^{\prime}(0)|}
-i\cdot\frac{v_{m_{*}}(x;q_*)\overline{v_{m_{*}}(y;q_*)}}{|\mu_{m_{*}}^{\prime}(q_*)|},
\end{equation}
and
\begin{equation} \label{eq_G_left_decomposition}
G(x,y;\lambda_*)=G_0^-(x,y;\lambda_*)-i\cdot\frac{v_{m_{*}}(x;0)\overline{v_{m_{*}}(y;0)}}{|\mu_{m_{*}}^{\prime}(0)|}
-i\cdot\frac{v_{n_{*}}(x;-q_*)\overline{v_{n_{*}}(y;-q_*)}}{|\mu_{n_{*}}^{\prime}(-q_*)|},
\end{equation}
where $G_0^{+}(x,y;\lambda)$ ($G_0^{-}(x,y;\lambda)$) decays exponentially as $x_1\to +\infty$ (for $x_1\to -\infty$) for fixed $y$. 

\medskip

Physically, $v_{n_{*}}(x;0)$ and $v_{m_*}(x;q_{*})$ in \eqref{eq_G_right_decomposition} are the right-propagating modes with frequency $w=\sqrt{\lambda_*}$ in the waveguide. Analogously, $v_{m_{*}}(x;0)$ and $v_{n_*}(x;-q_{*})$ in \eqref{eq_G_left_decomposition} are the left-propagating modes. The following properties hold for these propagating modes. Their proofs are similar to \cite[Lemma 2.2 and 2.3]{qiu2023mathematical}.
\begin{lemma}[Reflection relations between Bloch modes] \label{lemma_equivalence}
The following relations hold
\begin{equation*}
\begin{aligned}
&v_{n_{*}}(x;0)\sim \mathcal{P}v_{m_{*}}(x;0),\quad
v_{n_{*}}(x;0)\sim \overline{v_{m_{*}}(x;0)}, \\
&v_{n_{*}}(x;-q_*)\sim \mathcal{P}v_{m_{*}}(x;q_*),\quad
v_{n_{*}}(x;-q_*)\sim \overline{v_{m_{*}}(x;q_*)}.
\end{aligned}
\end{equation*}
\end{lemma}

\begin{lemma}[Energy flux between Bloch modes] \label{lemma_energy_flux}
For each $u,v\in H^1_{loc}(\Omega)$, we define the sesquilinear form 
$$
\mathfrak{q}(u,v):=\int_{\Gamma}\frac{\partial u}{\partial x_1}\overline{v}dx_2.
$$ 
Then 
\begin{equation} \label{eq_energy_flux_1}
\begin{aligned}
&\mathfrak{q}(v_{n_{*}}(x;p),v_{n_{*}}(x;p))=\frac{i}{2}\mu_{n_{*}}^{\prime}(p),
\quad p=0\text{ or }-q_*, \\
&\mathfrak{q}(v_{m_{*}}(x;p),v_{m_{*}}(x;p))=\frac{i}{2}\mu_{m_{*}}^{\prime}(p),
\quad p=0\text{ or }q_*,
\end{aligned}
\end{equation}
and
\begin{equation*} 
\begin{aligned}
&\mathfrak{q}(v_{n_{*}}(x;0),v_{n_{*}}(x;-q_*))=\mathfrak{q}(v_{n_{*}}(x;0),v_{m_{*}}(x;q_*))=0, \\
&\mathfrak{q}(v_{m_{*}}(x;0),v_{m_{*}}(x;q_*))=\mathfrak{q}(v_{m_{*}}(x;0),v_{n_{*}}(x;-q_*))=0.
\end{aligned}
\end{equation*}

\end{lemma}

As a consequence of the reflection symmetry and time-reversal symmetry of the system, the physical Green function $G(x,y;\lambda_*)$ is symmetric in the following sense (the proof is similar to \cite[Lemma 2.4]{qiu2023mathematical}):
\begin{lemma}
For $x,y\in\Omega$ and $x\neq y$, 
$G(x,y;\lambda_*)=G(y,x;\lambda_*)$.
On the other hand, when $y\in \Gamma$, there holds
\begin{equation*}
    G(x,y;\lambda_*)=(\mathcal{P}G)(x,y;\lambda_*),\quad \forall x\in \Omega.
\end{equation*}
\end{lemma}

Finally, using the physical Green function, we have a representation formula for the solution to the Helmholtz equation in the semi-infinite domain $\Omega^{+}$; the proof is similar to \cite[Proposition 2.5]{qiu2023mathematical}.

\begin{proposition}
Suppose $u\in H^1_{loc}(\Omega^{+})$ satisfies that
\begin{equation*} 
    (\frac{1}{n^2}\Delta+\lambda_*)u(x)=0\; \mbox{in} \; \Omega^{+},
    \nabla u\cdot \bm{n}_x=0\, (x\in \partial \Omega^{+} \backslash \Gamma). 
\end{equation*}
Assume that $u$ satisfies the same radiation condition as the physical Green function in the sense that
\begin{equation*}
u(x)-a\cdot v_{n_{*}}(x;0)-b\cdot v_{m_{*}}(x;q_*)\text{ decays exponentially as $x_1\to +\infty$ for some constants $a$ and $b$}.
\end{equation*}
Then
\begin{equation}  \label{eq-repre-u}
u(x)=2\int_{\Gamma}G(x,y;\lambda_*)\frac{\partial u}{\partial x_1}(0^+,y_2)dy_2,\quad x\in\Omega^{+} .
\end{equation}
\end{proposition}

\section{Asymptotic expansions for the perturbed structure}
\subsection{Asymptotic expansions of Floquet-Bloch eigenvalues and eigenfunctions}
In this section, we derive asymptotic expansions for the Bloch eigenpairs of the perturbed operator $\mathcal{L}_{\epsilon}$. The main result of this section, Theorem \ref{thm_local_gap_open}, shows that: 1) a ``local'' band gap is opened near the Dirac point $(0, \lambda_*)$, and 2) a band inversion occurs between the Bloch eigenspaces of $\mathcal{L}_{\epsilon}$ and $\mathcal{L}_{-\epsilon}$ near the Dirac point $(0, \lambda_*)$; see Remark \ref{rmk_gap_open_band_inversion}. These two phenomena, especially the band inversion, are essential for the subsequent bifurcation of eigenvalues at the Dirac point $(0, \lambda_*)$. In contrast, away from the Dirac point, near the momenta \(\pm q_*\), the perturbation produces only a smooth deformation of the dispersion relation; see Theorem \ref{thm_perturbed_fold_dispersion}.


\begin{theorem} \label{thm_local_gap_open}
Under Assumption \ref{assump_reflection}, \ref{assump_general spectrum}, \ref{assump_dirac_points} and \ref{assump_perturbation}, the following asymptotic expansions hold uniformly for $|\epsilon|\ll 1, |p|\ll 1$:
\begin{equation} \label{eq_perturbed_dirac_dispersion}
\begin{aligned}
&\lambda_{\mathfrak{n}_*,\epsilon}(p)=\lambda_*-\sqrt{\alpha_{n_{*}}^2 p^2+|t_{*}|^2\epsilon^2}\left(1+\mathcal{O}(|p|+|\epsilon|)\right), \\
&\lambda_{\mathfrak{n}_*+1,\epsilon}(p)=\lambda_*+\sqrt{\alpha_{n_{*}}^2 p^2+|t_{*}|^2\epsilon^2}\left(1+\mathcal{O}(|p|+|\epsilon|)\right),
\end{aligned}
\end{equation}
where $t_{*}$ is defined in Assumption \ref{assump_perturbation} and $\alpha_{n_{*}}:=\mu_{n_{*}}'(0)$. 
Moreover, for $\epsilon\neq 0$, the Floquet-Bloch eigenfunctions admit the follows asymptotic expansions in $H^1(Y)$:
\begin{equation} \label{eq_perturbed_dirac_eigenfunction_1}
u_{\mathfrak{n}_*,\epsilon}(x;p)=
\left\{
\begin{aligned}
&\frac{t_{*}\cdot\epsilon}{\alpha_{n_{*}}p+\sqrt{\alpha_{n_{*}}^2 p^2+|t_{*}|^2\epsilon^2}}v_{n_{*}}(x;0)+v_{m_{*}}(x;0)
+\mathcal{O}(|p|+|\epsilon|),\quad p>0, \\
&\frac{t_{*}\cdot\epsilon}{-\alpha_{n_{*}}p+\sqrt{\alpha_{n_{*}}^2 p^2+|t_{*}|^2\epsilon^2}}(\mathcal{P}v_{n_{*}})(x;0)+(\mathcal{P}v_{m_{*}})(x;0)
+\mathcal{O}(|p|+|\epsilon|),\quad p<0,
\end{aligned}
\right.
\end{equation}
and
\begin{equation} \label{eq_perturbed_dirac_eigenfunction_2}
u_{\mathfrak{n}_*+1,\epsilon}(x;p)=
\left\{
\begin{aligned}
&v_{n_{*}}(x;0)-\frac{\overline{t_{*}}\cdot \epsilon}{\alpha_{n_{*}}p+\sqrt{\alpha_{n_{*}}^2 p^2+|t_{*}|^2\epsilon^2}}v_{m_{*}}(x;0)
+\mathcal{O}(|p|+|\epsilon|),\quad p>0, \\
&(\mathcal{P}v_{n_{*}})(x;0)-
\frac{\overline{t_{*}}\cdot \epsilon}{-\alpha_{n_{*}}p+\sqrt{\alpha_{n_{*}}^2 p^2+|t_{*}|^2\epsilon^2}}(\mathcal{P}v_{m_{*}})(x;0)
+\mathcal{O}(|p|+|\epsilon|),\quad p<0,
\end{aligned}
\right.
\end{equation}
where $v_{n_{*}}(x;0)$ ($v_{m_{*}}(x;0)$, resp.) is the right- (left-, resp.) propagating mode at the Dirac point $(p,\lambda)=(0,\lambda_*)$. 
\end{theorem}

\begin{remark} \label{rmk_gap_open_band_inversion}
By Theorem \ref{thm_local_gap_open}, a ``local'' band gap $(\lambda_*-c_0 |t_*\epsilon|,\lambda_*-c_0 |t_*\epsilon|)$ is opened near the Dirac point $(0,\lambda_*)$, for both operators $\mathcal{L}_{\epsilon}$ and $\mathcal{L}_{-\epsilon}$ if $t_*\neq 0$; see Figure \ref{fig_perturbed_band_structure}. Moreover, for $\epsilon>0$, \eqref{eq_perturbed_dirac_eigenfunction_1} and \eqref{eq_perturbed_dirac_eigenfunction_2} indicate that 
\begin{equation*}
\begin{aligned}
&u_{\mathfrak{n}_*,-\epsilon}(x;0)
\approx -e^{i\text{arg}(t_*)}v_{n_{*}}(x;0)+v_{m_{*}}(x;0)
\approx -e^{i\text{arg}(t_*)}u_{\mathfrak{n}_*+1,\epsilon}(x;0), \\
&u_{\mathfrak{n}_*+1,-\epsilon}(x;0)
\approx v_{n_{*}}(x;0)+e^{-i\text{arg}(t_*)}v_{m_{*}}(x;0)
\approx e^{-i\text{arg}(t_*)}u_{\mathfrak{n}_*,\epsilon}(x;0).
\end{aligned}
\end{equation*}
where $\text{arg}\, (t_*)$ denotes the argument of the complex number $t_*$.  Roughly speaking, this means that the eigenspace of $\mathcal{L}_{-\epsilon}$ at the lower vertex $(0,-|t_*|\epsilon)$ (the upper vertex $(0,|t_*|\epsilon)$, resp.) transitions to the eigenspace of $\mathcal{L}_{\epsilon}$ at the upper vertex (the lower vertex, resp.) when one tunes the perturbation parameter $\epsilon$ across the origin. This switch of eigenspaces, or the so-called band inversion, plays a key role in the bifurcation of the Dirac point.
\end{remark}

\begin{figure}
\centering
\begin{tikzpicture}[scale=0.25]
\draw[->] (-11,0)--(11,0);
\draw[->] (0,0)--(0,18);
\node[right] at (11.2,0) {$p$};
\node[below] at (0,-0.2) {$0$};
\draw[thick] (-10,-0.1)--(-10,0.1);
\node[below] at (-10,-0.2) {$-\pi$};
\draw[thick] (10,-0.1)--(10,0.1);
\node[below] at (10,-0.2) {$\pi$};

\draw[thick,blue] plot [smooth] coordinates {(-10,9) (-7,7.8) (-4,1.8) (0,5.8) (4,1.8) (7,7.8) (10,9)};
\draw[thick,red] plot [smooth] coordinates {(-10,16.15) (-7,15.2) (-4,13.2) (0,8.2) (4,13.2) (7,15.2) (10,16.15)};
\node[above] at (9.5,16.15) {$\lambda=\lambda_{\mathfrak{n}_*+1,\epsilon}(p)$};
\node[above] at (9.5,9) {$\lambda=\lambda_{\mathfrak{n}_*,\epsilon}(p)$};
\draw[dashed] (-8,7)--(8,7);
\node[right] at (8.5,7) {$\lambda=\lambda_*$};
\end{tikzpicture}
\caption{Perturbed band structure: $\lambda=\lambda_{\mathfrak{n}_*,\epsilon}(p)$ is plotted in blue while $\lambda=\lambda_{\mathfrak{n}_*+1,\epsilon}(p)$ is plotted in red; they are smooth near $p=0$. Compared with Figure \ref{fig_unperturbed_band_structure}, a ``local'' band gap is opened near $(0,\lambda_*)$ when we apply the perturbation. However, no ``global'' band gap appears, due to the failure of spectral no-fold condition ($\lambda=\lambda_{\mathfrak{n}_*+1,\epsilon}(p)$ always intersects with $\lambda=\lambda_*$ for $|\epsilon|\ll 1$).}
\label{fig_perturbed_band_structure}
\end{figure}
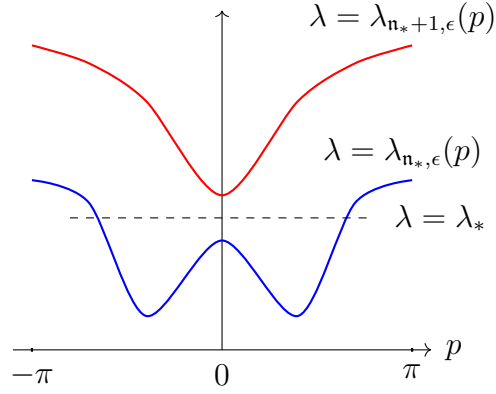

In the next theorem, we present asymptotic expansions for the Floquet eigenvalues and eigenfunctions near $p=\pm q_*$. The proof follows the standard perturbation theory for simple eigenvalue problems of self-adjoint operators (See Chapter VII in \cite{kato2013perturbation}) and hence is omitted here.  

\begin{theorem} \label{thm_perturbed_fold_dispersion}
Under Assumption \ref{assump_general spectrum}, \ref{assump_dirac_points} and \ref{assump_perturbation}, the following asymptotic expansions hold near $p=q_*$:
\begin{equation} \label{eq_perturbed_fold_eigenvalue_eigenfunction}
\begin{aligned}
&\lambda_{\mathfrak{n}_*,\epsilon}(p)=\lambda_*+\mathcal{O}(|p-q_*|+|\epsilon|),\quad
u_{\mathfrak{n}_*,\epsilon}(x;p)=u_{\mathfrak{n}_*}(x;q_{*})+\mathcal{O}(|p-q_*|+|\epsilon|) \,\,\, \mbox{in}\,\,\, H^1(Y).
\end{aligned}
\end{equation}
Moreover,
\begin{equation} \label{eq_perturbed_fold_remainder_2}
\lambda_{\mathfrak{n}_*,\epsilon}^\prime(p)-\lambda_{\mathfrak{n}_*}^{\prime}(q_*)=\mathcal{O}(|p-q_{*}|+|\epsilon|),\quad
\|\left(\partial_p  u_{\mathfrak{n}_*,\epsilon}\right)(x;p)-
\left(\partial_p  u_{\mathfrak{n}_*}\right)(x;q_*)\|_{H^1(Y)}=\mathcal{O}(|p-q_{*}|+|\epsilon|).
\end{equation}
Similarly, near $p=-q_*$:
\begin{equation} \label{eq_perturbed_fold_eigenvalue_eigenfunction_opposite}
\begin{aligned}
&\lambda_{\mathfrak{n}_*,\epsilon}(p)=\lambda_*+\mathcal{O}(|p+q_*|+|\epsilon|),\quad
u_{\mathfrak{n}_*,\epsilon}(x;p)=u_{\mathfrak{n}_*}(x;-q_{*})+\mathcal{O}(|p+q_*|+|\epsilon|) \,\,\, \mbox{in}\,\,\, H^1(Y).
\end{aligned}
\end{equation}
Moreover,
\begin{equation} \label{eq_perturbed_fold_remainder_2_opposite}
\lambda_{\mathfrak{n}_*,\epsilon}^{\prime}(p)-\lambda_{\mathfrak{n}_*}^{\prime}(-q_*)=\mathcal{O}(|p+q_*|+|\epsilon|),\quad
\|\left(\partial_p  u_{\mathfrak{n}_*,\epsilon}\right)(x;p)-
\left(\partial_p  u_{\mathfrak{n}_*}\right)(x;-q_*)\|_{H^1(Y)}=\mathcal{O}(|p+q_*|+|\epsilon|).
\end{equation}
\end{theorem}
\begin{remark} \label{rmk_thm_fold_complex}
By the analytic perturbation theory, the Bloch eigenvalue $\lambda_{\mathfrak{n}_*,\epsilon}(p)$ and eigenfunction $u_{\mathfrak{n}_*,\epsilon}(x;p)$ is analytic in $p$ for $|p-q_*|\ll 1$ and $p\in\mathbf{C}$.
\end{remark}
\begin{remark}
A consequence of Theorem \ref{thm_perturbed_fold_dispersion} is that the dispersion curve $\lambda=\lambda_{\mathfrak{n}_*+1,\epsilon}(p)$ intersects with $\lambda=\lambda_*$ for $|\epsilon|\ll 1$; 
Thus, a ``global'' band gap is not opened when we apply a small perturbation to the system \eqref{eq_waveeq}. See Figure \ref{fig_perturbed_band_structure}. 
\end{remark}

\subsection{Proof of Theorem \ref{thm_local_gap_open}}
We apply perturbation theory to solve the eigenvalue problem of $\mathcal{L}_{\epsilon}(p)$. As a preparation, We write
\begin{equation} \label{eq_sec32_1}
    \mathcal{L}_{\epsilon}(p)=e^{ip x_1}\circ \tilde{\mathcal{L}}_{\epsilon}(p)\circ e^{-ip x_1},
\end{equation}
where
\begin{equation*}
\tilde{\mathcal{L}}_{\epsilon}(p): H_{0,b}^1(\Delta,\Omega)\subset L^2_{0}(\Omega)\to L^{2}_{0}(\Omega),\quad 
u\mapsto -\frac{1}{n_{\epsilon}^2}\Big(\Delta+2ip\cdot\frac{\partial}{\partial x_1}-p^2\Big)u .
\end{equation*}
Since for \(p\in(-\pi,\pi]\) the operators \(\tilde{\mathcal{L}}_{\epsilon}(p)\) and \(\mathcal{L}_{\epsilon}(p)\) are unitarily equivalent, it suffices to solve the eigenvalue problem for \(\tilde{\mathcal{L}}_{\epsilon}(p)\).  Moreover, by construction
\[
\mathrm{Dom}\bigl(\tilde{\mathcal{L}}_{\epsilon}(p)\bigr)
=H^1_{0,b}(\Delta,\Omega),
\]
which is independent of \(p\), making \(\tilde{\mathcal{L}}_{\epsilon}(p)\) particularly well‐suited for perturbation analysis.
For $\epsilon=0$, we have 
\begin{equation} \label{eq_sec32_3}
\tilde{\mathcal{L}}(p)\tilde{v}_{n}(x;p)=\mu_{n} (p)\tilde{v}_{n}(x;p),
\end{equation}
where
\begin{equation} \label{eq_sec32_2}
\tilde{v}_{n}(x;p)=e^{-ip x_1}v_{n}(x;p)\in H_{0,b}^1(\Delta,\Omega).
\end{equation}
We claim the following identities, which is the key ingredients for calculating the perturbed band structure when we transform the perturbation analysis of $\tilde{\mathcal{L}}(p)$ to a finite-dimensional problem (i.e. \eqref{eq_sec32_11}).
\begin{lemma} \label{lemma_derivative_wrt_p}
Let $\tilde{B}(p):=\frac{-2i}{n^2} \frac{\partial}{\partial x_1}+\frac{2p}{n^2}$. Then the following identities hold:
\begin{equation} \label{eq_derivative_wrt_p_1}
\left(\tilde{B}(0)\tilde{v}_{n_{*}}(x;0), \tilde{v}_{n_{*}}(x;0)\right)_{L^2(Y;n(x))}=\alpha_{n_{*}},
\end{equation}
\begin{equation} \label{eq_derivative_wrt_p_2}
\left(\tilde{B}(0)\tilde{v}_{m_{*}}(x;0), \tilde{v}_{m_{*}}(x;0)\right)_{L^2(Y;n(x))}=-\alpha_{n_{*}},
\end{equation}
and
\begin{equation} \label{eq_derivative_wrt_p_3}
\begin{aligned}
\left(\tilde{B}(0)\tilde{v}_{n_{*}}(x;0), \tilde{v}_{m_{*}}(x;0)\right)_{L^2(Y;n(x))}
=
\left(\tilde{B}(0)\tilde{v}_{m_{*}}(x;0), \tilde{v}_{n_{*}}(x;0)\right)_{L^2(Y;n(x))}
=0.
\end{aligned}
\end{equation}
\end{lemma}

\begin{proof}[Proof of Lemma \ref{lemma_derivative_wrt_p}]
We prove only \eqref{eq_derivative_wrt_p_1}; analogous arguments yield \eqref{eq_derivative_wrt_p_2} and \eqref{eq_derivative_wrt_p_3}. By letting $n=n_*$ in \eqref{eq_sec32_3} and taking inner product with $\tilde{v}_{n_*}$, we have
\begin{equation} \label{eq_sec32_4}
\mu_{n_{*}}(p)=\left(\tilde{\mathcal{L}}(p)\tilde{v}_{n_{*}}(x;p),\tilde{v}_{n_{*}}(x;p)\right)_{L^2(Y;n(x))}.
\end{equation}
Differentiating \eqref{eq_sec32_4} with respect to $p$ and evaluating at $p=0$ yields
\begin{equation*}
\begin{aligned}
\mu_{n_{*}}^{\prime}(0)&=\left(\tilde{B}(0)\tilde{v}_{n_{*}}(x;0),\tilde{v}_{n_{*}}(x;0)\right)_{L^2(Y;n(x))}
+\left(\tilde{\mathcal{L}}(0)(\partial_p\tilde{v}_{n_{*}})(x;0),\tilde{v}_{n_{*}}(x;0)\right) _{L^2(Y;n(x))} \\
&\quad +\left(\tilde{\mathcal{L}}(0)\tilde{v}_{n_{*}}(x;0),(\partial_p\tilde{v}_{n_{*}})(x;0)\right)_{L^2(Y;n(x))}.
\end{aligned}
\end{equation*}
Since $\tilde{\mathcal{L}}(0)$ is self-adjoint, we have
\begin{equation*}
\begin{aligned}
\mu_{n_{*}}^{\prime}(0)
=&\left(\tilde{B}(0)\tilde{v}_{n_{*}}(x;0),\tilde{v}_{n_{*}}(x;0)\right)_{L^2(Y;n(x))} \\
&\quad +\lambda_*
\Big(\left(\partial_p\tilde{v}_{n_{*}}(x;0),\tilde{v}_{n_{*}}(x;0)\right)_{L^2(Y;n(x))}
+\left(\tilde{v}_{n_{*}}(x;0),\partial_p\tilde{v}_{n_{*}}(x;0)\right)_{L^2(Y;n(x))}\Big) \\
=&\left(\tilde{B}(0)\tilde{v}_{n_{*}}(x;0),\tilde{v}_{n_{*}}(x;0)\right)_{L^2(Y;n(x))}
+\lambda_*\frac{d}{dp}\Big|_{p=0}\Big(\tilde{v}_{n_{*}}(x;p),\tilde{v}_{n_{*}}(x;p)\Big)_{L^2(Y;n(x))}.
\end{aligned}
\end{equation*}
The normalization condition $\Big(\tilde{v}_{n_{*}}(x;p),\tilde{v}_{n_{*}}(x;p)\Big)_{L^2(Y;n(x))}= 1$ implies that
$$
\alpha_{n_{*}}=\mu_{n_{*}}^{\prime}(0)=\left(\tilde{B}(0)\tilde{v}_{n_{*}}(x;0), \tilde{v}_{n_{*}}(x;0)\right)_{L^2(Y;n(x))}.
$$
This completes the proof of \eqref{eq_derivative_wrt_p_1}.
\end{proof}

Now we are ready to prove Theorem \ref{thm_local_gap_open}.
\begin{proof}[Proof of Theorem \ref{thm_local_gap_open}]
We derive \eqref{eq_perturbed_dirac_dispersion}, \eqref{eq_perturbed_dirac_eigenfunction_1} and \eqref{eq_perturbed_dirac_eigenfunction_2} by solving the following eigenvalue problem for $|\epsilon|,|p|,|\lambda_{\epsilon}-\lambda_*|\ll 1$:
\begin{equation} \label{eq_sec32_6}
\tilde{\mathcal{L}}_{\epsilon}(p)\tilde{u}_{\epsilon}=\lambda_{\epsilon}\tilde{u}_{\epsilon}.
\end{equation}
Note that for each $u\in H^{1}_{0,b}(\Delta,\Omega)$, the following estimate holds in $L^2_{0}(\Omega)$:
\begin{equation*}
\|\tilde{\mathcal{L}}_{\epsilon}(p)u-\tilde{\mathcal{L}}(p)u\|
=\Big\|\frac{n^2_{\epsilon}-n^2}{n^2_{\epsilon}}\tilde{\mathcal{L}}(p)u\Big\|
\lesssim
\|n_{\epsilon}-n\|_{L^{\infty}(\Omega)}\cdot \|\tilde{\mathcal{L}}(p)u\|.
\end{equation*}
Hence, by \cite[Chapter IV, Theorem 2.24]{kato2013perturbation}, $\tilde{\mathcal{L}}_{\epsilon}(p)$ converges to $\tilde{\mathcal{L}}(p)$ in the generalized sense. Consequently, for $|\epsilon|\ll 1$, \eqref{eq_sec32_6} has the same number of solutions near $\lambda_*$ as in the unperturbed case $\epsilon=0$, which is known to be exactly two (namely, the two Bloch eigenvalues $\lambda_{\mathfrak{n}_*}(p)$ and $\lambda_{\mathfrak{n}_*+1}(p)$). 

Next, we explicitly determine these two eigenpairs using a perturbation argument.
For $|p|\ll 1$, we write
\begin{equation} \label{eq_sec32_7}
\lambda_{\epsilon}=\lambda_*+\lambda^{(1)}, \quad
\tilde{u}_{\epsilon}=\tilde{u}^{(0)}_{\epsilon}+\tilde{u}^{(1)}_{\epsilon}
\end{equation}
with
\begin{equation*}
|\lambda^{(1)}|\ll 1 ,\quad
\tilde{u}^{(0)}_{\epsilon}=a\cdot\tilde{v}_{n_*}(x;0)+b\cdot\tilde{v}_{m_*}(x;0)\in \ker (\tilde{\mathcal{L}}(0)-\lambda_*),\quad
\tilde{u}^{(1)}_{\epsilon}\in (\ker (\tilde{\mathcal{L}}(0)-\lambda_*))^{\perp},
\end{equation*}
where $(\ker (\tilde{\mathcal{L}}(0)-\lambda_*))^{\perp}$ denotes the orthogonal complement of $\ker (\tilde{\mathcal{L}}(0)-\lambda_*)$ in $L_{0}^2(\Omega)$. 

Note that the following expansion holds in $\mathcal{B}(H_{0,b}^2(\Omega),L_{0}^2(\Omega))$:
\begin{equation} \label{eq_sec32_5}
\tilde{\mathcal{L}}_{\epsilon}(p)=\tilde{\mathcal{L}}(0)+\epsilon\cdot \tilde{A}(0)+p\cdot \tilde{B}(0)
+\mathcal{O}(p^2+\epsilon^2),
\end{equation}
where $\tilde{A}(p)=-\frac{2}{n}\frac{\partial \left(n_{\epsilon}\right)}{\partial \epsilon}\big|_{\epsilon=0}\cdot \tilde{\mathcal{L}}(p)$ and $\tilde{B}(p)=\frac{-2i}{n^2} \frac{\partial}{\partial x_1}+\frac{2p}{n^2}$. Plugging \eqref{eq_sec32_5} and \eqref{eq_sec32_7} into \eqref{eq_sec32_6} leads to
\begin{equation} \label{eq_sec32_8}
\begin{aligned}
(\tilde{\mathcal{L}}(0)-\lambda_*)\tilde{u}^{(1)}_{\epsilon}
=&(\lambda^{(1)}-\epsilon\cdot \tilde{A}(0)-p\cdot \tilde{B}(0)
+\mathcal{O}(p^2+\epsilon^2))\tilde{u}^{(0)}_{\epsilon} \\
&+(\lambda^{(1)}-\epsilon\cdot \tilde{A}(0)-p\cdot \tilde{B}(0)
+\mathcal{O}(p^2+\epsilon^2))\tilde{u}^{(1)}_{\epsilon}.
\end{aligned}
\end{equation}

We then solve \eqref{eq_sec32_8} following a Lyapunov-Schmidt reduction argument. To do so, we introduce the orthogonal projection $Q_{\perp}:L_{0}^2(\Omega)\to (\ker (\tilde{\mathcal{L}}(0)-\lambda_*))^{\perp}$. Applying $Q_{\perp}$ to \eqref{eq_sec32_8} gives
\begin{equation*}
\begin{aligned}
(\tilde{\mathcal{L}}(0)-\lambda_*)\tilde{u}^{(1)}_{\epsilon}
=&Q_{\perp}(\lambda^{(1)}-\epsilon\cdot \tilde{A}(0)-p\cdot \tilde{B}(0)
+\mathcal{O}(p^2+\epsilon^2))\tilde{u}^{(0)}_{\epsilon} \\
&+Q_{\perp}(\lambda^{(1)}-\epsilon\cdot \tilde{A}(0)-p\cdot \tilde{B}(0)
+\mathcal{O}(p^2+\epsilon^2))\tilde{u}^{(1)}_{\epsilon}.
\end{aligned}
\end{equation*} 
The above equation can be rewritten as
\begin{equation*}
\begin{aligned}
(I-T)\tilde{u}^{(1)}_{\epsilon}
=T\tilde{u}^{(0)}_{\epsilon},
\end{aligned}
\end{equation*}
where
\begin{equation*}
T=T(\epsilon,p,\lambda^{(1)}):=(\tilde{\mathcal{L}}(0)-\lambda_*)^{-1}Q_{\perp}(\lambda^{(1)}-\epsilon\cdot \tilde{A}(0)-p\cdot \tilde{B}(0)
+\mathcal{O}(p^2+\epsilon^2)).
\end{equation*}
For $\epsilon,p,\lambda^{(1)}$ sufficiently small, $(I-T)^{-1}\in \mathcal{B}(Q_{\perp}H_{0,b}^2(\Omega))$. It holds that
\begin{equation} \label{eq_sec32_10}
\tilde{u}^{(1)}_{\epsilon}
=(I-T)^{-1}T\tilde{u}^{(0)}_{\epsilon}
=a\cdot (I-T)^{-1}T\tilde{v}_{n_*}(x;0)+b\cdot (I-T)^{-1}T\tilde{v}_{m_*}(x;0).
\end{equation}
Note that the map $(\epsilon,p,\lambda^{(1)})\mapsto (I-T)^{-1}T\tilde{v}_{n}(x;0)$ ($n=n_*,m_*$) is smooth from a neighborhood of $(0,0,0)$ to $H_{0,b}^2(\Omega)$ with the following estimate
\begin{equation*}
\|(I-T)^{-1}T\tilde{v}_{n}(x;0)\|\lesssim |\epsilon|+|p|+|\lambda^{(1)}|.
\end{equation*}
By taking $L^2(Y;n(x))-$inner product with $\tilde{v}_{n_*}(x;0)$ and $\tilde{v}_{m_*}(x;0)$ respectively on both sides of \eqref{eq_sec32_8}, we obtain a two-dimensional linear problem
\begin{equation} \label{eq_sec32_11}
\mathcal{M}(\epsilon,p,\lambda^{(1)})
\begin{pmatrix}
a \\ b
\end{pmatrix}
=0,
\end{equation}
where the components of $\mathcal{M}(\epsilon,p,\lambda^{(1)})$ are given by
\begin{equation*}
\begin{aligned}
&M_{11}=\lambda^{(1)}-\epsilon \big(\tilde{A}(0)\tilde{v}_{n_*}(x;0),\tilde{v}_{n_*}(x;0)\big)_{L^2(Y;n(x))}
-p \big(\tilde{B}(0)\tilde{v}_{n_*}(x;0),\tilde{v}_{n_*}(x;0)\big)_{L^2(Y;n(x))} \\
&\quad\qquad +\mathcal{O}((\lambda^{(1)})^2+p^2+\epsilon^2),\\
&M_{22}=\lambda^{(1)}-\epsilon \big(\tilde{A}(0)\tilde{v}_{m_*}(x;0),\tilde{v}_{m_*}(x;0)\big)_{L^2(Y;n(x))}
-p \big(\tilde{B}(0)\tilde{v}_{m_*}(x;0),\tilde{v}_{m_*}(x;0)\big)_{L^2(Y;n(x))} \\
&\quad\qquad +\mathcal{O}((\lambda^{(1)})^2+p^2+\epsilon^2),\\
&M_{12}=-\epsilon \big(\tilde{A}(0)\tilde{v}_{m_*}(x;0),\tilde{v}_{n_*}(x;0)\big)_{L^2(Y;n(x))}
-p \big(\tilde{B}(0)\tilde{v}_{m_*}(x;0),\tilde{v}_{n_*}(x;0)\big)_{L^2(Y;n(x))} \\
&\quad\qquad +\mathcal{O}((\lambda^{(1)})^2+p^2+\epsilon^2),\\
&M_{21}=-\epsilon \big(\tilde{A}(0)\tilde{v}_{n_*}(x;0),\tilde{v}_{m_*}(x;0)\big)_{L^2(Y;n(x))}
-p \big(\tilde{B}(0)\tilde{v}_{n_*}(x;0),\tilde{v}_{m_*}(x;0)\big)_{L^2(Y;n(x))} \\
&\quad\qquad +\mathcal{O}((\lambda^{(1)})^2+p^2+\epsilon^2).\\
\end{aligned}
\end{equation*}

\noindent We now simplify the expression of $\mathcal{M}(\epsilon,p,\lambda^{(1)})$. 
By \eqref{eq_sec32_1} and \eqref{eq_sec32_2},  
$$
\big(\mathcal{L}_{\epsilon}(0)v_{i}(x;0),v_{j}(x;0)\big)_{L^2(Y;n(x))}=\big(\tilde{\mathcal{L}}_{\epsilon}(0)\tilde{v}_{i}(x;0),\tilde{v}_{j}(x;0)\big)_{L^2(Y;n(x))}\quad \text{for $i,j\in \{n_*,m_*\}$}. 
$$
Taking derivative with respect to $\epsilon$ yields $$\big(A(0)v_{i}(x;0),v_{j}(x;0)\big)_{L^2(Y;n(x))}=\big(\tilde{A}(0)\tilde{v}_{i}(x;0),\tilde{v}_{j}(x;0)\big)_{L^2(Y;n(x))},$$
where $A(0)$ and $\tilde{A}(0)$ are introduced in Assumption \ref{assump_perturbation} and equation \eqref{eq_sec32_5}, respectively. Then Assumption \ref{assump_perturbation} and Lemma \ref{lemma_derivative_wrt_p} yield:
\begin{equation*}
\mathcal{M}(\epsilon,p,\lambda^{(1)})=
\begin{pmatrix}
\lambda^{(1)}-\alpha_{n_*}p & t_*\epsilon \\
\overline{t_*}\epsilon & \lambda^{(1)}+\alpha_{n_*}p
\end{pmatrix}
+\mathcal{O}((\lambda^{(1)})^2+p^2+\epsilon^2).
\end{equation*}
Thus, for each $p$, $\lambda_{\epsilon}=\lambda_*+\lambda^{(1)}$ solves the eigenvalue problem \eqref{eq_sec32_6} if and only if $\lambda^{(1)}$ solves 
\begin{equation} \label{eq_sec32_13}
F(\epsilon,p,\lambda^{(1)}):=\det \mathcal{M}(\epsilon,p,\lambda^{(1)})
=(\lambda^{(1)})^2-\alpha_{n_*}^2 p^2-|t_*|^2\epsilon^2+\rho(\epsilon,p,\lambda^{(1)})=0,
\end{equation}
where
\begin{equation} \label{eq_sec32_14}
\rho(\epsilon,p,\lambda^{(1)})=\mathcal{O}(|\epsilon|^3+|p|^3+|\lambda^{(1)}|^3).
\end{equation}

We then solve $\lambda^{(1)}=\lambda^{(1)}(\epsilon,p)$ from \eqref{eq_sec32_13} for each $p$ and $\epsilon$. First, note that $\pm \sqrt{\alpha_{n_*}^2 p^2+|t_*|^2\epsilon^2}$ give two branches of solutions if we drop the remainder $\rho$ from \eqref{eq_sec32_13}. Thus, we seek a solution to \eqref{eq_sec32_13} in the following form
\begin{equation} \label{eq_sec32_15}
\lambda^{(1)}(\epsilon,p)=x\cdot \sqrt{\alpha_{n_*}^2 p^2+|t_*|^2\epsilon^2}
\end{equation}
with $|x|$ close to $1$. By substituting \eqref{eq_sec32_15} into \eqref{eq_sec32_13}, we obtain the following equation of $x$ (with $p$ and $\epsilon$ being regarded as parameters): 
\begin{equation} \label{eq_sec32_16}
\begin{aligned}
H(x;\epsilon,p)&:=\frac{1}{\alpha_{n_*}^2 p^2+|t_*|^2\epsilon^2}F(\epsilon,p,x\cdot \sqrt{\alpha_{n_*}^2 p^2+|t_*|^2\epsilon^2})
=x^2-1+\rho_1(x;\epsilon,p)=0,
\end{aligned}
\end{equation}
where $\rho_1(x;\epsilon,p):=\frac{\rho(\epsilon,p,x\cdot \sqrt{\alpha_{n_*}^2 p^2+|t_*|^2\epsilon^2})}{\alpha_{n_*}^2 p^2+|t_*|^2\epsilon^2}$. 
Now we consider the solution to \eqref{eq_sec32_16} with $|x-1|\ll 1$. Note that, by \eqref{eq_sec32_14}, the following estimate holds uniformly in $x$ when $|x-1|\ll 1$
\begin{equation*}
\rho_1(x;\epsilon,p)=\mathcal{O}(|\epsilon|+|p|).
\end{equation*}
We conclude that there exists a unique solution $x_s(\epsilon,p)$ to \eqref{eq_sec32_16} with $x_s(\epsilon,p)=1+\mathcal{O}(|\epsilon|+|p|)$ for $|\epsilon|,|p|\ll 1$. It follows from \eqref{eq_sec32_15} that there exists a unique solution $\lambda^{(1)}_{+}(\epsilon,p)$ to the equation \eqref{eq_sec32_13} near $\sqrt{\alpha_{n_*}^2 p^2+|t_*|^2\epsilon^2}$. Moreover,
\begin{equation*}
\begin{aligned}
\lambda^{(1)}_{+}(\epsilon,p)=x_s(\epsilon,p)\cdot \sqrt{\alpha_{n_*}^2 p^2+|t_*|^2\epsilon^2}
=\sqrt{\alpha_{n_*}^2 p^2+|t_*|^2\epsilon^2}\cdot \left(1+\mathcal{O}(|\epsilon|+|p|)\right).
\end{aligned}
\end{equation*}
Similarly, the other solution to \eqref{eq_sec32_13} is given by
\begin{equation*}
\begin{aligned}
\lambda^{(1)}_{-}(\epsilon,p)=-\sqrt{\alpha_{n_*}^2 p^2+|t_*|^2\epsilon^2}\cdot\left(1+\mathcal{O}(|\epsilon|+|p|)\right).
\end{aligned}
\end{equation*}
Note that $\lambda_{\mathfrak{n}_*+1,\epsilon}(p)=\lambda_*+\lambda^{(1)}_{+}(\epsilon,p)$ and $\lambda_{\mathfrak{n}_*,\epsilon}(p)=\lambda_*+\lambda^{(1)}_{-}(\epsilon,p)$. This proves \eqref{eq_perturbed_dirac_dispersion}. 

\medskip

Finally, by substituting $\lambda^{(1)}=\lambda^{(1)}_{\pm}(p)$ inside \eqref{eq_sec32_11}, we obtain the following two solutions $(a_{+},b_{+})^T$ and $(a_{-},b_{-})^T$ for $p>0$:
\begin{equation*}
\begin{pmatrix}
a_+ \\ b_+
\end{pmatrix}
=
\begin{pmatrix}
1 \\ -\frac{\overline{t_{*}}\cdot \epsilon}{\alpha_{n_{*}}p+\sqrt{\alpha_{n_{*}}^2 p^2+|t_{*}|^2\epsilon^2}}+\mathcal{O}(|\epsilon|+|p|)
\end{pmatrix}
,\quad
\begin{pmatrix}
a_- \\ b_- 
\end{pmatrix}
=
\begin{pmatrix}
\frac{t_*\cdot \epsilon}{\alpha_{n_{*}}p+\sqrt{\alpha_{n_{*}}^2 p^2+|t_{*}|^2\epsilon^2}}+\mathcal{O}(|\epsilon|+|p|)
\\ 1
\end{pmatrix}.
\end{equation*}
Then \eqref{eq_sec32_7} and \eqref{eq_sec32_10} show that the two eigenfunctions $\tilde{u}_{\mathfrak{n}_*,\epsilon}(x;p)$ and $\tilde{u}_{\mathfrak{n}_*+1,\epsilon}(x;p)$ (corresponding to $\lambda_{\mathfrak{n}_*,\epsilon}(p)$ and $\lambda_{\mathfrak{n}_*+1,\epsilon}(p)$, respectively) that solve \eqref{eq_sec32_6} for $p>0$ are given by
\begin{equation*}
\begin{aligned}
&\tilde{u}_{\mathfrak{n}_*,\epsilon}(x;p)=\frac{t_{*}\cdot\epsilon}{\alpha_{n_{*}}p+\sqrt{\alpha_{n_{*}}^2 p^2+|t_{*}|^2\epsilon^2}}\tilde{v}_{n_{*}}(x;0)+\tilde{v}_{m_{*}}(x;0)
+\mathcal{O}(|\epsilon|+|p|), \\
&\tilde{u}_{\mathfrak{n}_*+1,\epsilon}(x;p)=\tilde{v}_{n_{*}}(x;0)-\frac{\overline{t_{*}}\cdot \epsilon}{\alpha_{n_{*}}p+\sqrt{\alpha_{n_{*}}^2 p^2+|t_{*}|^2\epsilon^2}}\tilde{v}_{m_{*}}(x;0)
+\mathcal{O}(|\epsilon|+|p|).
\end{aligned}
\end{equation*}
Thus, by the relation $u_{n,\epsilon}(x;p)=e^{ipx_1}\tilde{u}_{n,\epsilon}(x;p)$ ($n=\mathfrak{n}_*,\mathfrak{n}_*+1$) and \eqref{eq_sec32_2}, we conclude the proof of \eqref{eq_perturbed_dirac_eigenfunction_1} and \eqref{eq_perturbed_dirac_eigenfunction_2} for $p>0$. For $p<0$, the reflectional symmetry of the system implies the equivalence $u_{n,\epsilon}(x;p)\sim (\mathcal{P}u_{n,\epsilon})(x;-p)$ (similar to Lemma \ref{lemma_equivalence}); this proves \eqref{eq_perturbed_dirac_eigenfunction_1} and \eqref{eq_perturbed_dirac_eigenfunction_2}.
\end{proof}

\section{Bifurcation of the Dirac point under perturbation}
In this section, we prove Theorem \ref{thm_main result} by constructing a solution $u(x;\lambda^{\star})$ to \eqref{eq_joint_system} with $\lambda^{\star}\in \mathcal{I}_{\epsilon}$. The construction is based on the following single-layer potential operator for the perturbed structure:
\begin{equation} \label{eq_formal_single_layer}
    \varphi\in\tilde{H}^{-\frac{1}{2}}(\Gamma) \mapsto u(x;\lambda)=\int_{\Gamma} G_{\epsilon}(x,y;\lambda) \varphi(y)d\sigma(y),
\end{equation}
where $G_{\epsilon}(x,y;\lambda)$ is the physical Green function of the perturbed operator $\mathcal{L}_{\epsilon}$. $G_{\epsilon}(x,y;\lambda)$ admits similar properties as those introduced in Section 2.2 for the unperturbed operator $\mathcal{L}$. Specifically, 
\begin{equation} \label{eq_G_eps_def}
\begin{aligned}
G_{\epsilon}(x,y;\lambda)
&= \lim_{\eta \to 0^+}G_{\epsilon}(x,y;\lambda_* + i \eta)= \frac{1}{2\pi}\lim_{\eta \to 0^+}\int_{-\pi}^{\pi}\sum_{n\geq 1}K_{n,\epsilon}(x,y;p,\lambda+i\eta)dp
\end{aligned}
\end{equation}
with
\begin{equation} \label{eq_K_n_eps_def}
K_{n,\epsilon}(x,y;p,\lambda):=\frac{u_{n,\epsilon}(x;p)\overline{u_{n,\epsilon}(y;\overline{p})}}{\lambda-\lambda_{n,\epsilon}(p)}. 
\end{equation}
Here $(\lambda_{n,\epsilon}(p),u_{n,\epsilon}(x;p))$ denotes the Floquet-Bloch eigenpair of the perturbed operator $\mathcal{L}_{\epsilon}(p)$. We note that 
$$G_{\epsilon}(x,y;\lambda)=G_{\epsilon}^{I}(x,y;\lambda)+ G_{\epsilon}^{II}(x,y;\lambda),$$
where
\begin{equation} \label{eq_G_eps_I_def}
\begin{aligned}
G_{\epsilon}^{I}(x,y;\lambda)
:= \frac{1}{2\pi}\lim_{\eta \to 0^+}\int_{-\pi}^{\pi}\sum_{n\neq \mathfrak{n}_{*}}K_{n,\epsilon}(x,y;p,\lambda+i\eta)dp
=\int_{-\pi}^{\pi}\sum_{n\neq \mathfrak{n}_{*}}K_{n,\epsilon}(x,y;p,\lambda)dp
\end{aligned}
\end{equation}
contains the contribution from energy bands that do not intersect the Dirac energy level $\lambda_*$, and 
\begin{equation} \label{eq_G_eps_II_def}
\begin{aligned}
G_{\epsilon}^{II}(x,y;\lambda)
:= \frac{1}{2\pi}\lim_{\eta \to 0^+}\int_{-\pi}^{\pi}K_{\mathfrak{n}_{*},\epsilon}(x,y;p,\lambda+i\eta)dp
\end{aligned}
\end{equation}
contains contribution from the single band $n=\mathfrak{n}_{*}$ that intersects the Dirac energy level. A detailed expression of $G_{\epsilon}^{II}$ is shown later in the proof of Proposition \ref{prop_G_eps_cont_real}; see \eqref{eq_G_eps_I_II_detail}. Observe that $G_{\epsilon}^{I}(x,y;\lambda)$ is analytic near $\lambda=\lambda_*$,
while $G_{\epsilon}^{II}(x,y;\lambda)$ is not. Thus the single-layer potential operator (\ref{eq_formal_single_layer}) is not analytic near $\lambda_*$. To study the bifurcation of resonance or embedded eigenvalue from the Dirac point, we extend $G_{\epsilon}(x,y;\lambda)$ (more precisely, $G_{\epsilon}^{II}(x,y;\lambda)$) analytically using an appropriate contour integral; see Section 4.1. The continued operator is denoted as $\tilde{\mathbb{G}}_{\epsilon}(\lambda)$ and defined explicitly in \eqref{eq_G_eps_continued}. The detailed properties of $\tilde{\mathbb{G}}_{\epsilon}(\lambda)$ are summarized in Proposition \ref{prop_G_eps_cont_analyticity} and \ref{prop_G_eps_cont_real}. In Section 4.2, we use $\tilde{\mathbb{G}}_{\epsilon}(\lambda)$ to construct a solution $u(x;\lambda)$ to \eqref{eq_joint_system} with $\lambda\in\mathcal{I}_{\epsilon}$; see \eqref{eq_u_left_right_decomposition}. In particular, we prove that $u(x;\lambda)$ indeed solves \eqref{eq_joint_system} if and only if the boundary integral equation \eqref{eq_interface_match} has a solution; see Proposition \ref{prop_charac_to_resonance}. In Sections 4.3 and 4.4, we use the Gohberg-Sigal theory to solve the characteristic value problem \eqref{eq_interface_match} and conclude the proof of Theorem \ref{thm_main result}.

\subsection{Analytic continuation of the single-layer potential operator}
We first present a lemma characterizing the analytic domain of the $n_*$-th Floquet-Bloch eigenpair of the perturbed structure, and the location of pinching singularities (see Section 1.2 for its definition). These results are fundamental for constructing the analytic continuation of the single-layer potential operator. In the sequel, we denote the integral operator associated with the kernel $K_{n,\epsilon}$ by $\mathbb{K}_{n,\epsilon}$:
\begin{equation} \label{eq_K_n_eps_operator_def}
\big(\mathbb{K}_{n,\epsilon}(p,\lambda)\varphi\big)(x)
:=\langle\varphi(\cdot),K_{n,\epsilon}(x,\cdot;p,\lambda)  \rangle = \frac{u_{n,\epsilon}(x;p)\langle \varphi(\cdot),\overline{u_{n,\epsilon}(\cdot ;\overline{p})} \rangle}{\lambda-\lambda_{n,\epsilon}(p)}. 
\end{equation}
We denote the projection to the $n$-th Bloch mode of the perturbed structure and its orthogonal complement by $\mathbb{P}_{n,\epsilon}$ and $\mathbb{Q}_{n,\epsilon}$, respectively: 
\[
\mathbb{P}_{n,\epsilon}(p):=(\cdot,u_{n,\epsilon}(x;\overline{p}))_{L^2(Y;n_{\epsilon}(x))}u_{n,\epsilon}(x;p), \quad \mathbb{Q}_{n,\epsilon}(p):=1-\mathbb{P}_{n,\epsilon}(p).
\]

\begin{lemma} \label{lem_analytic_domain}
Define the following complex domain (see Figure \ref{fig_analytic_domain}):
\begin{equation*}
D_{\epsilon,\nu}:=B(q_*,|\epsilon|^{\frac{1}{3}})\cup B(-q_*,|\epsilon|^{\frac{1}{3}})\cup \{p\in\mathbf{C}:-(1+\nu)\pi<\text{Re}(p)< (1+\nu)\pi, |\Im(p)|<\nu \,(\nu>0)\}.
\end{equation*}
There exists $\epsilon_0>0$ such that for any $\epsilon$ with $0<|\epsilon|<\epsilon_0$, there exists $\nu=\nu(\epsilon)>0$ such that the following statements hold: 
\begin{enumerate}
    \item $p\mapsto\lambda_{\mathfrak{n}_*,\epsilon}(p)$ and $p\mapsto u_{\mathfrak{n}_*,\epsilon}(x;p)$ are analytic in $D_{\epsilon,\nu}$;
    \item for each $\lambda\in \mathcal{I}_{\epsilon}$, the equation $\lambda_{\mathfrak{n}_*,\epsilon}(p)=\lambda$ has exactly two roots $q_{+,\epsilon}(\lambda)$ and $q_{-,\epsilon}(\lambda)=-q_{+}(\lambda,\epsilon)$ in $D_{\epsilon,\nu}$ with $|q_{\pm,\epsilon}(\lambda)\mp q_*|=\mathcal{O}(|\epsilon|)$. Moreover,
    \begin{equation*}
    \begin{aligned}
    &\Im(q_{+,\epsilon}(\lambda))>0,\, \Im(q_{-,\epsilon}(\lambda))<0 \,\text{ when }\, \Im(\lambda)>0,\\
    &\Im(q_{+,\epsilon}(\lambda))<0,\, \Im(q_{-,\epsilon}(\lambda))>0 \,\text{ when }\, \Im(\lambda)<0;
    \end{aligned}
    \end{equation*}
    \item  $p\mapsto \mathbb{P}_{\mathfrak{n}_*,\epsilon}(p)$ is analytic in $D_{\epsilon,\nu}$. Moreover, $p\mapsto (\mathcal{L}_{\epsilon}(p)\mathbb{Q}_{\mathfrak{n}_*,\epsilon}(p)-\lambda)^{-1}\in \mathcal{B}((H^1_{p,b}(\Omega))^{*},H^1_{p,b}(\Omega))$ is analytic in $\{p\in\mathbf{C}:-(1+\nu)\pi<\text{Re}(p)< (1+\nu)\pi, |\Im(p)|<\nu\}$ for any $\lambda \in \mathcal{I}_{\epsilon}$.
\end{enumerate}
\end{lemma}
\begin{proof}
See Appendix A.
\end{proof}

\begin{figure}
\centering
\begin{tikzpicture}[scale=0.4]
\tikzset{->-/.style=
{decoration={markings,mark=at position #1 with 
{\arrow{latex}}},postaction={decorate}}};

\draw[->] (-11,0)--(11,0);
\draw[->] (0,-5)--(0,5);
\node[right] at (11.2,0) {$\text{Re}(p)$};
\node[above] at (0,5.2) {$\Im(p)$};
\node[right] at (0,-0.2) {$0$};
\draw[thick] (-10,-0.1)--(-10,0.1);
\node[below] at (-10,-0.29) {$-\pi$};
\draw[thick] (10,-0.1)--(10,0.1);
\node[below] at (10,-0.29) {$\pi$};
\draw[thick] (-6.25,-0.1)--(-6.25,0.1);
\node[below] at (-6.25,-0.2) {$-q_*$};
\draw[thick] (6.25,-0.1)--(6.25,0.1);
\node[below] at (6.25,-0.2) {$q_*$};

\draw[thick] plot [smooth] coordinates {(-11,4) (-4,4) (-1,1) (1,1) (4,4) (11,4)};
\draw[thick] plot [smooth] coordinates {(-11,-4) (-4,-4) (-1,-1) (1,-1) (4,-4) (11,-4)};
\draw[dashed] (-11,-4)--(-11,4);
\draw[dashed] (11,-4)--(11,4);
\node[right,scale=1] at (-10,3) {$\tilde{D}_{\epsilon}$};

\draw[thick,red,domain=10:170] plot ({-6.25+2*cos(\x)},{2*sin(\x)});
\draw[thick,red,domain=190:350] plot ({-6.25+2*cos(\x)},{2*sin(\x)});
\draw[thick,red,domain=10:170] plot ({6.25+2*cos(\x)},{2*sin(\x)});
\draw[thick,red,domain=190:350] plot ({6.25+2*cos(\x)},{2*sin(\x)});
\draw[thick,red] (-11,0.347)--(-8.219,0.347);
\draw[thick,red] (-11,-0.347)--(-8.219,-0.347);
\draw[thick,red] (-4.28,0.347)--(4.28,0.347);
\draw[thick,red] (-4.28,-0.347)--(4.28,-0.347);
\draw[thick,red] (11,0.347)--(8.219,0.347);
\draw[thick,red] (11,-0.347)--(8.219,-0.347);
\draw[fill=red,opacity=0.2] (-6.25,0) circle(2);
\draw[fill=red,opacity=0.2] (6.25,0) circle(2);
\draw[fill=blue,opacity=0.1] (-11,-0.347) rectangle(11,0.347);
\node[above,scale=1] at (-6.25,0.7) {$D_{\epsilon,\nu}$};
\end{tikzpicture}
\caption{The domain $D_{\epsilon,\nu}$ consists of two disks of radius $\epsilon^{\frac{1}{3}}$ (filled with red) and a strip of width $\nu=\nu(\epsilon)$ (filled with blue). A dumbbell-shaped domain $\tilde{D}_{\epsilon}$ that contains $D_{\epsilon,\nu}$ is also drawn here. $\tilde{D}_{\epsilon}$ is used in the proof of Lemma \ref{lem_analytic_domain} to construct $D_{\epsilon,\nu}$.}
\label{fig_analytic_domain}
\end{figure}
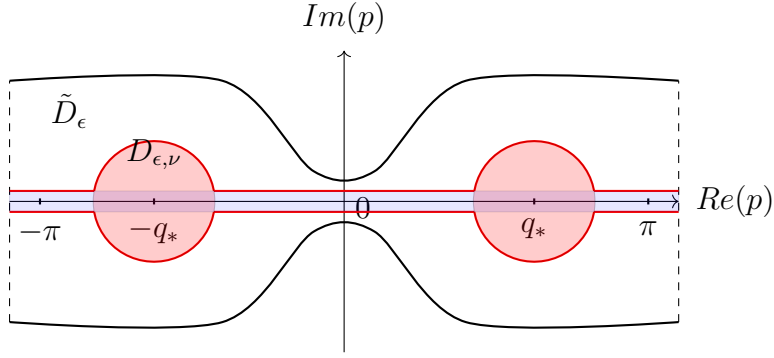

Based on Lemma \ref{lem_analytic_domain}, we construct the analytic continuation of the single-layer operator \eqref{eq_formal_single_layer} for $\lambda\in\mathcal{I}_{\epsilon}$. Let $C_{\epsilon}$ (See Figure \ref{fig_integral_contour}(a)) be the complex contour defined as
\begin{equation*}
\begin{aligned}
C_{\epsilon}:=&
[-\pi,-q_*-|\epsilon|^{\frac{1}{3}}]\cup[-q_*+|\epsilon|^{\frac{1}{3}},q_*-|\epsilon|^{\frac{1}{3}}]\cup [q_*+|\epsilon|^{\frac{1}{3}},\pi] \\
&\cup \{-q_*+|\epsilon|^{\frac{1}{3}}e^{i\theta}:\pi\geq \theta\geq 0\}
\cup \{q_*+|\epsilon|^{\frac{1}{3}}e^{i\theta}:\pi\leq \theta\leq 2\pi\}.
\end{aligned}
\end{equation*}
Then we define the analytic continuation of the layer potential operator \eqref{eq_formal_single_layer}, namely $\tilde{\mathbb{G}}_{\epsilon}(\lambda)\in\mathcal{B}(\tilde{H}^{-\frac{1}{2}}(\Gamma),H^{\frac{1}{2}}(\Gamma))$, as follows:
\begin{equation} \label{eq_G_eps_continued}
\begin{aligned}
\tilde{\mathbb{G}}_{\epsilon}(\lambda)\varphi &:=
\frac{1}{2\pi}\int_{C_{\epsilon}}\big(\mathbb{K}_{\mathfrak{n}_{*},\epsilon}(p,\lambda)\varphi\big)dp 
+
\frac{1}{2\pi}\int_{-\pi}^{\pi}
\sum_{n\neq \mathfrak{n_*}}\big(\mathbb{K}_{n,\epsilon}(p,\lambda)\varphi\big) dp.
\end{aligned}
\end{equation}
Note that \eqref{eq_G_eps_continued} differs from \eqref{eq_formal_single_layer} only by the integral contour in the kernel $G_{\epsilon}^{II}(x,y;\lambda)$. By Lemma \ref{lem_analytic_domain}, we see $\lambda\neq \lambda_{\mathfrak{n}_*,\epsilon}(p)$ for $\lambda\in \mathcal{I}_{\epsilon}$ and $p\in C_{\epsilon}$. This ensures that the new contour $C_{\epsilon}$ in \eqref{eq_G_eps_continued} does not pass any pinching singularities. As a consequence, the denominator in $K_{\mathfrak{n}_{*},\epsilon}(x,y;p,\lambda)$ is non-singular for all $\lambda\in \mathcal{I}_{\epsilon}$ and $p\in C_{\epsilon}$, resulting in the analyticity of \eqref{eq_G_eps_continued}:
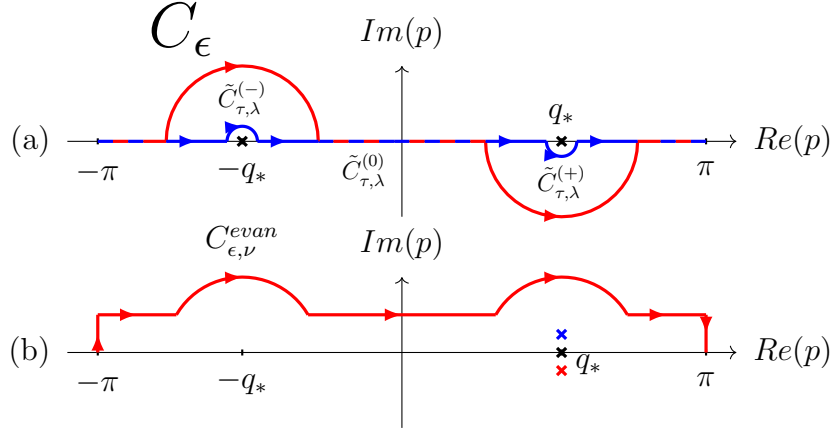
\begin{figure}
\centering
\begin{tikzpicture}[scale=0.4]
\tikzset{->-/.style={decoration={markings,mark=at position #1 with 
{\arrow{latex}}},postaction={decorate}}};

\draw[->] (-11,0)--(11,0);
\draw[->] (0,-2.5)--(0,2.5);
\node[right] at (11.2,0) {$\text{Re}(p)$};
\node[above] at (0,2.7) {$\Im(p)$};
\draw[thick] (-10,-0.1)--(-10,0.1);
\node[below] at (-10,-0.29) {$-\pi$};
\draw[thick] (10,-0.1)--(10,0.1);
\node[below] at (10,-0.29) {$\pi$};
\draw[thick] (-5.25,-0.1)--(-5.25,0.1);
\node[below] at (-5.25,-0.2) {$-q_*$};
\draw[thick] (5.25,-0.1)--(5.25,0.1);
\node[above] at (5.25,0.2) {$q_*$};
\draw[very thick] (5.1,0.15)--(5.4,-0.15);
\draw[very thick] (5.1,-0.15)--(5.4,0.15);
\draw[very thick] (-5.4,0.15)--(-5.1,-0.15);
\draw[very thick] (-5.4,-0.15)--(-5.1,0.15);

\draw[very thick,blue] (-10,0)--(-9.5,0);
\draw[very thick,red] (-9.5,0)--(-9,0);
\draw[very thick,blue] (-9,0)--(-8.5,0);
\draw[very thick,red] (-8.5,0)--(-7.75,0);
\draw[very thick,blue] (-2.75,0)--(-2.25,0);
\draw[very thick,red] (-2.25,0)--(-1.75,0);
\draw[very thick,blue] (-1.75,0)--(-1.25,0);
\draw[very thick,red] (-1.25,0)--(-0.75,0);
\draw[very thick,blue] (-0.75,0)--(-0.25,0);
\draw[very thick,blue] (-0.25,0)--(0.25,0);
\draw[very thick,red] (0.25,0)--(0.75,0);
\draw[very thick,blue] (0.75,0)--(1.25,0);
\draw[very thick,red] (1.25,0)--(1.75,0);
\draw[very thick,blue] (1.75,0)--(2.25,0);
\draw[very thick,red] (2.25,0)--(2.75,0);
\draw[very thick,blue] (10,0)--(9.5,0);
\draw[very thick,red] (9.5,0)--(9,0);
\draw[very thick,blue] (9,0)--(8.5,0);
\draw[very thick,red] (8.5,0)--(7.75,0);
\draw[->-=0.5,very thick,red,domain=180:0] plot ({-5.25+2.5*cos(\x)},{2.5*sin(\x)});
\draw[->-=0.5,very thick,red,domain=180:360] plot ({5.25+2.5*cos(\x)},{2.5*sin(\x)});
\node[above,scale=2] at (-7.25,2.1) {$C_{\epsilon}$};

\draw[->-=0.5,very thick,blue] (-7.75,0)--(-5.75,0);
\draw[->-=0.5,very thick,blue] (-4.75,0)--(-2.75,0);
\draw[->-=0.5,very thick,blue] (2.75,0)--(4.75,0);
\draw[->-=0.5,very thick,blue] (5.75,0)--(7.75,0);
\draw[->-=0.5,very thick,blue,domain=180:0] plot ({-5.25+0.5*cos(\x)},{0.5*sin(\x)});
\draw[->-=0.5,very thick,blue,domain=180:360] plot ({5.25+0.5*cos(\x)},{0.5*sin(\x)});
\node[above,scale=0.8] at (-5.25,0.5) {$\tilde{C}^{(-)}_{\tau,\lambda}$};
\node[below,scale=0.8] at (5.25,-0.5) {$\tilde{C}^{(+)}_{\tau,\lambda}$};
\node[below,scale=0.8] at (-1.25,-0.1) {$\tilde{C}^{(0)}_{\tau,\lambda}$};

\draw[->] (-11,-7)--(11,-7);
\draw[->] (0,-9.5)--(0,-4.5);
\node[right] at (11.2,-7) {$\text{Re}(p)$};
\node[above] at (0,-4.3) {$\Im(p)$};
\draw[thick] (-10,-7.1)--(-10,-6.9);
\node[below] at (-10,-7.29) {$-\pi$};
\draw[thick] (10,-7.1)--(10,-6.9);
\node[below] at (10,-7.29) {$\pi$};
\draw[thick] (-5.25,-7.1)--(-5.25,-6.9);
\node[below] at (-5.25,-7.2) {$-q_*$};
\draw[thick] (5.25,-7.1)--(5.25,-6.9);
\node[right] at (5.35,-7.4) {$q_*$};
\draw[very thick] (5.1,-6.85)--(5.4,-7.15);
\draw[very thick] (5.1,-7.15)--(5.4,-6.85);
\draw[blue,very thick] (5.1,-6.25)--(5.4,-6.55);
\draw[blue,very thick] (5.1,-6.55)--(5.4,-6.25);
\draw[red,very thick] (5.1,-7.45)--(5.4,-7.75);
\draw[red,very thick] (5.1,-7.75)--(5.4,-7.45);

\draw[->-=0.5,very thick,red] (-10,-7)--(-10,-5.75);
\draw[->-=0.5,very thick,red] (-10,-5.75)--(-7.415,-5.75);
\draw[->-=0.5,very thick,red] (-3.085,-5.75)--(3.085,-5.75);
\draw[->-=0.5,very thick,red] (7.415,-5.75)--(10,-5.75);
\draw[->-=0.5,very thick,red] (10,-5.75)--(10,-7);
\draw[->-=0.5,very thick,red,domain=150:30] plot ({-5.25+2.5*cos(\x)},{-7+2.5*sin(\x)});
\draw[->-=0.5,very thick,red,domain=150:30] plot ({5.25+2.5*cos(\x)},{-7+2.5*sin(\x)});
\node[above,scale=1] at (-5.25,-4.3) {$C^{evan}_{\epsilon,\nu}$};

\node[left,scale=1] at (-11.2,0) {(a)};
\node[left,scale=1] at (-11.2,-7) {(b)};
\end{tikzpicture}
\caption{Several integral contours: (a) the contour $C_{\epsilon}$ is drawn in red lines while $\tilde{C}_{\tau,\lambda}$ is drawn in blue. The equation $\lambda_{\mathfrak{n}_*,\epsilon}(p)=\lambda$ has two roots for $
\lambda\in \mathcal{I}_{\epsilon}\cap \mathbf{R}$, which are marked by the black crosses in the Figure. (b) the contour $C^{evan}_{\epsilon,\nu}$ is drawn in red lines. In the figure, the red, black and blue crosses denote the root of $\lambda_{\mathfrak{n}_*,\epsilon}(p)=\lambda$ ($\lambda\in \mathcal{I}_{\epsilon}$, $p$ near $q_*$) for $\Im(\lambda)<0$, $\Im(\lambda)=0$ and $\Im(\lambda)>0$, respectively.}
\label{fig_integral_contour}
\end{figure}

\begin{proposition} \label{prop_G_eps_cont_analyticity}
$\lambda\mapsto\tilde{\mathbb{G}}_{\epsilon}(\lambda)\in\mathcal{B}(\tilde{H}^{-\frac{1}{2}}(\Gamma),H^{\frac{1}{2}}(\Gamma))$ is an operator-valued analytical function in $\mathcal{I}_{\epsilon}$.
\end{proposition}
The operator $\tilde{\mathbb{G}}_{\epsilon}(\lambda)$ indeed continues the single-layer potential operator \eqref{eq_formal_single_layer} as they coincide on the real line:
\begin{proposition} \label{prop_G_eps_cont_real}
For $\lambda\in\mathcal{I}_{\epsilon}\cap \mathbf{R}$, we have $\tilde{\mathbb{G}}_{\epsilon}(\lambda)\varphi
=\int_{\Gamma}G_{\epsilon}(x,y;\lambda)\varphi(y)dy_2$. Moreover, they both admit the following expansion: 
\begin{equation} \label{eq_G_eps_cont_real}
\begin{aligned}
\tilde{\mathbb{G}}_{\epsilon}(\lambda)\varphi
&=\int_{\Gamma}G_{\epsilon}(x,y;\lambda)\varphi(y)dy_2 \\
&=\frac{1}{2\pi}\int_{-\pi}^{\pi}\sum_{n\neq \mathfrak{n_*}}\big(\mathbb{K}_{n,\epsilon}(p,\lambda)\varphi\big)dp
+\frac{1}{2\pi}p.v.\int_{-\pi}^{\pi}\big(\mathbb{K}_{\mathfrak{n}_{*},\epsilon}(p,\lambda)\varphi\big)dp \\
&\quad -
\Bigg(\frac{i\langle \varphi(\cdot),\overline{u_{\mathfrak{n}_*,\epsilon}(\cdot ;q_{+,\epsilon}(\lambda))} \rangle}{2|\lambda_{\mathfrak{n}_*,\epsilon}^{\prime}(q_{+,\epsilon}(\lambda))|}u_{\mathfrak{n}_*,\epsilon}(x ;q_{+,\epsilon}(\lambda))
+
\frac{i\langle \varphi(\cdot),\overline{u_{\mathfrak{n}_*,\epsilon}(\cdot ;q_{-,\epsilon}(\lambda))} \rangle}{2|\lambda_{\mathfrak{n}_*,\epsilon}^{\prime}(q_{-,\epsilon}(\lambda))|}u_{\mathfrak{n}_*,\epsilon}(x ;q_{-,\epsilon}(\lambda))
\Bigg),
\end{aligned}
\end{equation}
where $G_{\epsilon}(x,y;\lambda)$ is the Green function defined in \eqref{eq_G_eps_def}.
\end{proposition}
\begin{proof}
We first note that for $\lambda\in\mathcal{I}_{\epsilon}\cap\mathbf{R}$, the second equality in \eqref{eq_G_eps_cont_real} is already proved in \cite[Theorem 6]{joly2016solutions}:
\begin{equation} \label{eq_G_eps_I_II_detail}
\begin{aligned}
&\int_{\Gamma}G_{\epsilon}(x,y;\lambda)\varphi(y)dy_2
=\frac{1}{2\pi}\int_{-\pi}^{\pi}\sum_{n\neq \mathfrak{n_*}}\big(\mathbb{K}_{n,\epsilon}(p,\lambda)\varphi\big)dp
+\frac{1}{2\pi}p.v.\int_{-\pi}^{\pi}\big(\mathbb{K}_{\mathfrak{n}_{*},\epsilon}(p,\lambda)\varphi\big)dp \\
&\quad -
\Bigg(\frac{i\langle \varphi(\cdot),\overline{u_{\mathfrak{n}_*,\epsilon}(\cdot ;q_{+,\epsilon}(\lambda))} \rangle}{2|\lambda_{\mathfrak{n}_*,\epsilon}^{\prime}(q_{+,\epsilon}(\lambda))|}u_{\mathfrak{n}_*,\epsilon}(x ;q_{+,\epsilon}(\lambda))
+
\frac{i\langle \varphi(\cdot),\overline{u_{\mathfrak{n}_*,\epsilon}(\cdot ;q_{-,\epsilon}(\lambda))} \rangle}{2|\lambda_{\mathfrak{n}_*,\epsilon}^{\prime}(q_{-,\epsilon}(\lambda))|}u_{\mathfrak{n}_*,\epsilon}(x ;q_{-,\epsilon}(\lambda))
\Bigg) .
\end{aligned}
\end{equation}
Hence, comparing \eqref{eq_G_eps_continued} with \eqref{eq_G_eps_cont_real}, it's sufficient to prove the following identity for $\lambda\in\mathcal{I}_{\epsilon}\cap\mathbf{R}$:
\begin{equation} \label{eq_sec4_3}
\begin{aligned}
&\frac{1}{2\pi}\int_{C_{\epsilon}}
\big(\mathbb{K}_{\mathfrak{n}_{*},\epsilon}(p,\lambda)\varphi\big)dp 
=\frac{1}{2\pi}p.v.\int_{-\pi}^{\pi}\big(\mathbb{K}_{\mathfrak{n}_{*},\epsilon}(p,\lambda)\varphi\big)dp \\
& -
\Bigg(\frac{i\langle \varphi(\cdot),\overline{u_{\mathfrak{n}_*,\epsilon}(\cdot ;q_{+,\epsilon}(\lambda))} \rangle}{2|\lambda_{\mathfrak{n}_*,\epsilon}^{\prime}(q_{+,\epsilon}(\lambda))|}u_{\mathfrak{n}_*,\epsilon}(x ;q_{+,\epsilon}(\lambda))
+
\frac{i\langle \varphi(\cdot),\overline{u_{\mathfrak{n}_*,\epsilon}(\cdot ;q_{-,\epsilon}(\lambda))} \rangle}{2|\lambda_{\mathfrak{n}_*,\epsilon}^{\prime}(q_{-,\epsilon}(\lambda))|}u_{\mathfrak{n}_*,\epsilon}(x ;q_{-,\epsilon}(\lambda))
\Bigg).
\end{aligned}
\end{equation}
This follows from a standard analysis of complex integral as follows. We first introduce the following auxiliary contour (See Figure \ref{fig_integral_contour}(a))
\begin{equation*}
\begin{aligned}
\tilde{C}_{\tau,\lambda}:=&\big([-\pi,q_{-,\epsilon}(\lambda)-\tau]\cup[q_{-,\epsilon}(\lambda)+\tau,q_{+,\epsilon}(\lambda)-\tau]\cup [q_{+,\epsilon}(\lambda)+\tau,\pi] \big)\\
&\cup \{q_{-,\epsilon}(\lambda)+\tau e^{i\theta}:\pi\geq \theta\geq 0\}
\cup \{q_{+,\epsilon}(\lambda)+\tau e^{i\theta}:\pi\leq \theta\leq 2\pi\} \\
&:=\tilde{C}_{\tau,\lambda}^{(0)}\cup \tilde{C}_{\tau,\lambda}^{(-)}
\cup \tilde{C}_{\tau,\lambda}^{(+)} \quad \text{with}\, \tau\in(0,|\epsilon|^{\frac{1}{3}}).
\end{aligned}
\end{equation*}
Since $q_{\pm,\epsilon}(\lambda)$ are real for $\lambda\in\mathbf{R}$, the denominator of $\mathbb{K}_{\mathfrak{n}_{*},\epsilon}(p,\lambda)\varphi
=\frac{u_{\mathfrak{n}_*,\epsilon}(x;p)\langle \varphi(\cdot),\overline{u_{\mathfrak{n}_*,\epsilon}(\cdot ;\overline{p})} \rangle}{\lambda-\lambda_{\mathfrak{n}_*,\epsilon}(p)}$ is non-singular, and hence $p\mapsto \mathbb{K}_{\mathfrak{n}_{*},\epsilon}(p,\lambda)\varphi$ is analytic in the closed region bounded by $C_{\epsilon}$ and $\tilde{C}_{\tau,\lambda}$. The Cauchy theorem indicates that the integral is unchanged when we deform the contour within this region
\begin{equation*}
\begin{aligned}
\int_{C_{\epsilon}}
\big(\mathbb{K}_{\mathfrak{n}_{*},\epsilon}(p,\lambda)\varphi\big)dp 
=
\int_{\tilde{C}_{\tau,\lambda}}
\big(\mathbb{K}_{\mathfrak{n}_{*},\epsilon}(p,\lambda)\varphi\big)dp =\int_{\tilde{C}_{\tau,\lambda}^{(0)}\cup \tilde{C}_{\tau,\lambda}^{(-)}
\cup \tilde{C}_{\tau,\lambda}^{(+)}}
\big(\mathbb{K}_{\mathfrak{n}_{*},\epsilon}(p,\lambda)\varphi\big)dp.
\end{aligned}
\end{equation*}
Thus, to obtain \eqref{eq_sec4_3}, it's sufficient to prove
\begin{equation} \label{eq_sec4_4}
\begin{aligned}
&\lim_{\tau\to 0}\int_{\tilde{C}_{\tau,\lambda}^{(0)}\cup \tilde{C}_{\tau,\lambda}^{(-)}
\cup \tilde{C}_{\tau,\lambda}^{(+)}}
\big(\mathbb{K}_{\mathfrak{n}_{*},\epsilon}(p,\lambda)\varphi\big) dp
=p.v.\int_{-\pi}^{\pi}\big(\mathbb{K}_{\mathfrak{n}_{*},\epsilon}(p,\lambda)\varphi\big) dp \\
& -
\Bigg(i\pi\frac{\langle \varphi(\cdot),\overline{u_{\mathfrak{n}_*,\epsilon}(\cdot ;q_{+,\epsilon}(\lambda))} \rangle}{|\lambda_{\mathfrak{n}_*,\epsilon}^{\prime}(q_{+,\epsilon}(\lambda))|}u_{\mathfrak{n}_*,\epsilon}(x ;q_{+,\epsilon}(\lambda))
+i\pi
\frac{\langle \varphi(\cdot),\overline{u_{\mathfrak{n}_*,\epsilon}(\cdot ;q_{-,\epsilon}(\lambda))} \rangle}{|\lambda_{\mathfrak{n}_*,\epsilon}^{\prime}(q_{-,\epsilon}(\lambda))|}u_{\mathfrak{n}_*,\epsilon}(x ;q_{-,\epsilon}(\lambda))
\Bigg).
\end{aligned}
\end{equation}
Note that the definition of Cauchy principal value integral implies that
\begin{equation} \label{eq_sec4_5}
\lim_{\tau\to 0}\int_{\tilde{C}_{\tau,\lambda}^{(0)}}
\big(\mathbb{K}_{\mathfrak{n}_{*},\epsilon}(p,\lambda)\varphi\big) dp
=p.v.\int_{-\pi}^{\pi}\big(\mathbb{K}_{\mathfrak{n}_{*},\epsilon}(p,\lambda)\varphi\big) dp.
\end{equation}
On the other hand, since $\lambda^{\prime}_{\mathfrak{n}_*,\epsilon}(q_{+,\epsilon}(\lambda))\neq 0$ by Theorem \ref{thm_perturbed_fold_dispersion}, $q_{+,\epsilon}(\lambda)$ is a simple pole of the map $p\mapsto \mathbb{K}_{\mathfrak{n}_{*},\epsilon}(p,\lambda)\varphi=\frac{u_{\mathfrak{n}_*,\epsilon}(x;p)\langle \varphi(\cdot),\overline{u_{\mathfrak{n}_*,\epsilon}(\cdot ;\overline{p})} \rangle}{\lambda-\lambda_{\mathfrak{n}_*,\epsilon}(p)}$ (marked as the black cross in Figure \ref{fig_integral_contour}). Therefore, the residue formula gives
\begin{equation*}
\lim_{\tau\to 0}\int_{\tilde{C}_{\tau,\lambda}^{(+)}}
\big(\mathbb{K}_{\mathfrak{n}_{*},\epsilon}(p,\lambda)\varphi\big) dp
=-i\pi\frac{\langle \varphi(\cdot),\overline{u_{\mathfrak{n}_*,\epsilon}(\cdot ;q_{+,\epsilon}(\lambda))} \rangle}{|\lambda_{\mathfrak{n}_*,\epsilon}^{\prime}(q_{+,\epsilon}(\lambda))|}u_{\mathfrak{n}_*,\epsilon}(x ;q_{+,\epsilon}(\lambda)).
\end{equation*}
Similarly,
\begin{equation} \label{eq_sec4_7}
\lim_{\tau\to 0}\int_{\tilde{C}_{\tau,\lambda}^{(-)}}
\big(\mathbb{K}_{\mathfrak{n}_{*},\epsilon}(p,\lambda)\varphi\big) dp 
=-i\pi\frac{\langle \varphi(\cdot),\overline{u_{\mathfrak{n}_*,\epsilon}(\cdot ;q_{-,\epsilon}(\lambda))} \rangle}{|\lambda_{\mathfrak{n}_*,\epsilon}^{\prime}(q_{-,\epsilon}(\lambda))|}u_{\mathfrak{n}_*,\epsilon}(x ;q_{-,\epsilon}(\lambda)).
\end{equation}
Combining \eqref{eq_sec4_5}-\eqref{eq_sec4_7}, we obtain \eqref{eq_sec4_4} and complete the proof. 
\end{proof}

\begin{proposition} \label{prop_G_eps_cont_jump+radiation}
For $\lambda\in \mathcal{I}_{\epsilon}$ and $\varphi\in \tilde{H}^{-\frac{1}{2}}(\Gamma)$, we define $u(x;\lambda)=(\tilde{\mathbb{G}}_{\epsilon}(\lambda)\varphi)(x)$ ($x\in\Omega$). Then 
\begin{equation} \label{eq_G_eps_cont_helmholtz}
(\mathcal{L}_{\epsilon}-\lambda)u(x)=0,\quad x\in\Omega^{+}\cup \Omega^{-},
\end{equation}
and $u(x;\lambda)$ is reflection symmetric with a nonzero jump of its Neumann trace across the interface {
\begin{equation} \label{eq_G_eps_cont_even}
    u(x;\lambda)=u(\mathcal{P}x;\lambda),
\end{equation}
\begin{equation} \label{eq_G_eps_cont_jump}
\Big(\frac{\partial u}{\partial x_1}\Big)\Big|_{\Gamma_+}:=\lim_{x_1\to 0^+}\frac{\partial u}{\partial x_1}=\frac{\varphi}{2},\quad
\Big(\frac{\partial u}{\partial x_1}\Big)\Big|_{\Gamma_-}:=\lim_{x_1\to 0^-}\frac{\partial u}{\partial x_1}=-\frac{\varphi}{2}.
\end{equation}
Moreover, $u(x;\lambda)$ admits the following asymptotics at infinity
\begin{equation} \label{eq_G_eps_cont_asymptotic}
\lim_{x_1\to \pm\infty}\Big|u(x;\lambda)
-\frac{-i\langle \varphi(\cdot),\overline{u_{\mathfrak{n}_*,\epsilon}(\cdot ;\overline{q_{\pm,\epsilon}(\lambda)})} \rangle}{\lambda^{\prime}_{\mathfrak{n}_*,\epsilon}(q_{\pm,\epsilon}(\lambda))}u_{\mathfrak{n}_*,\epsilon}(x ;q_{\pm,\epsilon}(\lambda))\Big|=0,
\end{equation}
where the convergence rate is exponential.}
\end{proposition}

\begin{proof}
We first prove \eqref{eq_G_eps_cont_helmholtz}. Note that the map $\lambda\mapsto (\mathcal{L}_{\epsilon}-\lambda)(\tilde{\mathbb{G}}_{\epsilon}(\lambda)\varphi)$ is analytical in $\mathcal{I}_{\epsilon}$ for each fixed $\varphi \in \tilde{H}^{-\frac{1}{2}}(\Gamma)$. It is sufficient to prove
\eqref{eq_G_eps_cont_helmholtz} for all $\lambda \in \mathcal{I}_{\epsilon}\cap \mathbf{R}$. Indeed, for $\lambda \in \mathcal{I}_{\epsilon}\cap \mathbf{R}$, Proposition \ref{prop_G_eps_cont_real} shows that
\begin{equation*}
(\mathcal{L}_{\epsilon}-\lambda)(\tilde{\mathbb{G}}_{\epsilon}(\lambda)\varphi)(x)
=\int_{\Gamma}(\mathcal{L}_{\epsilon}-\lambda)G_{\epsilon}(x,y;\lambda)\varphi(y)dy_2
=\frac{1}{n^2_{\epsilon}(x)}\varphi(x_2)\delta(x_1),
\end{equation*}
where $\delta(x_1)$ denotes the Dirac-delta function. Therefore, \eqref{eq_G_eps_cont_helmholtz} follows. 

Second, \eqref{eq_G_eps_cont_even} and \eqref{eq_G_eps_cont_jump} follow the same strategy. Note that the case for $\lambda \in \mathcal{I}_{\epsilon}\cap \mathbf{R}$ can be proved using the same approach as that of Lemmas 2.4 and 2.5 in \cite{qiu2023mathematical}, respectively. By analytical extension, we obtain \eqref{eq_G_eps_cont_even} and \eqref{eq_G_eps_cont_jump} for $\lambda \in \mathcal{I}_{\epsilon}$.

Finally, we prove \eqref{eq_G_eps_cont_asymptotic} for the case $x_1\to\infty$; the proof for $x_1\to -\infty$ is similar.  We introduce the following auxiliary contour (See Figure \ref{fig_integral_contour}(b))
\begin{equation*}
\begin{aligned}
C_{\epsilon,\nu}^{evan}:=&
\big([-\pi,-q_{*}-|\epsilon|^{\frac{1}{3}}\cos(\theta_{\epsilon}) ]+i\nu\big)
\cup \big([-q_{*}+|\epsilon|^{\frac{1}{3}}\cos(\theta_{\epsilon}),q_{*}-|\epsilon|^{\frac{1}{3}}\cos(\theta_{\epsilon}) ]+i\nu \big) \\
&\cup \big([q_{*}+|\epsilon|^{\frac{1}{3}}\cos(\theta_{\epsilon}),\pi]+i\nu \big) \\
&\cup \{-q_*+|\epsilon|^{\frac{1}{3}} e^{i\theta}:\pi-\theta_{\epsilon}\geq \theta\geq \theta_{\epsilon}\}
\cup \{q_*+|\epsilon|^{\frac{1}{3}} e^{i\theta}:\pi-\theta_{\epsilon}\geq \theta\geq \theta_{\epsilon}\} \\
&\cup \{-\pi+it:0\leq t\leq \nu\}\cup \{\pi+it:\nu\geq t\geq 0\},
\end{aligned}
\end{equation*}
where $\theta_{\epsilon}:=\sin^{-1}(\frac{\nu}{|\epsilon|^{\frac{1}{3}}})$. By following the proof of \cite[Theorem 7]{joly2016solutions}, we can show that the following function decays exponentially as $x_1 \to +\infty$:
\begin{equation} \label{eq_sec4_9}
w^{evan,+}(x):=\frac{1}{2\pi}\int_{C_{\epsilon,\nu}^{evan}}
\big(\mathbb{K}_{\mathfrak{n}_{*},\epsilon}(p,\lambda)\varphi\big)(x)dp
+
\frac{1}{2\pi}\int_{0}^{2\pi}
\sum_{n\neq \mathfrak{n_*}}\big(\mathbb{K}_{n,\epsilon}(p,\lambda)\varphi\big)(x)dp.
\end{equation}
On the other hand, when we shift the integral of $\mathbb{K}_{\mathfrak{n}_{*},\epsilon}(p,\lambda)\varphi$ from $C_{\epsilon}$ to $C_{\epsilon,\nu}^{evan}$, it passes through a simple pole of the integrand, which lies at $p=q_{+,\epsilon}(\lambda)$ (marked as the cross in Figure \ref{fig_integral_contour}(b)). The residue formula yields that
\begin{equation}  \label{eq_sec4_10}
\begin{aligned}
\frac{1}{2\pi}\int_{C_{\epsilon}}
\big(\mathbb{K}_{\mathfrak{n}_{*},\epsilon}(p,\lambda)\varphi\big)dp
=\frac{1}{2\pi}\int_{C_{\epsilon,\nu}^{evan}}
\big(\mathbb{K}_{\mathfrak{n}_{*},\epsilon}(p,\lambda)\varphi\big)dp
-\frac{i\langle \varphi(\cdot),\overline{u_{\mathfrak{n}_*,\epsilon}(\cdot ;\overline{q_{+,\epsilon}(\lambda)})} \rangle}{\lambda^{\prime}_{\mathfrak{n}_*,\epsilon}(q_{+,\epsilon}(\lambda))}u_{\mathfrak{n}_*,\epsilon}(x ;q_{+,\epsilon}(\lambda)).
\end{aligned}
\end{equation}
By \eqref{eq_sec4_9}, \eqref{eq_sec4_10} and \eqref{eq_G_eps_continued}, we have
\begin{equation*}
u(x;\lambda)=w^{evan,+}(x)-\frac{i\langle \varphi(\cdot),\overline{u_{\mathfrak{n}_*,\epsilon}(\cdot ;\overline{q_{+,\epsilon}(\lambda)})} \rangle}{\lambda^{\prime}_{\mathfrak{n}_*,\epsilon}(q_{+,\epsilon}(\lambda))}u_{\mathfrak{n}_*,\epsilon}(x ;q_{+,\epsilon}(\lambda)).
\end{equation*}
This proves \eqref{eq_G_eps_cont_asymptotic} by noting that $w^{evan,+}(x)$ decays exponentially as $x_1\to +\infty$.
\end{proof}

\begin{proposition}\label{corollary_G_eps_decay_blow}
Let $u(x;\lambda)$ be defined as in Proposition \ref{prop_G_eps_cont_jump+radiation}. If $\Im(\lambda)>0$, then $u(x;\lambda)$ decays exponentially as $x_1\to \pm\infty$. If $\Im(\lambda)\leq 0$, then $u(x;\lambda)$ decays exponentially as $x_1\to \pm\infty$ if and only if
\begin{equation} \label{eq_criteria_interface_mode}
    \langle \varphi(\cdot),\overline{u_{\mathfrak{n}_*,\epsilon}(\cdot ;\overline{q_{\pm,\epsilon}(\lambda)})} \rangle=0.
\end{equation}
\end{proposition} 
\begin{proof}
We only prove that if $\Im(\lambda)<0$, then
$u(x;\lambda)$ decays exponentially as $x_1\to +\infty$ if and only if $\langle \varphi(\cdot),\overline{u_{\mathfrak{n}_*,\epsilon}(\cdot ;\overline{q_{+,\epsilon}(\lambda)})} \rangle = 0$. The other cases can be treated similarly. By Proposition \ref{prop_G_eps_cont_jump+radiation}, the asymptotic limit of $u(x;\lambda)$ as $x_1\to +\infty$ is given by
\begin{equation*}
-i\frac{\langle \varphi(\cdot),\overline{u_{\mathfrak{n}_*,\epsilon}(\cdot ;\overline{q_{+,\epsilon}(\lambda)})} \rangle}{\lambda^{\prime}_{\mathfrak{n}_*,\epsilon}(q_{+,\epsilon}(\lambda))}u_{\mathfrak{n}_*,\epsilon}(x ;q_{+,\epsilon}(\lambda)).
\end{equation*}
Note that $u_{\mathfrak{n}_*,\epsilon}(x ;q_{+}(\lambda))$ satisfies the following quasi-periodic boundary condition
\begin{equation} \label{eq_sec4_14}
u_{\mathfrak{n}_*,\epsilon}(x+\bm{e}_1 ;q_{+,\epsilon}(\lambda))=e^{iq_{+,\epsilon}(\lambda)}u_{\mathfrak{n}_*,\epsilon}(x ;q_{+,\epsilon}(\lambda)).
\end{equation}
Since $q_{+,\epsilon}(\lambda)$ has a negative imaginary part for $\Im(\lambda)<0$ by Lemma \ref{lem_analytic_domain}, \eqref{eq_sec4_14} implies that $u_{\mathfrak{n}_*,\epsilon}(x ;q_{+,\epsilon}(\lambda))$ blows up exponentially as $x_1\to +\infty$. Consequently, $u(x;\lambda)$ decays exponentially as $x_1\to +\infty$ if and only if $\langle \varphi(\cdot),\overline{u_{\mathfrak{n}_*,\epsilon}(\cdot ;\overline{q_{+,\epsilon}(\lambda)})} \rangle = 0$.
\end{proof}

\subsection{Boundary-integral formulation for the mode bifurcated from the Dirac point}

We construct a solution to \eqref{eq_joint_system} in the form
\begin{equation} \label{eq_u_left_right_decomposition}
    u(x_1,x_2;\lambda):=\left\{
    \begin{aligned}
        &(\tilde{\mathbb{G}}_{\epsilon}(\lambda)\varphi)(x), \,\, x_1>0, \\
        &-(\tilde{\mathbb{G}}_{-\epsilon}(\lambda)\varphi)(x), \,\, x_1<0, 
    \end{aligned}
    \right.
\end{equation}
for some $\varphi\in\tilde{H}^{-\frac{1}{2}}(\Gamma)$. By Proposition \ref{prop_G_eps_cont_jump+radiation},
\begin{equation*}
(\mathcal{L}_{\epsilon}-\lambda)u=0 \quad (x_1>0)\quad
\text{and} \quad
(\mathcal{L}_{-\epsilon}-\lambda)u=0 \quad (x_1<0).
\end{equation*}
Note that we require the Dirichlet and Neumann data of $u$ to be continuous across the interface:
\begin{equation} \label{eq_cont_dirichlet}
    u(0^-,x_2;\lambda)=u(0^+,x_2;\lambda),
\end{equation}
\begin{equation} \label{eq_cont_neumann}
    \frac{\partial u}{\partial x_1}(0^-,x_2;\lambda)=\frac{\partial u}{\partial x_1}(0^+,x_2;\lambda).
\end{equation}
By \eqref{eq_G_eps_cont_jump}, condition \eqref{eq_cont_neumann} holds for the constructed solution \eqref{eq_u_left_right_decomposition}. By substituting \eqref{eq_u_left_right_decomposition} into \eqref{eq_cont_dirichlet}, we obtain the following boundary integral equation
\begin{equation} \label{eq_interface_match}
\mathbb{G}_{\epsilon}^{\Gamma}(\lambda)\varphi
:=\big(\tilde{\mathbb{G}}_{\epsilon}(\lambda)\varphi+\tilde{\mathbb{G}}_{-\epsilon}(\lambda)\varphi\big)\big|_{\Gamma}=0.
\end{equation}
Thus, the function $u(x;\lambda)$ defined in \eqref{eq_u_left_right_decomposition} is a $H^1_{loc}(\Omega)-$solution of \eqref{eq_joint_system} only if $\lambda\in\mathcal{I}_{\epsilon}$ is a characteristic value of $\mathbb{G}_{\epsilon}^{\Gamma}(\lambda)\in \mathcal{B}(\tilde{H}^{-\frac{1}{2}}(\Gamma), H^{\frac{1}{2}}(\Gamma))$. Conversely, when \eqref{eq_interface_match} holds, $u(x;\lambda)$ in \eqref{eq_u_left_right_decomposition} solves \eqref{eq_joint_system}. Hence, the eigenvalue problem \eqref{eq_joint_system} is equivalent to the characteristic value problem \eqref{eq_interface_match}. In addition, the following holds:
\begin{proposition} \label{prop_charac_to_resonance}
Let $\lambda\in \mathcal{I}_{\epsilon}$ be a characteristic value of $\mathbb{G}_{\epsilon}^{\Gamma}(\lambda)$ and $\varphi\in \tilde{H}^{-\frac{1}{2}}(\Gamma)$ the associated eigenvector for \eqref{eq_interface_match}. Then the function $u(x;\lambda)$ defined by \eqref{eq_u_left_right_decomposition} satisfies the properties in Theorem \ref{thm_main result}.
\end{proposition}
\begin{proof}
We first prove that $\Im(\lambda)\leq 0$. 
Assume the contrary that $\Im(\lambda)>0$. Then Proposition \ref{corollary_G_eps_decay_blow} indicates that $u(x;\lambda)$ exponentially decays as $|x_1|\to \infty$. An integral by parts shows that
\begin{equation} \label{eq_lem_ctr_1}
    \int_{\Omega^+}|\nabla u|^2 dx-\lambda \int_{\Omega^+}n^2_{\epsilon}|u|^2 dx=-\int_{\Gamma}\frac{\partial u}{\partial x_1}\overline{u} dx_2.
\end{equation}
Similarly,
\begin{equation} \label{eq_lem_ctr_2}
    \int_{\Omega^-}|\nabla u|^2 dx-\lambda \int_{\Omega^-}n^2_{-\epsilon}|u|^2 dx=\int_{\Gamma}\frac{\partial u}{\partial x_1}\overline{u} dx_2.
\end{equation}
Since the interface conditions \eqref{eq_cont_dirichlet} and \eqref{eq_cont_neumann} are satisfied when \eqref{eq_interface_match} holds, we can sum \eqref{eq_lem_ctr_1} and \eqref{eq_lem_ctr_2} to get
\begin{equation} \label{eq_lem_ctr_3}
    \int_{\Omega}|\nabla u|^2 dx-\lambda 
    (\int_{\Omega^+}n^2_{\epsilon}|u|^2 dx+\int_{\Omega^-}n^2_{-\epsilon}|u|^2 dx)=0.
\end{equation}
Since $\Im(\lambda)>0$ and $n_{\pm\epsilon}(x)>0$ for any $x\in\Omega$, \eqref{eq_lem_ctr_3} holds only if $u(x)\equiv 0$ for $x\in \Omega$. Then \eqref{eq_G_eps_cont_jump} indicates that
\begin{equation*}
    \frac{\varphi}{2}=\Big(\frac{\partial u}{\partial x_1}\Big)\Big|_{\Gamma_+}=0,
\end{equation*}
which contradicts the assumption that $\varphi\neq 0$. Therefore $\Im(\lambda)\leq 0$. 

Next, we consider the two cases $\Im(\lambda)= 0$ and $\Im(\lambda)<0$, respectively.

Case 1: $\Im(\lambda)= 0$. We aim to prove \eqref{eq_criteria_interface_mode}; then Proposition \ref{corollary_G_eps_decay_blow} indicates that $u(x;\lambda)$ decays exponentially as $|x_1|\to \infty$. By \eqref{eq_G_eps_cont_real}, we write \eqref{eq_interface_match} as
\begin{equation*} 
\begin{aligned}
&\frac{1}{2\pi}\int_{-\pi}^{\pi}\sum_{n\neq \mathfrak{n_*}}\big(\mathbb{K}_{n,\epsilon}(p,\lambda)\varphi\big)dp
+\frac{1}{2\pi}p.v.\int_{-\pi}^{\pi}\big(\mathbb{K}_{\mathfrak{n}_{*},\epsilon}(p,\lambda)\varphi\big)dp \\
&+\frac{1}{2\pi}\int_{-\pi}^{\pi}\sum_{n\neq \mathfrak{n_*}}\big(\mathbb{K}_{n,-\epsilon}(p,\lambda)\varphi\big)dp
+\frac{1}{2\pi}p.v.\int_{-\pi}^{\pi}\big(\mathbb{K}_{\mathfrak{n}_{*},-\epsilon}(p,\lambda)\varphi\big)dp \\
& -
\Bigg(\frac{i\langle \varphi(\cdot),\overline{u_{\mathfrak{n}_*,\epsilon}(\cdot ;q_{+,\epsilon}(\lambda))} \rangle}{2|\lambda_{\mathfrak{n}_*,\epsilon}^{\prime}(q_{+,\epsilon}(\lambda))|}u_{\mathfrak{n}_*,\epsilon}(x ;q_{+,\epsilon}(\lambda))
+
\frac{i\langle \varphi(\cdot),\overline{u_{\mathfrak{n}_*,\epsilon}(\cdot ;q_{-,\epsilon}(\lambda))} \rangle}{2|\lambda_{\mathfrak{n}_*,\epsilon}^{\prime}(q_{-,\epsilon}(\lambda))|}u_{\mathfrak{n}_*,\epsilon}(x ;q_{-,\epsilon}(\lambda))
\Bigg) \\
& -
\Bigg(\frac{i\langle \varphi(\cdot),\overline{u_{\mathfrak{n}_*,-\epsilon}(\cdot ;q_{+,-\epsilon}(\lambda))} \rangle}{2|\lambda_{\mathfrak{n}_*,-\epsilon}^{\prime}(q_{+,-\epsilon}(\lambda))|}u_{\mathfrak{n}_*,-\epsilon}(x ;q_{+,-\epsilon}(\lambda))
+
\frac{i\langle \varphi(\cdot),\overline{u_{\mathfrak{n}_*,-\epsilon}(\cdot ;q_{-,-\epsilon}(\lambda))} \rangle}{2|\lambda_{\mathfrak{n}_*,-\epsilon}^{\prime}(q_{-,-\epsilon}(\lambda))|}u_{\mathfrak{n}_*,-\epsilon}(x ;q_{-,-\epsilon}(\lambda))
\Bigg) \\
&=0.
\end{aligned}
\end{equation*}
Taking dual product with $\overline{\varphi}$ to both sides of the above equation, the imaginary part of the resulted equation yields
\begin{equation*}
\frac{|\langle \varphi(\cdot),\overline{u_{\mathfrak{n}_*,-\epsilon}(\cdot;q_{-,-\epsilon}(\lambda))}|^2}{|\lambda_{\mathfrak{n}_*,-\epsilon}^{'}(q_{-,-\epsilon}(\lambda))|}
=\frac{|\langle \varphi(\cdot),\overline{u_{\mathfrak{n}_*,\epsilon}(\cdot;q_{+,\epsilon}(\lambda))}|^2}{|\lambda_{\mathfrak{n}_*,\epsilon}^{'}(q_{+,\epsilon}(\lambda))|}=0,
\end{equation*}
which gives \eqref{eq_criteria_interface_mode}.  

Case 2: $\Im(\lambda)<0$. If $u(x;\lambda)$ decays exponentially as $|x_1|\to +\infty$, then we draw the same contradiction as in \eqref{eq_lem_ctr_3}. As a result, Proposition \ref{corollary_G_eps_decay_blow} implies that $u(x;\lambda)$ blows up at least in one direction. Thus, $\|u(x;\lambda)\|_{L^2(\Omega)}=\infty$ and $u$ is a resonant mode.
\end{proof}

\subsection{Properties of boundary integral operators}
In this section, we investigate the boundary integral operator $\mathbb{G}^{\Gamma}_{\epsilon}(\lambda_*+\epsilon\cdot h)$ for $h\in\mathcal{J}$, especially its limit as $\epsilon\to 0^{+}$.  The results obtained here pave the way for applying the Gohberg-Sigal theory (see, for instance, \cite{ammari2018mathematical}) to solve the characteristic value problem \eqref{eq_interface_match}.

Here and henceforth, for each $\lambda\in \mathcal{I}_{\epsilon}$, we parameterize $\lambda$ as $\lambda:=\lambda_*+\epsilon\cdot h$ for $h\in \mathcal{J}:=\{h\in \mathbf{C}:|h|<c_0 |t_*|\}$.

We first show the limit of $\mathbb{G}^{\Gamma}_{\epsilon}(\lambda_*+\epsilon\cdot h)$ in $\mathcal{B}(\tilde{H}^{-\frac{1}{2}}(\Gamma),H^{\frac{1}{2}}(\Gamma))$ as $\epsilon\to 0$:
\begin{proposition} \label{prop_limit_integral_operator}
$\mathbb{G}^{\Gamma}_{\epsilon}(\lambda_*+\epsilon\cdot h)$ converges uniformly for $h\in\overline{\mathcal{J}}$ as $\epsilon\to 0$:
\begin{equation*}
\lim_{\epsilon\to 0}\Bigg\|
\mathbb{G}^{\Gamma}_{\epsilon}(\lambda_*+\epsilon\cdot h)
-\Big(2\mathbb{T}+\beta(h)\mathbb{P}^{Dirac}
\Big)
\Bigg\|_{\mathcal{B}(\tilde{H}^{-\frac{1}{2}}(\Gamma),H^{\frac{1}{2}}(\Gamma))}=0, 
\end{equation*}
where
\begin{equation*}
\beta(h):=-\frac{1}{|t_*|\alpha_{n_*}}\frac{h}{\sqrt{1-\frac{h^2}{|t_*|^2}}},\quad
\mathbb{P}^{Dirac}:=\langle \cdot, \overline{v_{n_*}(x;0)}\rangle v_{n_*}(x;0),
\end{equation*}
and
\begin{equation} \label{eq-T}
\begin{aligned}
\mathbb{T}\varphi:=&
\frac{1}{2\pi}\int_{-\pi}^{\pi}\sum_{n\neq n_*,m_*}(\mathbb{K}_{n}(p,\lambda)\varphi) dp \\
&+\Bigg(-\frac{i}{2}\frac{v_{n_{*}}(x;-q_{*})
\langle \cdot,\overline{v_{n_{*}}(y;-q_{*})}\rangle}{|\mu_{n_{*}}^{'}(-q_{*})|}
-\frac{i}{2}\frac{v_{m_{*}}(x;q_{*})
\langle\cdot ,\overline{v_{m_{*}}(y;q_{*})}\rangle}{|\mu_{m_{*}}^{'}(q_{*})|}\\
&\qquad+\frac{1}{2\pi}p.v.\int_{-\pi}^{\pi}(\mathbb{K}_{n_{*}}(p,\lambda)\varphi)dp
+\frac{1}{2\pi}p.v.\int_{-\pi}^{\pi}(\mathbb{K}_{m_{*}}(p,\lambda)\varphi)dp
\Bigg).
\end{aligned}
\end{equation}
Here $\mathbb{K}_{n}(p,\lambda)$ is the integral operator associated with the kernel $K_{n}(x,y;p,\lambda)$ (introduced in Section 2.2).
\end{proposition}
\begin{proof}
See Appendix B.
\end{proof}
We then investigate the operator $\mathbb{T}$. 
\begin{proposition} \label{prop_T_fred_kernel}
$\mathbb{T}\in \mathcal{B}(\tilde{H}^{-\frac{1}{2}}(\Gamma),H^{\frac{1}{2}}(\Gamma))$ is a Fredholm operator of zero index. Moreover, the kernel of $\mathbb{T}$ is given by
\begin{equation*}
\ker(\mathbb{T})=\text{span}\Big\{\frac{\partial v_{n_*}(x;0)}{\partial x_1}\Big|_{\Gamma}\Big\}.
\end{equation*}
\end{proposition}
\begin{proof}
See Appendix C.
\end{proof}
By Proposition \ref{prop_T_fred_kernel}, we can show that the limit operator $\mathbb{G}^{\Gamma}_{0}(h):=2\mathbb{T}+\beta(h)\mathbb{P}^{Dirac}$ has the following properties. The proof is the same as \cite[Proposition 4.6]{qiu2023mathematical}. 
\begin{proposition} \label{prop_G_inter_limit_fred_kernel_charac}
For $h\in\mathcal{J}$, $\mathbb{G}^{\Gamma}_{0}(h)$ is a Fredholm operator with index zero, analytical for $h\in\mathcal{J}$, and continuous for $h\in\partial\mathcal{J}$. As a function of $h$, it attains a unique characteristic value $h=0$ in $\mathcal{J}$, whose null multiplicity is one. Moreover, $\mathbb{G}^{\Gamma}_{0}(h)$ is invertible for any $h\in \overline{\mathcal{J}}$ with $h\neq 0$.
\end{proposition}
As a consequence, we have
\begin{corollary} \label{prop_G_inter_fred}
For $h\in\mathcal{J}$, $\mathbb{G}_{\epsilon}^{\Gamma}(\lambda_*+\epsilon\cdot h)$ is a Fredholm operator with zero index and is analytic as a function $h$.
\end{corollary}
\begin{proof}
By Proposition \ref{prop_limit_integral_operator} and \ref{prop_G_inter_limit_fred_kernel_charac}, and the fact that Fredholm index is stable under small perturbation \cite{qiu2023mathematical}, we conclude that $\mathbb{G}^{\Gamma}_{\epsilon}(\lambda_*+\epsilon\cdot h)$ is a Fredholm operator with zero index for $h\in\mathcal{J}$ and small $\epsilon$. The analyticity of $\mathbb{G}^{\Gamma}_{\epsilon}(\lambda_*+\epsilon\cdot h)$ follows from \eqref{eq_interface_match} and Proposition \ref{prop_G_eps_cont_analyticity}.
\end{proof}

\subsection{Proof of Theorem \ref{thm_main result}}
With properties of the operator $\mathbb{G}^{\Gamma}(\lambda_*+\epsilon\cdot h)$ established in the previous sections, the proof of Theorem \ref{thm_main result} now follows directly from the Gohberg-Sigal theory.

By Proposition \ref{prop_charac_to_resonance}, it's sufficient to show that $\mathbb{G}^{\Gamma}(\lambda_*+\epsilon\cdot h)$ has a characteristic value for $h\in \mathcal{J}$. Since $\mathbb{G}^{\Gamma}_{\epsilon}(\lambda_*+\epsilon\cdot h)\to \mathbb{G}^{\Gamma}_{0}(h)$ (by Proposition \ref{prop_limit_integral_operator}), and $\mathbb{G}^{\Gamma}_{0}(h)$ is invertible for $h\in\partial J$ (by Proposition \ref{prop_limit_integral_operator}), we have
\begin{equation*}
(\mathbb{G}^{\Gamma}_{0}(h))^{-1}\Big(\mathbb{G}^{\Gamma}_{\epsilon}(\lambda_*+\epsilon\cdot h)-\mathbb{G}^{\Gamma}_{0}(h)\Big) \to 0 \,\text{ for }\, h\in\partial J.
\end{equation*}
Note that the convergence is uniform in $h$. As a consequence, the following inequality holds for $h\in\partial \mathcal{J}$ and $\epsilon$ being sufficiently small
\begin{equation*}\label{eq_sec44_1}
\Bigg\|(\mathbb{G}^{\Gamma}_{0}(h))^{-1}\Big(\mathbb{G}^{\Gamma}_{\epsilon}(\lambda_*+\epsilon\cdot h)-\mathbb{G}^{\Gamma}_{0}(h)\Big) \Bigg\|_{\mathcal{B}(\tilde{H}^{-\frac{1}{2}}(\Gamma))}<1.
\end{equation*}
Then the generalized Rouché theorem (see \cite[Theorem 2.9]{qiu2023mathematical}) indicates that, for sufficiently small $\epsilon>0$, $\mathbb{G}^{\Gamma}_{\epsilon}(\lambda_*+\epsilon\cdot h)$ attains a unique characteristic value $\lambda^\star :=\lambda_*+h^\star$ with $h^\star\in \mathcal{J}$. This concludes the proof.

\appendix
\setcounter{secnumdepth}{0}
\section{Appendix}

\subsection{Appendix A: Proof of Lemma \ref{lem_analytic_domain}}
\setcounter{equation}{0}
\setcounter{subsection}{0}
\renewcommand{\theequation}{A.\arabic{equation}}
\renewcommand{\thesubsection}{A.\arabic{subsection}}

\begin{proof}[Proof of Lemma \ref{lem_analytic_domain}]
Step 1: We first show that statement $1$ holds for $\nu=\epsilon^2$ when $\epsilon$ is sufficiently small. 
The key is that $\lambda_{\mathfrak{n}_*,\epsilon}(p)$ is an isolated eigenvalue of $\mathcal{L}_{\epsilon}(p)$ for each $p\in [-\pi,\pi]$, as seen in Figure \ref{fig_perturbed_band_structure}. Since $\{\mathcal{L}_{\epsilon}(p)\}$ is an analytic family indexed by $p$ for fixed $\epsilon\neq 0$, the Kato-Rellich theorem \cite{kato2013perturbation} indicates that $\lambda_{\mathfrak{n}_*,\epsilon}(p)$ is analytic in a neighborhood of $[-\pi,\pi]$. Hence we only need to show the maximal analytic domain of $\lambda_{\mathfrak{n}_*,\epsilon}(p)$ contains $D_{\epsilon,
\nu}$ as claimed. This is achieved following a standard covering argument. For each $p\in [-\pi,\pi]$, $\lambda_{\mathfrak{n}_*,\epsilon}$ is analytic in a neighborhood $B(p,r(p;\epsilon))$ with the radius $r(p;\epsilon)>0$. Moreover, \cite[Chapter VII, Theorem 3.9]{kato2013perturbation} indicates $r(p;\epsilon)$ is estimated by the distance between $\lambda_{\mathfrak{n}_*,\epsilon}$ and its adjacent bands, i.e.
\begin{equation} \label{eq_app_A_radius_distance}
    r(p;\epsilon)\gtrsim d(p;\epsilon):=\min \{|\lambda_{\mathfrak{n}_*,\epsilon}(p)-\lambda_{\mathfrak{n}_*-1,\epsilon}(p)|,|\lambda_{\mathfrak{n}_*,\epsilon}(p)-\lambda_{\mathfrak{n}_*+1,\epsilon}(p)|\}.
\end{equation}
Observe that by Assumption \ref{assump_dirac_points},  $d(p;\epsilon)=O(1)$ for $p$ near $\pm q_*$ (see Figure \ref{fig_perturbed_band_structure}). More precisely,
\begin{equation} \label{eq_app_A_8}
d(p;\epsilon)>d_0>0,\quad \forall p\in (q_*-|\epsilon|^{\frac{1}{3}},q_*+|\epsilon|^{\frac{1}{3}})\cup (-q_*-|\epsilon|^{\frac{1}{3}},-q_*+|\epsilon|^{\frac{1}{3}}),
\end{equation}
for some constant $d_0$ that is independent of $\epsilon$. Thus, by \eqref{eq_app_A_radius_distance},  $B(q_*,|\epsilon|^{\frac{1}{3}})\cup B(-q_*,|\epsilon|^{\frac{1}{3}})$ is contained in the analytic domain of $\lambda_{\mathfrak{n}_*,\epsilon}(p)$. 

We next introduce the following rectangle 
\begin{equation*}
\mathcal{R}_{\nu}:=\{p\in\mathbf{C}:-(1+\nu)\pi<\text{Re}(p)< (1+\nu)\pi, |\Im(p)|<\nu \}. 
\end{equation*}
Note that for $\nu= \epsilon^2$, the width of $R_{\epsilon^2}$ equals $\epsilon^2$. In contrast, $\min_{p\in [-\pi,\pi]}d(p;\epsilon) =O(|\epsilon|)$, which follows from 
Theorem \ref{thm_local_gap_open}. Thus, by the estimate \eqref{eq_app_A_radius_distance}, $\lambda_{\mathfrak{n}_*,\epsilon}(p)$ is analytic inside $R_{\epsilon^2}$ for $\epsilon$ sufficiently small. Therefore, statement $1$ holds for $\epsilon$ sufficiently small and $\nu= \epsilon^2$. 


\medskip

Step 2: We prove that for $\epsilon$ sufficiently small, statement $2$ holds for $\nu= \epsilon^2$. 

Step 2.1: We first note that $\lambda_{\mathfrak{n}_*,\epsilon}(p)$ is analytic in a neighbor of $p=q_*$. For $\epsilon$ small enough,  equation \eqref{eq_perturbed_fold_remainder_2} implies that
\begin{equation*} \label{eq_app_A_9}
\lambda_{\mathfrak{n}_*,\epsilon}^{\prime}(q_*) > \frac{1}{2}\lambda_{\mathfrak{n}_*}^{\prime}(q_*)>0. 
\end{equation*}
By the inverse function theorem for analytical functions, $\lambda_{\mathfrak{n}_*,\epsilon}(\cdot)$ maps the open disc $B(q_*,|\epsilon|^{\frac{1}{3}})$ biholomorphically to an open neighborhood of $\lambda_{\mathfrak{n}_*,\epsilon}(q_*)$, denoted as $U(\lambda_{\mathfrak{n}_*,\epsilon}(q_*))$, which contains an open disc $B(\lambda_{\mathfrak{n}_*,\epsilon}(q_*), c_1|\epsilon|^{\frac{1}{3}})$ for some constant $c_1$ independent of $\epsilon$. In addition, since 
$\lambda_{\mathfrak{n}_*,\epsilon}(p)$ is real-valued for real-valued $p$, $\lambda_{\mathfrak{n}_*,\epsilon}(\cdot)$ maps the upper half-disc $B(q_*,|\epsilon|^{\frac{1}{3}}) \cap \mathbf{C}_+$ (the lower half-disc $B(q_*,|\epsilon|^{\frac{1}{3}}) \cap \mathbf{C}_-$, resp.) to the upper neighborhood $U(\lambda_{\mathfrak{n}_*,\epsilon}(q_*))\cap \mathbf{C}_+$ (the lower half-neighborhood $U(\lambda_{\mathfrak{n}_*,\epsilon}(q_*))\cap \mathbf{C}_-$, resp.). On the other hand, the estimate \eqref{eq_perturbed_fold_eigenvalue_eigenfunction} indicates that $\lambda_{\mathfrak{n}_*,\epsilon}(q_*)-\lambda_*=\mathcal{O}(|\epsilon|)$. 
Then $\mathcal{I}_\epsilon \subset U(\lambda_{\mathfrak{n}_*,\epsilon}(q_*))$ for $\epsilon$ is small enough. 
Therefore, for any $\lambda\in\mathcal{I}_\epsilon$, we conclude that for $\epsilon$ being small enough there exists a unique root $p=q_{+,\epsilon}(\lambda)\in B(q_*,|\epsilon|^{\frac{1}{3}})$ with the estimate $|q_{+,\epsilon}(\lambda)-q_*|=\mathcal{O}(\epsilon)$ and the property $\Im(q_{+,\epsilon}(\lambda))\gtrless 0$ for $\Im(\lambda)\gtrless 0$. 

Next, observe that the reflection symmetry of the system implies the band structure is symmetric about the origin
\begin{equation} \label{eq_sec41_7}
\lambda_{\mathfrak{n}_*,\epsilon}(p)=\lambda_{\mathfrak{n}_*,\epsilon}(-p),\quad p\in[-\pi,\pi].
\end{equation}
The uniqueness of analytic continuation implies that \eqref{eq_sec41_7} holds for all $p\in D_{\epsilon,\epsilon^2}$. Consequently, $q_{-}(\lambda,\epsilon):=-q_{+}(\lambda,\epsilon)$ is also a root of $\lambda_{\mathfrak{n}_*,\epsilon}(p)=\lambda$. It's then straightforward to check that $q_{-,\epsilon}(\lambda)$ satisfies all the desired properties in statement $2$.

Step 2.2: We prove that for $\epsilon$ small enough, $\lambda_{\mathfrak{n}_*,\epsilon}(p)=\lambda$ ($\lambda\in\mathcal{I}_{\epsilon}$) has no root for $p\in D_{\epsilon,\epsilon^2}\backslash (B(q_*,|\epsilon|^{\frac{1}{3}})\cup B(-q_*,|\epsilon|^{\frac{1}{3}})$. By Theorem \ref{thm_local_gap_open}, we have
\begin{equation} \label{eq_sec41_2}
    |\lambda_{\mathfrak{n}_*,\epsilon}(\text{Re}(p))-\lambda_*|\geq |t_*||\epsilon|,\quad \forall p\in D_{\epsilon,\epsilon^2}\backslash (B(q_*,|\epsilon|^{\frac{1}{3}})\cup B(-q_*,|\epsilon|^{\frac{1}{3}}).
\end{equation}
On the other hand, 
\begin{equation} \label{eq_sec41_3}
    |\lambda_{\mathfrak{n}_*,\epsilon}(p)-\lambda_{\mathfrak{n}_*,\epsilon}(\text{Re}(p))|\lesssim |\Im(p)|\lesssim\epsilon^2 ,\quad \forall p\in D_{\epsilon,\epsilon^2}\backslash (B(q_*,|\epsilon|^{\frac{1}{3}})\cup B(-q_*,|\epsilon|^{\frac{1}{3}}).
\end{equation}
As a result, \eqref{eq_sec41_2} and \eqref{eq_sec41_3} imply that for all $\lambda\in\mathcal{I}_{\epsilon}$,
\begin{equation*}
\begin{aligned}
|\lambda_{\mathfrak{n}_*,\epsilon}(p)-\lambda|
&\geq
|\lambda_{\mathfrak{n}_*,\epsilon}(\text{Re}(p))-\lambda_*|-|\lambda-\lambda_*|-|\lambda_{\mathfrak{n}_*,\epsilon}(p)-\lambda_{\mathfrak{n}_*,\epsilon}(\text{Re}(p))| \\
&\geq (1-c_0)|t_*||\epsilon|-\epsilon^2,
\end{aligned}
\end{equation*}
where $c_0<1$ is fixed in Theorem \ref{thm_main result}. Thus for $\epsilon$ small enough, 
\begin{equation*}
|\lambda_{\mathfrak{n}_*,\epsilon}(p)-\lambda|
\geq (1-c_0)|t_*||\epsilon|-\epsilon^2>0 \quad \forall \lambda\in\mathcal{I}_{\epsilon}.
\end{equation*}
Hence $\lambda_{\mathfrak{n}_*,\epsilon}(p)=\lambda$ has no root for $p\in D_{\epsilon,\epsilon^2}\backslash (B(q_*,|\epsilon|^{\frac{1}{3}})\cup B(-q_*,|\epsilon|^{\frac{1}{3}})$. 

In conclusion, we proved that statement 2 holds for $\epsilon$ sufficiently small and $\nu=\epsilon^2$.  

\medskip

Step 3: We prove that statement $3$ holds for $\epsilon$ sufficiently small and properly chosen $\nu=\nu_1(\epsilon)$. Since $ u_{\mathfrak{n}_{*},\epsilon}(x;p)$ is analytic in $p$, so is $\overline{u_{\mathfrak{n}_{*},\epsilon}(x;\overline{p})}$ and $\mathbb{P}_{\mathfrak{n}_*,\epsilon}(p)$. Now we consider the analyticity of  $(\mathcal{L}_{\epsilon}(p)\mathbb{Q}_{\mathfrak{n}_*,\epsilon}(p)-\lambda)^{-1}$. Note that for $p\in [-\pi,\pi]$, we have
\begin{equation*}
\sigma(\mathcal{L}_{\epsilon}(p)\mathbb{Q}_{\mathfrak{n}_*,\epsilon}(p))
=\{0\}\cup \{\lambda_{n,\epsilon}(p)\}_{n\neq \mathfrak{n}_*}.
\end{equation*}
Thus, for $\lambda \in\mathcal{I}_{\epsilon}$ and $p\in [-\pi,\pi]$, it holds that
\begin{equation} \label{eq_app_A_4}
\begin{aligned}
\text{dist}(\lambda,\sigma(\mathcal{L}_{\epsilon}(p)\mathbb{Q}_{\mathfrak{n}_*,\epsilon}(p)))
&\geq \text{dist}(\text{Re}(\lambda),\sigma(\mathcal{L}_{\epsilon}(p)\mathbb{Q}_{\mathfrak{n}_*,\epsilon}(p))) \\
&=\min \{\min_{p\in [-\pi,\pi]} |\text{Re}(\lambda)-\lambda_{\mathfrak{n}_*-1,\epsilon}(p)|,\min_{p\in [-\pi,\pi]} |\text{Re}(\lambda)-\lambda_{\mathfrak{n}_*+1,\epsilon}(p)|\}.
\end{aligned}
\end{equation}
By Assumption \ref{assump_dirac_points}, $\text{Re}(\lambda)\neq \lambda_{\mathfrak{n}_*-1}(p)$ for any $p\in [-\pi,\pi]\backslash \{0\}$. Thus, the perturbation theory implies that  $\epsilon$ small enough, 
\begin{equation*}
\min_{p\in [-\pi,\pi]} |\text{Re}(\lambda)-\lambda_{\mathfrak{n}_*-1,\epsilon}(p)|>0.
\end{equation*}
On the other hand, by equation \eqref{eq_perturbed_dirac_dispersion} in Theorem
\ref{thm_local_gap_open}, we have
\begin{equation*}
\begin{aligned}
\min_{p\in [-\pi,\pi]} |\text{Re}(\lambda)-\lambda_{\mathfrak{n}_*+1,\epsilon}(p)|
&=|\text{Re}(\lambda)-\lambda_{\mathfrak{n}_*+1,\epsilon}(0)| \\
&\gtrsim (1-c_0)|t_*||\epsilon|>0,\quad \text{for any}\, \lambda\in\mathcal{I}_{\epsilon}.
\end{aligned}
\end{equation*}
In conclusion, \eqref{eq_app_A_4} indicates that $\text{dist}(\lambda,\sigma(\mathcal{L}_{\epsilon}(p)\mathbb{Q}_{\mathfrak{n}_*,\epsilon}(p)))>c_2|\epsilon|$ for any $\lambda\in \mathcal{I}_{\epsilon}$ and $p\in [-\pi,\pi]$, where $c_2>0$ is a constant independent of $\epsilon$. Thus the resolvent $(\mathcal{L}_{\epsilon}(p)\mathbb{Q}_{\mathfrak{n}_*,\epsilon}(p)-\lambda)^{-1}$ is well-defined for $\lambda \in\mathcal{I}_{\epsilon}$ and $p\in [-\pi,\pi]$. As a consequence, since $p\mapsto \mathcal{L}_{\epsilon}(p)\mathbb{Q}_{\mathfrak{n}_*,\epsilon}(p)$ is analytic in a complex neighborhood of $[-\pi,\pi]$, denoted as $\mathcal{R}_{\nu}$, with $\nu=\nu_1(\epsilon)>0$, the same holds for its resolvent $(\mathcal{L}_{\epsilon}(p)\mathbb{Q}_{\mathfrak{n}_*,\epsilon}(p)-\lambda)^{-1}$. This completes the proof of statement 3.

\medskip
Step 4. Finally, we conclude that for $\epsilon$ sufficiently small, say $|\epsilon|<\epsilon_0$ for some constant $\epsilon_0>0$, $\nu(\epsilon)=\min\{\nu_1(\epsilon),\epsilon^2\}$ satisfies all the required properties in Lemma \ref{lem_analytic_domain}.
\end{proof}

\subsection{Appendix B: Proof of Proposition \ref{prop_limit_integral_operator}}
\setcounter{equation}{0}
\setcounter{theorem}{0}
\setcounter{subsection}{0}
\renewcommand{\theequation}{B.\arabic{equation}}
\renewcommand{\thetheorem}{B.\arabic{theorem}}
\renewcommand{\thesubsection}{B.\arabic{subsection}}
We write $\mathbb{G}^{\Gamma}_{\epsilon}=\mathbb{G}^{\Gamma,evan}_{\epsilon}+\mathbb{G}^{\Gamma,prop,1}_{\epsilon}+\mathbb{G}^{\Gamma,prop,2}_{\epsilon}$ and $\mathbb{T}=\mathbb{T}^{evan}+\mathbb{T}^{prop}$, where
\begin{equation} \label{eq_B_decom_1}
\begin{aligned}
\mathbb{G}^{\Gamma,evan}_{\epsilon}(\lambda_*+\epsilon h)
&:=\frac{1}{2\pi}\int_{-\pi}^{\pi}
\sum_{n\neq \mathfrak{n_*},\mathfrak{n_*}+1}\mathbb{K}_{n,\epsilon}(p,\lambda_*+\epsilon h)dp 
 +\frac{1}{2\pi}\int_{-\pi}^{\pi}
\sum_{n\neq \mathfrak{n_*},\mathfrak{n_*}+1}\mathbb{K}_{n,-\epsilon}(p,\lambda_*+\epsilon h)dp \\
&=:\tilde{\mathbb{G}}^{evan}_{\epsilon}(\lambda_*+\epsilon h)
+\tilde{\mathbb{G}}^{evan}_{-\epsilon}(\lambda_*+\epsilon h),
\end{aligned}
\end{equation}

\begin{equation} \label{eq_B_decom_2}
\begin{aligned}
\mathbb{G}^{\Gamma,prop,1}_{\epsilon}(\lambda_*+\epsilon h)
&:=\Big(\frac{1}{2\pi}\int_{C_{\epsilon}\backslash (-\epsilon^{\frac{1}{3}},\epsilon^{\frac{1}{3}})}
\mathbb{K}_{\mathfrak{n}_{*},\epsilon}(p,\lambda_*+\epsilon h)dp 
+
\frac{1}{2\pi}\int_{(-\pi,\pi)\backslash (-\epsilon^{\frac{1}{3}},\epsilon^{\frac{1}{3}})}\mathbb{K}_{\mathfrak{n}_{*}+1,\epsilon}(p,\lambda_*+\epsilon h)dp\Big) \\
&\quad +\Big(\frac{1}{2\pi}\int_{C_{\epsilon}\backslash (-\epsilon^{\frac{1}{3}},\epsilon^{\frac{1}{3}})}
\mathbb{K}_{\mathfrak{n}_{*},-\epsilon}(p,\lambda_*+\epsilon h)dp 
+ \frac{1}{2\pi}\int_{(-\pi,\pi)\backslash (-\epsilon^{\frac{1}{3}},\epsilon^{\frac{1}{3}})}\mathbb{K}_{\mathfrak{n}_{*}+1,-\epsilon}(p,\lambda_*+\epsilon h)dp\Big) \\
&=:\tilde{\mathbb{G}}^{prop,1}_{\epsilon}(\lambda_*+\epsilon h)+\tilde{\mathbb{G}}^{prop,1}_{-\epsilon}(\lambda_*+\epsilon h),
\end{aligned}
\end{equation}

\begin{equation} \label{eq_B_decom_3}
\begin{aligned}
\mathbb{G}^{\Gamma,prop,2}_{\epsilon}(\lambda_*+\epsilon h)
&:=\frac{1}{2\pi}
\int_{-\epsilon^{\frac{1}{3}}}^{\epsilon^{\frac{1}{3}}}
\mathbb{K}_{\mathfrak{n}_{*},\epsilon}(p,\lambda_*+\epsilon h)dp
+\frac{1}{2\pi}
\int_{-\epsilon^{\frac{1}{3}}}^{\epsilon^{\frac{1}{3}}}\mathbb{K}_{\mathfrak{n}_{*}+1,\epsilon}(p,\lambda_*+\epsilon h)dp \\
&\quad +\frac{1}{2\pi}\int_{-\epsilon^{\frac{1}{3}}}^{\epsilon^{\frac{1}{3}}}
\mathbb{K}_{\mathfrak{n}_{*},-\epsilon}(p,\lambda_*+\epsilon h)dp
+\frac{1}{2\pi}\int_{-\epsilon^{\frac{1}{3}}}^{\epsilon^{\frac{1}{3}}}\mathbb{K}_{\mathfrak{n}_{*}+1,-\epsilon}(p,\lambda_*+\epsilon h)dp \\
&=:\mathbb{J}_{1, \epsilon}( h)+\mathbb{J}_{2, \epsilon}(h)+\mathbb{J}_{1, -\epsilon}(h)+\mathbb{J}_{2, -\epsilon}(h)
\end{aligned}
\end{equation}
and
\begin{equation*} 
\begin{aligned}
\mathbb{T}^{evan}
&:=\frac{1}{2\pi}\int_{-\pi}^{\pi}\sum_{n\neq n_*,m_*}\mathbb{K}_{n}(p,\lambda_*)dp ,
\end{aligned}
\end{equation*}

\begin{equation*} 
\begin{aligned}
\mathbb{T}^{prop}
&:=\frac{1}{2\pi}p.v.\int_{-\pi}^{\pi}\mathbb{K}_{n_*}(p,\lambda_*)dp
+\frac{1}{2\pi}p.v.\int_{-\pi}^{\pi}\mathbb{K}_{m_*}(p,\lambda_*)dp \\
&\quad -\frac{i}{2}\frac{v_{n_{*}}(x;-q_{*})
\langle \cdot,\overline{v_{n_{*}}(y;-q_{*})}\rangle}{|\mu_{n_{*}}^{'}(-q_{*})|}
-\frac{i}{2}\frac{v_{m_{*}}(x;q_{*})
\langle\cdot ,\overline{v_{m_{*}}(y;q_{*})}\rangle}{|\mu_{m_{*}}^{'}(q_{*})|}.
\end{aligned}
\end{equation*}
Here $\mathbb{K}_{n,\epsilon}$ is the integral operator associated with the kernel $K_{n,\epsilon}$ as defined in \eqref{eq_K_n_eps_operator_def}, and $\mathbb{K}_{n}$ is defined similarly with the kernel replaced by $K_n$ in \eqref{eq_K_n_def}. Then 
\begin{equation*}
\begin{aligned}
\mathbb{G}^{\Gamma}_{\epsilon}(\lambda_*+\epsilon h)-\Big(2\mathbb{T}+\beta(h)\mathbb{P}^{Dirac}\Big)
=&\Big(\tilde{\mathbb{G}}^{evan}_{\epsilon}(\lambda_*+\epsilon h)
-\mathbb{T}^{evan}
\Big)+\Big(\tilde{\mathbb{G}}^{evan}_{-\epsilon}(\lambda_*+\epsilon h)
-\mathbb{T}^{evan}
\Big) \\
&+\Big(\tilde{\mathbb{G}}^{prop,1}_{\epsilon}(\lambda_*+\epsilon h)
-\mathbb{T}^{prop}
\Big)+\Big(\tilde{\mathbb{G}}^{prop,1}_{-\epsilon}(\lambda_*+\epsilon h)
-\mathbb{T}^{prop}
\Big) \\
&+\Big(\mathbb{G}^{\Gamma,prop,2}_{\epsilon}(\lambda_*+\epsilon h)
-\beta(h)\mathbb{P}^{Dirac}
\Big),
\end{aligned}
\end{equation*}
and hence Proposition \ref{prop_limit_integral_operator} is a direct consequence of Lemma \ref{lem_app_B_1}, \ref{lem_app_B_2} and \ref{lem_app_B_3} below. Note that all the convergences therein are uniform in $h\in\overline{\mathcal{J}}$.
\begin{lemma} \label{lem_app_B_1}
\begin{equation} \label{eq_app_B_3}
\lim_{\epsilon\to 0}\Big\|
\tilde{\mathbb{G}}^{evan}_{\pm\epsilon}(\lambda_*+\epsilon h) - \mathbb{T}^{evan}
\Big\|_{\mathcal{B}(\tilde{H}^{-\frac{1}{2}}(\Gamma),H^{\frac{1}{2}}(\Gamma))}
=0.
\end{equation}
\end{lemma}

\begin{lemma} \label{lem_app_B_2}
\begin{equation} \label{eq_app_B_4}
\lim_{\epsilon\to 0}\Big\|
\tilde{\mathbb{G}}^{prop,1}_{\pm\epsilon}(\lambda_*+\epsilon h)
-\mathbb{T}^{prop}
\Big\|_{\mathcal{B}(\tilde{H}^{-\frac{1}{2}}(\Gamma),H^{\frac{1}{2}}(\Gamma))}
=0.
\end{equation}
\end{lemma}
\begin{lemma} \label{lem_app_B_3}
\begin{equation} \label{eq_app_B_29}
\lim_{\epsilon\to 0}\Big\|
\mathbb{G}^{\Gamma,prop,2}_{\epsilon}(\lambda_*+\epsilon h)
-\beta(h)\mathbb{P}^{Dirac}
\Big\|_{\mathcal{B}(\tilde{H}^{-\frac{1}{2}}(\Gamma),H^{\frac{1}{2}}(\Gamma))}
=0.
\end{equation}
\end{lemma}

\begin{proof}[Proof of Lemma \ref{lem_app_B_1}]
We prove \eqref{eq_app_B_3} for $\tilde{\mathbb{G}}^{evan}_{\epsilon}(\lambda_*+\epsilon h)$, while the proof of $\tilde{\mathbb{G}}^{evan}_{-\epsilon}(\lambda_*+\epsilon h)$ is the same. We define
\begin{equation*}
\begin{aligned}
&\tilde{\mathbb{G}}_{+,\epsilon}^{evan}(\lambda_*+\epsilon h) 
:=\frac{1}{2\pi}\int_{-\pi}^{\pi}
\sum_{n>\mathfrak{n_*}+1}\mathbb{K}_{n,\epsilon}(p,\lambda_*+\epsilon h)dp, \quad 
\\
&\tilde{\mathbb{G}}_{-,\epsilon}^{evan}(\lambda_*+\epsilon h) 
:=\frac{1}{2\pi}\int_{-\pi}^{\pi}
\sum_{n<\mathfrak{n_*}}\mathbb{K}_{n,\epsilon}(p,\lambda_*+\epsilon h)dp.  
\end{aligned}
\end{equation*}
Correspondingly,
\begin{equation*}
\mathbb{T}_{+} 
:=\frac{1}{2\pi}\int_{-\pi}^{\pi}\sum_{n>\mathfrak{n_*}+1}\mathbb{K}_{n}(p,\lambda_*)dp,\quad
\mathbb{T}_{-} 
:=\frac{1}{2\pi}\int_{-\pi}^{\pi}\sum_{n<\mathfrak{n_*}}\mathbb{K}_{n}(p,\lambda_*)dp. 
\end{equation*}
Note that \eqref{eq_app_B_3} follows from the decompositions $\tilde{\mathbb{G}}^{evan}_{\epsilon}=\tilde{\mathbb{G}}_{+,\epsilon}^{evan}+\tilde{\mathbb{G}}_{-,\epsilon}^{evan}$, $\mathbb{T}=\mathbb{T}_{+}+\mathbb{T}_{-} $, and the following identities
\begin{equation} \label{eq_app_B_6}
\lim_{\epsilon\to 0}\Big\|
\tilde{\mathbb{G}}_{+,\epsilon}^{evan}(\lambda_*+\epsilon h) 
-\mathbb{T}_{+} 
\Big\|_{\mathcal{B}(\tilde{H}^{-\frac{1}{2}}(\Gamma),H^{\frac{1}{2}}(\Gamma))}
=0,
\end{equation}
\begin{equation} \label{eq_app_B_7}
\lim_{\epsilon\to 0}\Big\|
\tilde{\mathbb{G}}_{-,\epsilon}^{evan}(\lambda_*+\epsilon h) 
-\mathbb{T}_{-} 
\Big\|_{\mathcal{B}(\tilde{H}^{-\frac{1}{2}}(\Gamma),H^{\frac{1}{2}}(\Gamma))}
=0.
\end{equation}
We prove \eqref{eq_app_B_6} and \eqref{eq_app_B_7} in the following two steps.

\medskip
Step 1: We first prove \eqref{eq_app_B_6}. Define 
\begin{equation*}
\mathbb{P}_{+,\epsilon}(p)f:=\sum_{n\leq \mathfrak{n}_*+1}(f,u_{n,\epsilon}(x;\overline{p}))_{L^2(Y;n_{\epsilon}(x))}u_{n,\epsilon}(x;p),\quad
\mathbb{Q}_{+,\epsilon}(p):=1-\mathbb{P}_{+,\epsilon}(p).
\end{equation*}
Then $\tilde{\mathbb{G}}^{evan}_{+,\epsilon}(\lambda_*+\epsilon h)$ and $\mathbb{T}_{+}^{evan}$ can be rewritten as follows, which are obtained by expanding the resolvents in Floquet series respectively: 
\begin{equation} \label{eq_app_B_5}
\begin{aligned}
\tilde{\mathbb{G}}^{evan}_{+,\epsilon}(\lambda_*+\epsilon h)
&=\frac{1}{2\pi}\int_{-\pi}^{\pi}\Big(
\text{Tr}\circ(\mathcal{L}_{\epsilon}(p)\mathbb{Q}_{+,\epsilon}(p)-\lambda_*-\epsilon h)^{-1}\mathbb{Q}_{+,\epsilon}(p)\circ(\frac{1}{n_{\epsilon}^2(y)}\mathcal{M})
\Big)dp, \\
\mathbb{T}_{+}^{evan}
&=\frac{1}{2\pi}\int_{-\pi}^{\pi}\Big(
\text{Tr}\circ(\mathcal{L}(p)\mathbb{Q}_{+}(p)-\lambda_*)^{-1}\mathbb{Q}_{+}(p)\circ(\frac{1}{n^2(y)}\mathcal{M})
\Big)dp.
\end{aligned}
\end{equation}
By Lemma \ref{lem_analytic_domain}, 
the map $p\mapsto \text{Tr}\circ(\mathcal{L}_{\epsilon}(p)\mathbb{Q}_{+,\epsilon}(p)-\lambda_*-\epsilon h)^{-1}\mathbb{Q}_{+,\epsilon}(p)\circ (\frac{1}{n_{\epsilon}^2(y)}\mathcal{M})$ is analytic near the real line. We claim that 
\begin{enumerate}[label =(\arabic*)]
    \item $\text{Tr}\circ(\mathcal{L}_{\epsilon}(p)\mathbb{Q}_{+,\epsilon}(p)-\lambda_*-\epsilon h)^{-1}\mathbb{Q}_{+,\epsilon}(p)\circ (\frac{1}{n_{\epsilon}^2(y)}\mathcal{M})\in\mathcal{B}(\tilde{H}^{-\frac{1}{2}}(\Gamma),H^\frac{1}{2}(\Gamma))$ has uniformly bounded operator norm for $p\in [-\pi,\pi]$ and $\epsilon \ll 1$;
    \item $\text{Tr}\circ(\mathcal{L}_{\epsilon}(p)\mathbb{Q}_{+,\epsilon}(p)-\lambda_*-\epsilon h)^{-1}\mathbb{Q}_{+,\epsilon}(p)\circ (\frac{1}{n_{\epsilon}^2(y)}\mathcal{M})$ converges to $\text{Tr}\circ(\mathcal{L}(p)\mathbb{Q}_{+}(p)-\lambda_*)^{-1}\mathbb{Q}_{+}(p)\circ (\frac{1}{n^2(y)}\mathcal{M})$ in operator norm for each $p\in [-\pi,\pi]$. 
\end{enumerate}
Then \eqref{eq_app_B_6} follows from \eqref{eq_app_B_5} and the dominated convergence theorem. The proofs of claim (1)-(2) are presented in Steps 1.1 and 1.2 below, respectively.

\medskip
Step 1.1: Note that 
\begin{equation*}
\sigma(\mathcal{L}_{\epsilon}(p)\mathbb{Q}_{+}(p))=\{0\}\cup \{\lambda_{n}(p)\}_{n>\mathfrak{n}_*+1}.
\end{equation*}
Thus, for $h\in\overline{\mathcal{J}}$, $\lambda_*+\epsilon h\notin \sigma(\mathcal{L}_{\epsilon}(p)\mathbb{Q}_{+}(p))$ and the resolvent $(\mathcal{L}_{\epsilon}(p)\mathbb{Q}_{+}(p)-\lambda_*-\epsilon h)^{-1}$ is well-defined. Moreover,
\begin{equation*}
\|(\mathcal{L}_{\epsilon}(p)\mathbb{Q}_{+,\epsilon}(p)-\lambda_*-\epsilon h)^{-1}\|_{\mathcal{B}((H^{1}_{p,b}(\Omega)^*,H^{1}_{p,b}(\Omega)))}
\leq \frac{1}{dist(\lambda_*+\epsilon h,\sigma(\mathcal{L}_{\epsilon}(p)\mathbb{Q}_{+}(p)))}=\mathcal{O}(1),
\end{equation*}
where the estimate is uniform for $p\in [-\pi,\pi]$. Thus,
\begin{equation*}
\begin{aligned}
\Big\|\text{Tr}\circ(\mathcal{L}_{\epsilon}(p)\mathbb{Q}_{+,\epsilon}(p)&-\lambda_*-\epsilon h)^{-1}\mathbb{Q}_{+,\epsilon}(p)\circ (\frac{1}{n_{\epsilon}^2(y)}\mathcal{M})\Big\|_{\mathcal{B}(\tilde{H}^{-\frac{1}{2}}(\Gamma),H^\frac{1}{2}(\Gamma))} 
=\mathcal{O}(1),
\end{aligned}
\end{equation*}
and claim (1) follows.

\medskip

Step 1.2: Note that $\mathcal{L}_{\epsilon}(p)\mathbb{Q}_{+,\epsilon}(p)$ converges to $\mathcal{L}(p)\mathbb{Q}_{+}(p)$ in the generalized sense as $\epsilon\to 0$ for each $p\in[-\pi,\pi]$. As a consequence, 
\begin{equation} \label{eq_app_B_10}
\lim_{\epsilon\to 0}\Big\|(\mathcal{L}_{\epsilon}(p)\mathbb{Q}_{+,\epsilon}(p)-\lambda_*)^{-1}-(\mathcal{L}(p)\mathbb{Q}_{+}(p)-\lambda_*)^{-1}\Big\|_{\mathcal{B}((H^{1}_{p,b}(\Omega)^*,H^{1}_{p,b}(\Omega)))}=0
\end{equation}
for each $p\in[-\pi,\pi]$. On the other hand, for each $n\leq \mathfrak{n}_*+1$, $\lim_{\epsilon\to 0}\|u_{n,\epsilon}(\cdot;p)-u_{n}(\cdot;p)\|_{H^1_{p,b}}\to 0$. So $\lim_{\epsilon\to 0}\|\mathbb{P}_{+,\epsilon}(p)-\mathbb{P}_{+}(p)\|_{\mathcal{B}((H^{1}_{p,b}(\Omega)^*,H^{1}_{p,b}(\Omega)))}\to 0$. Thus,
\begin{equation} \label{eq_app_B_12}
\begin{aligned}
\lim_{\epsilon\to 0}\|
\mathbb{Q}_{+,\epsilon}(p)
-\mathbb{Q}_{+}(p)\|_{\mathcal{B}((H^{1}_{p,b}(\Omega)^*,H^{1}_{p,b}(\Omega)))}
=\lim_{\epsilon\to 0}\|
\mathbb{P}_{+,\epsilon}(p)
-\mathbb{P}_{+}(p)\|_{\mathcal{B}((H^{1}_{p,b}(\Omega)^*,H^{1}_{p,b}(\Omega)))}=0.
\end{aligned}
\end{equation}
\eqref{eq_app_B_10} and \eqref{eq_app_B_12} imply that
\begin{equation} \label{eq_app_B_25}
\begin{aligned}
&\Big\|(\mathcal{L}_{\epsilon}(p)\mathbb{Q}_{+,\epsilon}(p)-\lambda_*)^{-1}\mathbb{Q}_{+,\epsilon}(p)-(\mathcal{L}(p)\mathbb{Q}_{+}(p)-\lambda_*)^{-1}\mathbb{Q}_{+}(p)\Big\|_{\mathcal{B}((H^{1}_{p,b}(\Omega)^*,H^{1}_{p,b}(\Omega)))} \\
&\leq \Big\|(\mathcal{L}_{\epsilon}(p)\mathbb{Q}_{+,\epsilon}(p)-\lambda_*)^{-1}-(\mathcal{L}(p)\mathbb{Q}_{+}(p)-\lambda_*)^{-1}\Big\|_{\mathcal{B}((H^{1}_{p,b}(\Omega)^*,H^{1}_{p,b}(\Omega)))}\cdot\Big\|\mathbb{Q}_{+,\epsilon}(p)\Big\|_{\mathcal{B}((H^{1}_{p,b}(\Omega)^*,H^{1}_{p,b}(\Omega)))} \\
&\quad +\Big\|(\mathcal{L}(p)\mathbb{Q}_{+}(p)-\lambda_*)^{-1}\Big\|_{\mathcal{B}((H^{1}_{p,b}(\Omega)^*,H^{1}_{p,b}(\Omega)))}\cdot \Big\|
\mathbb{Q}_{+,\epsilon}(p)
-\mathbb{Q}_{+}(p)\Big\|_{\mathcal{B}((H^{1}_{p,b}(\Omega)^*,H^{1}_{p,b}(\Omega)))} \\
&\to 0, \quad \mbox{as} \,\,\, \epsilon \to 0.
\end{aligned}
\end{equation}
Next, note that
\begin{equation*}
\begin{aligned}
\Big\|(\mathcal{L}_{\epsilon}(p)&\mathbb{Q}_{+,\epsilon}(p)-\lambda_*)^{-1}\mathbb{Q}_{+,\epsilon}(p)
-(\mathcal{L}_{\epsilon}(p)\mathbb{Q}_{+,\epsilon}(p)-\lambda_*-\epsilon h)^{-1}\mathbb{Q}_{+,\epsilon}(p)\Big\|_{\mathcal{B}((H^{1}_{p,b}(\Omega)^*,H^{1}_{p,b}(\Omega)))} \\
&=|\epsilon h|\cdot \Big\|(\mathcal{L}_{\epsilon}(p)\mathbb{Q}_{+,\epsilon}(p)-\lambda_*)^{-1}\circ (\mathcal{L}_{\epsilon}(p)\mathbb{Q}_{+,\epsilon}(p)-\lambda_*-\epsilon h)^{-1}\mathbb{Q}_{+,\epsilon}(p)\Big\|_{\mathcal{B}((H^{1}_{p,b}(\Omega)^*,H^{1}_{p,b}(\Omega)))}.
\end{aligned}
\end{equation*}
Since $\|(\mathcal{L}_{\epsilon}(p)\mathbb{Q}_{+,\epsilon}(p)-\lambda_*-\epsilon h)^{-1}\|$ is uniformly bounded (proved in Step 1.1), we have
\begin{equation} \label{eq_app_B_11}
\begin{aligned}
\lim_{\epsilon\to 0}\Big\|(\mathcal{L}_{\epsilon}(p)\mathbb{Q}_{+,\epsilon}(p)-\lambda_*)^{-1}\mathbb{Q}_{+,\epsilon}(p)
-(\mathcal{L}_{\epsilon}(p)\mathbb{Q}_{+,\epsilon}(p)-\lambda_*-\epsilon h)^{-1}\mathbb{Q}_{+,\epsilon}(p)\Big\|_{\mathcal{B}((H^{1}_{p,b}(\Omega)^*,H^{1}_{p,b}(\Omega)))}=0.
\end{aligned}
\end{equation}
Combing \eqref{eq_app_B_11} with \eqref{eq_app_B_25}, we obtain
\begin{equation*}
\lim_{\epsilon\to 0}
\Big\|(\mathcal{L}_{\epsilon}(p)\mathbb{Q}_{+,\epsilon}(p)-\lambda_*-\epsilon h)^{-1}\mathbb{Q}_{+,\epsilon}(p)
-
(\mathcal{L}(p)\mathbb{Q}_{+}(p)-\lambda_*)^{-1}\mathbb{Q}_{+}(p)\Big\|_{\mathcal{B}((H^{1}_{p,b}(\Omega)^*,H^{1}_{p,b}(\Omega)))}=0,
\end{equation*}
whence the point-wise convergence in claim (2) follows. 

\medskip

Step 2: The proof of \eqref{eq_app_B_7} is also based on the dominated convergence theorem. With the projection operator $\mathbb{P}_{n,\epsilon}(p)$ defined in Lemma \ref{lem_analytic_domain}, we write
\begin{equation} \label{eq_app_B_13}
\begin{aligned}
&\tilde{\mathbb{G}}_{-,\epsilon}^{evan}(\lambda_*+\epsilon h)
=\frac{1}{2\pi}\sum_{n=1}^{\mathfrak{n}_*-1}\int_{-\pi}^{\pi}
\frac{\text{Tr}\circ\mathbb{P}_{n,\epsilon}(p)\circ(\frac{1}{n^2_{\epsilon}(y)}\mathcal{M})}{\lambda_*+\epsilon h-\lambda_{n,\epsilon}(p)}dp,\\
&\mathbb{T}_{-}
=\frac{1}{2\pi}\sum_{n=1}^{\mathfrak{n}_*-1}\int_{-\pi}^{\pi}
\frac{\text{Tr}\circ\mathbb{P}_{n}(p)\circ (\frac{1}{n^2(y)}\mathcal{M})}{\lambda_*-\lambda_{n}(p)}dp.
\end{aligned}
\end{equation}
For $1\leq n\leq \mathfrak{n}_*-1$, we have
\begin{equation*}
\Big\|\frac{\text{Tr}\circ\mathbb{P}_{n,\epsilon}(p)\circ (\frac{1}{n^2_{\epsilon}(y)}\mathcal{M})}{\lambda_*+\epsilon h-\lambda_{n}(p)}\Big\|_{\mathcal{B}(\tilde{H}^{-\frac{1}{2}}(\Gamma),H^\frac{1}{2}(\Gamma))} 
\lesssim
\frac{1}{|\lambda_*+\epsilon h-\lambda_{n}(p)|}=\mathcal{O}(1). 
\end{equation*}
On the other hand, since
\begin{equation*}
\lim_{\epsilon\to 0}\|\mathbb{P}_{n,\epsilon}(p)-\mathbb{P}_{n}(p)\|_{\mathcal{B}((H^{1}_{p,b}(\Omega)^*,H^{1}_{p,b}(\Omega)))}=0,\quad
\lim_{\epsilon\to 0}|\lambda_{n,\epsilon}(p)-\lambda_{n}(p)|=0 \quad (\forall p\in[-\pi,\pi]),  
\end{equation*}
$\frac{\text{Tr}\circ\mathbb{P}_{n,\epsilon}(p)\circ (\frac{1}{n^2_{\epsilon}(y)}\mathcal{M})}{\lambda_*+\epsilon\cdot h-\lambda_{n,\epsilon}(p)}$ converges to $\frac{\text{Tr}\circ\mathbb{P}_{n}(p)\circ (\frac{1}{n^2(y)}\mathcal{M})}{\lambda_*-\lambda_{n}(p)}$ in operator norm for each $p\in[-\pi,\pi]$. The dominated convergence theorem implies that for each $n$ with $1\leq n\leq \mathfrak{n}_*-1$,
\begin{equation*}
\Big\|
\int_{-\pi}^{\pi}
\frac{\text{Tr}\circ\mathbb{P}_{n,\epsilon}(p)\circ (\frac{1}{n^2_{\epsilon}(y)}\mathcal{M})}{\lambda_*+\epsilon h-\lambda_{n,\epsilon}(p)}dp
-\int_{-\pi}^{\pi}
\frac{\text{Tr}\circ\mathbb{P}_{n}(p)\circ (\frac{1}{n^2(y)}\mathcal{M})}{\lambda_*-\lambda_{n}(p)}dp
\Big\|_{\mathcal{B}(\tilde{H}^{-\frac{1}{2}}(\Gamma),H^\frac{1}{2}(\Gamma))} \to 0. 
\end{equation*}
Since the summation in \eqref{eq_app_B_13} is finite, the convergence of $\tilde{\mathbb{G}}_{-,\epsilon}^{evan}(\lambda_*+\epsilon h)$ follows.  
\end{proof}

\begin{proof}[Proof of Lemma \ref{lem_app_B_2}]
Let $\gamma_{-,\epsilon}:=\{-q_{*}+\epsilon^{\frac{1}{3}}e^{i\theta}:\pi\geq \theta \geq 0\}$, $\gamma_{+,\epsilon}:=\{q_{*}+\epsilon^{\frac{1}{3}}e^{i\theta}:\pi\leq \theta \leq 2\pi\}$. We decompose the contour $C_\epsilon$ as
$$
C_\epsilon=[-\pi,-q_*-\epsilon^{\frac{1}{3}}]\cup \gamma_{-,\epsilon}\cup[-q_*+\epsilon^{\frac{1}{3}},-\epsilon^{\frac{1}{3}}]\cup [-\epsilon^{\frac{1}{3}},\epsilon^{\frac{1}{3}}]\cup[\epsilon^{\frac{1}{3}},q_*-\epsilon^{\frac{1}{3}}]\cup \gamma_{+,\epsilon}\cup [q_*+\epsilon^{\frac{1}{3}},\pi].
$$
Our strategy is first to decompose the operator $\tilde{\mathbb{G}}^{prop,1}_{\pm\epsilon}(\lambda_*+\epsilon h)$ into four parts according to the decomposition of the contour $C_\epsilon$, and then prove the convergence of each part. More precisely, we shall prove the following convergences in $\mathcal{B}(\tilde{H}^{-\frac{1}{2}}(\Gamma),H^\frac{1}{2}(\Gamma))$: 
\begin{equation} \label{eq_app_B_14}
\begin{aligned}
\Big\|
\int_{\substack{[-\pi,-q_*-\epsilon^{\frac{1}{3}}]\\ \cup [-q_*+\epsilon^{\frac{1}{3}},-\epsilon^{\frac{1}{3}}]}}
\mathbb{K}_{\mathfrak{n}_{*},\pm\epsilon}(p,\lambda_*+\epsilon h)dp
+
\int_{\epsilon^{\frac{1}{3}}}^{\pi}\mathbb{K}_{\mathfrak{n}_{*}+1,\pm\epsilon}(p,\lambda_*+\epsilon h)dp
-
p.v.\int_{-\pi}^{\pi}\mathbb{K}_{n_*}(p,\lambda_*)dp
\Big\|\to 0.
\end{aligned}
\end{equation}

\begin{equation} \label{eq_app_B_15}
\begin{aligned}
\Big\|
\int_{\substack{[\epsilon^{\frac{1}{3}},q_*-\epsilon^{\frac{1}{3}}]\\ \cup [q_*+\epsilon^{\frac{1}{3}},\pi]}}
\mathbb{K}_{\mathfrak{n}_{*},\pm\epsilon}(p,\lambda_*+\epsilon h)dp
+
\int_{-\pi}^{-\epsilon^{\frac{1}{3}}}\mathbb{K}_{\mathfrak{n}_{*}+1,\pm\epsilon}(p,\lambda_*+\epsilon h)dp
-
p.v.\int_{-\pi}^{\pi}\mathbb{K}_{m_{*}}(p,\lambda_*)dp
\Big\|\to 0.
\end{aligned}
\end{equation}

\begin{equation} \label{eq_app_B_16}
\begin{aligned}
\Big\|\frac{1}{2\pi}
\int_{\gamma_{-,\epsilon}}
\mathbb{K}_{\mathfrak{n}_{*},\pm\epsilon}(p,\lambda_*+\epsilon h)dp
+
\frac{i}{2}\frac{v_{n_{*}}(x;-q_{*})
\langle \cdot,\overline{v_{n_{*}}(y;-q_{*})}\rangle}{|\mu_{n_{*}}^{\prime}(-q_{*})|}
\Big\|\to 0.
\end{aligned}
\end{equation}

\begin{equation} \label{eq_app_B_17}
\begin{aligned}
\Big\|\frac{1}{2\pi}
\int_{\gamma_{+,\epsilon}}
\mathbb{K}_{\mathfrak{n}_{*},\pm\epsilon}(p,\lambda_*+\epsilon h)dp
+
\frac{i}{2}\frac{v_{m_{*}}(x;q_{*})
\langle \cdot,\overline{v_{m_{*}}(y;q_{*})}\rangle}{|\mu_{m_{*}}^{\prime}(q_{*})|}
\Big\|\to 0.
\end{aligned}
\end{equation}
Note that the proofs of \eqref{eq_app_B_14} and \eqref{eq_app_B_15} are the same as the proof of \cite[Proposition 4.5]{qiu2023mathematical}. We skip it and refer the reader to \cite{qiu2023mathematical} for details. In what follows, we prove \eqref{eq_app_B_16} for $\mathbb{K}_{\mathfrak{n}_{*},\epsilon}(p,\lambda_*+\epsilon h)$; the proof of the case concerning $\mathbb{K}_{\mathfrak{n}_{*},-\epsilon}(p,\lambda_*+\epsilon h)$, and the proof of \eqref{eq_app_B_17} are similar.

\medskip

Observe that  \eqref{eq_app_B_16} follows from the following identities: 
\begin{equation} \label{eq_app_B_19}
\Big\|\frac{1}{2\pi}
\int_{\gamma_{-,\epsilon}}
\mathbb{K}_{\mathfrak{n}_{*},\epsilon}(p,\lambda_*+\epsilon h)dp
-\frac{1}{2\pi}
\int_{\gamma_{-,\epsilon}}
\mathbb{K}_{\mathfrak{n}_{*},\epsilon}(p,\lambda_*)dp
\Big\|\to 0,
\end{equation}
\begin{equation} \label{eq_app_B_20}
\Big\|\frac{1}{2\pi}
\int_{\gamma_{-,\epsilon}}
\mathbb{K}_{\mathfrak{n}_{*},\epsilon}(p,\lambda_*)dp
-\frac{1}{2\pi}
\int_{\gamma_{-,\epsilon}}
\mathbb{K}_{\mathfrak{n}_{*}}(p,\lambda_*)dp
\Big\|\to 0,
\end{equation}
\begin{equation} \label{eq_app_B_21}
\Big\|\frac{1}{2\pi}
\int_{\gamma_{-,\epsilon}}
\mathbb{K}_{\mathfrak{n}_{*}}(p,\lambda_*)dp
+
\frac{i}{2}\frac{v_{n_{*}}(x;-q_{*})
\langle \cdot,\overline{v_{n_{*}}(y;-q_{*})}\rangle}{|\mu_{n_{*}}^{\prime}(-q_{*})|}
\Big\|\to 0.
\end{equation}
Given that \eqref{eq_app_B_20} follows from the smoothness of $\lambda_{\mathfrak{n}_*,\epsilon}$ and $u_{\mathfrak{n}_*,\epsilon}$ in $\epsilon$ (inherited from the smoothness of $\mathcal{L}_{\epsilon}$ in $\epsilon$ by the standard perturbation theory), and \eqref{eq_app_B_21} from the Residue theorem, we only need to prove \eqref{eq_app_B_19} to conclude the proof. For $p\in \gamma_{-,\epsilon}$, we have $|q-(-q_*)|=|\epsilon|^{\frac{1}{3}}$. Then an application of the inverse function theorem as in the proof of statement 2 in Lemma \ref{lem_analytic_domain} gives that $\big|\lambda_*-\lambda_{\mathfrak{n}_*,\epsilon}(p)\big|\gtrsim |\epsilon|^{\frac{1}{3}}$. Consequently,
\begin{equation*}
\begin{aligned}
\big|\lambda_*+\epsilon h-\lambda_{\mathfrak{n}_*,\epsilon}(p)\big|
\geq
\big|\lambda_{\mathfrak{n}_*,\epsilon}(q_*)-\lambda_{\mathfrak{n}_*,\epsilon}(p)\big|
-
\big|\lambda_*+\epsilon h-\lambda_{\mathfrak{n}_*,\epsilon}(q_*)\big|
\gtrsim \epsilon^{\frac{1}{3}}-\epsilon
\gtrsim \epsilon^{\frac{1}{3}}
\end{aligned}
\end{equation*}
for $p\in \gamma_{-,\epsilon}$, $h\in\overline{\mathcal{J}}$. It follows that
\[
\Big|\frac{1}{\lambda_*+\epsilon h-\lambda_{\mathfrak{n}_*,\epsilon}(p)} - \frac{1}{\lambda_*-\lambda_{\mathfrak{n}_*,\epsilon}(p)}\Big| =\Big|\frac{\epsilon h}{(\lambda_*+\epsilon h-\lambda_{\mathfrak{n}_*,\epsilon}(p))\cdot(\lambda_*-\lambda_{\mathfrak{n}_*,\epsilon}(p))}\Big| 
\lesssim \epsilon^{\frac{1}{3}}|h|. 
\]
On the other hand, for any $\varphi\in \tilde{H}^{-\frac{1}{2}}(\Gamma)$, 
\[
|\langle \varphi(\cdot),\overline{u_{\mathfrak{n}_*,\epsilon}(\cdot ;\overline{p})} \rangle| \lesssim \|\varphi\|_{\tilde{H}^{-\frac{1}{2}}(\Gamma)}. 
\]
Therefore
\begin{equation*}
\Big\|
\int_{\gamma_{-,\epsilon}}
\big(\mathbb{K}_{\mathfrak{n}_{*},\epsilon}(p,\lambda_*+\epsilon h)\varphi\big)dp
-
\int_{\gamma_{-,\epsilon}}
\big(\mathbb{K}_{\mathfrak{n}_{*},\epsilon}(p,\lambda_*)\varphi\big)dp
\Big\|_{H^\frac{1}{2}(\Gamma)}
\lesssim \epsilon^{\frac{1}{3}}|h|\|\varphi\|_{\tilde{H}^{-\frac{1}{2}}(\Gamma)}
\end{equation*}
whence \eqref{eq_app_B_19} follows. 
\end{proof}

\begin{proof}[Proof of Lemma \ref{lem_app_B_3}]
To prove \eqref{eq_app_B_29}, we first note the parity of the perturbed eigenfunction $u_{\mathfrak{n}_{*},\epsilon}(x;p)\sim (\mathcal{P}u_{\mathfrak{n}_{*},\epsilon})(x;-p)$ (similar to Lemma \ref{lemma_equivalence}). In particular, when $x\in\Gamma$,  $u_{\mathfrak{n}_{*},\epsilon}(x;p)\sim u_{\mathfrak{n}_{*},\epsilon}(x;-p)$. Thus $\mathbb{K}_{\mathfrak{n}_{*},\epsilon}(-p,\lambda_*+\epsilon h)=\mathbb{K}_{\mathfrak{n}_{*},\epsilon}(p,\lambda_*+\epsilon h)$, and consequently, 
\begin{equation*}
\begin{aligned}
\mathbb{J}_{1, \epsilon}(h)
&=\frac{1}{\pi}
\int_{0}^{\epsilon^{\frac{1}{3}}}
\mathbb{K}_{\mathfrak{n}_{*},\epsilon}(p,\lambda_*+\epsilon h)dp.
\end{aligned}
\end{equation*}
In light of Theorem \ref{thm_local_gap_open}, we extract the leading-order term of $\mathbb{J}_{1, \epsilon}$ using \eqref{eq_perturbed_dirac_dispersion} and \eqref{eq_perturbed_dirac_eigenfunction_1}. More precisely, we write
$$
\mathbb{J}_{1, \epsilon}(h) = \mathbb{J}_{1, \epsilon}^{(0)}(h) + \mathbb{J}_{1, \epsilon}^{(1)}(h),
$$
where $\mathbb{J}_{1, \epsilon}^{(0)}(h)$ is associated with the leading-order kernel
\begin{equation*}
J_{1, \epsilon}^{(0)}(x,y;\ h)=\frac{1}{\pi}\int_{0}^{\epsilon^{\frac{1}{3}}}
\frac{\big(f_{\epsilon}(p)v_{n_*}(x;0)+v_{m_*}(x;0)\big)\cdot \big(\overline{f_{\epsilon}(p)}\overline{v_{n_*}(y;0)}+\overline{v_{m_*}(y;0)}\big)}{\big(\epsilon h+\sqrt{\alpha_{\mathfrak{n}_*}^2p^2+|t_*|^2\epsilon^2}\big)\cdot N^2_{\mathfrak{n}_*,\epsilon}(p)}dp.
\end{equation*}
Here
\begin{equation*}
f_{\epsilon}(p):=\frac{t_*\cdot \epsilon}{\alpha_{n_*}p+\sqrt{\alpha_{n_*}^2p^2+|t_*|^2\epsilon^2}},\quad
N_{n,\pm\epsilon}(p)=(1+|f_{\epsilon}(p)|^2+\mathcal{O}(p+\epsilon))^\frac{1}{2},\quad 
(n=\mathfrak{n}_*,\mathfrak{n}_*+1).
\end{equation*}
We claim that
\begin{equation} \label{eq_app_B_28}
\lim_{\epsilon\to 0}\mathbb{J}_{1, \epsilon}^{(1)}(h)=0.
\end{equation}
Indeed, by Theorem \ref{thm_local_gap_open}, one can show that
\begin{equation} \label{eq_app_B_24}
\begin{aligned}
\|\mathbb{J}_{1, \epsilon}^{(1)}(h)\|=\|\mathbb{J}_{1, \epsilon}^{(0)}(h)\|\cdot \mathcal{O}(\epsilon^{\frac{1}{3}}). 
\end{aligned}
\end{equation}
On the other hand, using the explicit expression of the kernel of the operator $\mathbb{J}_{1, \epsilon}^{(0)}$, we have 
\begin{equation*}
\begin{aligned}
\|\mathbb{J}_{1, \epsilon}^{(0)}(h)&\|
\leq 
\int_{0}^{\epsilon^{\frac{1}{3}}}
\frac{\|f_{\epsilon}(p)v_{n_*}(x;0)+v_{m_*}(x;0)\|_{H^{\frac{1}{2}}(\Gamma)} \cdot \|\overline{f_{\epsilon}(p)}\overline{v_{n_*}(y;0)}+\overline{v_{m_*}(y;0)}\|_{H^{\frac{1}{2}}(\Gamma)}}{\big(\epsilon h+\sqrt{\alpha_{n_*}^2p^2+|t_*|^2\epsilon^2}\big)(1+|f_{\epsilon}(p)|^2)}dp
 \\
&\lesssim
\int_{0}^{\epsilon^{\frac{1}{3}}}
\frac{|f_{\epsilon}(p)|^2 + |f_{\epsilon}(p)|+1}{(\epsilon h+\sqrt{\alpha_{n_*}^2p^2+|t_*|^2\epsilon^2}) \cdot (1+|f_{\epsilon}(p)|^2)}dp\\
 &\lesssim
\int_{0}^{\epsilon^{\frac{1}{3}}}\Big|\frac{1}{\epsilon h+\sqrt{\alpha_{n_*}^2p^2+|t_*|^2\epsilon^2}}\Big|dp =\frac{|t_*|}{\alpha_*}\int_{0}^{tan^{-1} (\frac{\alpha_*}{|t_*|}\epsilon^{-\frac{2}{3}})}\Big|\frac{\sec^2(\theta)}{h+|t_*|\sec(\theta)}\Big|d\theta
=\mathcal{O}(\log (\epsilon)). 
\end{aligned}
\end{equation*}
Therefore \eqref{eq_app_B_28} follows.

Similarly, one can show that 
$$
\lim_{\epsilon\to 0}\mathbb{J}_{2, \epsilon}^{(1)}(h)= \lim_{\epsilon\to 0}\mathbb{J}_{1, -\epsilon}^{(1)}(h)=\lim_{\epsilon\to 0}\mathbb{J}_{2, -\epsilon}^{(1)}(h)=0.
$$
Then
$$\lim_{\epsilon\to 0}(\mathbb{J}_{1, \epsilon}(h)+\mathbb{J}_{2, \epsilon}(h)+\mathbb{J}_{1, -\epsilon}(h)
+\mathbb{J}_{2, -\epsilon}(h))=  \lim_{\epsilon\to 0}\big(\mathbb{J}_{1, \epsilon}^{(0)}(h)+\mathbb{J}_{2, \epsilon}^{(0)}(h)+\mathbb{J}_{1, -\epsilon}^{(0)}(h)
+\mathbb{J}_{2, -\epsilon}^{(0)}(h)\big).
$$
On the other hand, explicit calculation yields 
\begin{equation*}
\begin{aligned}
J_{1, \epsilon}^{(0)}+J_{2, \epsilon}^{(0)}+J_{1, -\epsilon}^{(0)}
+J_{2, -\epsilon}^{(0)}
=
&
\int_{0}^{\epsilon^{\frac{1}{3}}}
\frac{|f_{\epsilon}(p)|^2v_{n_*}(x;0)\overline{v_{n_*}(y;0)}+v_{m_*}(x;0)\overline{v_{m_*}(y;0)}}{\big(\epsilon h+\sqrt{\alpha_{n_*}^2p^2+|t_*|^2\epsilon^2}\big)(1+|f_{\epsilon}(p)|^2)}dp \\
&+
\int_{0}^{\epsilon^{\frac{1}{3}}}
\frac{v_{n_*}(x;0)\overline{v_{n_*}(y;0)}+|f_{\epsilon}(p)|^2v_{m_*}(x;0)\overline{v_{m_*}(y;0)}}{\big(\epsilon h-\sqrt{\alpha_{n_*}^2p^2+|t_*|^2\epsilon^2}\big)(1+|f_{\epsilon}(p)|^2)}dp.
\end{aligned}
\end{equation*}
Using Lemma \ref{lemma_equivalence}, $v_{n_*}(x;0)\overline{v_{n_*}(y;0)}=v_{m_*}(x;0)\overline{v_{m_*}(y;0)}$. Thus,
\begin{equation*}
\begin{aligned}
&\lim_{\epsilon\to 0}(J_{1, \epsilon}^{(0)}+J_{2, \epsilon}^{(0)}+J_{1, -\epsilon}^{(0)}
+J_{2, -\epsilon}^{(0)}) \\
&=\lim_{\epsilon\to 0}
\Big(\frac{1}{\pi}
\int_{0}^{\epsilon^{\frac{1}{3}}}
\frac{1}{\epsilon h+\sqrt{\alpha_{n_*}^2p^2+|t_*|^2\epsilon^2}}dp +\frac{1}{\pi}
\int_{0}^{\epsilon^{\frac{1}{3}}}
\frac{1}{\epsilon h-\sqrt{\alpha_{n_*}^2p^2+|t_*|^2\epsilon^2}}dp\Big)v_{n_*}(x;0)\overline{v_{n_*}(y;0)} \\
&=-\frac{1}{|t_*|\alpha_*}\frac{h}{\sqrt{1-\frac{h^2}{|t_*|^2}}}v_{n_*}(x;0)\overline{v_{n_*}(y;0)}
=\beta(h)v_{n_*}(x;0)\overline{v_{n_*}(y;0)},
\end{aligned}
\end{equation*}
and consequently,
$$\lim_{\epsilon\to 0}(\mathbb{J}_{1, \epsilon}^{(0)}+\mathbb{J}_{2, \epsilon}^{(0)}+\mathbb{J}_{1, -\epsilon}^{(0)}
+\mathbb{J}_{2, -\epsilon}^{(0)})=\beta(h)\mathbb{P}^{Dirac}. $$
It follows that $\lim_{\epsilon\to 0}(\mathbb{J}_{1, \epsilon}+\mathbb{J}_{2, \epsilon}+\mathbb{J}_{1, -\epsilon}
+\mathbb{J}_{2, -\epsilon})=\beta(h)\mathbb{P}^{Dirac}$. This completes the proof of \eqref{eq_app_B_29}. 
\end{proof}

\subsection{Appendix C: Proof of Proposition \ref{prop_T_fred_kernel}}
\setcounter{equation}{0}
\setcounter{theorem}{0}
\setcounter{subsection}{0}
\renewcommand{\theequation}{C.\arabic{equation}}
\renewcommand{\thetheorem}{C.\arabic{theorem}}
\renewcommand{\thesubsection}{C.\arabic{subsection}}

\begin{proof}[Proposition \ref{prop_T_fred_kernel}]
We show that $\ker(\mathbb{T})=\text{span}\Big\{\frac{\partial v_{n_*}(x;0)}{\partial x_1}\Big|_{\Gamma}\Big\}$. The proof that $\mathbb{T}$ is Fredholm follows exactly as in \cite[Proposition 4.4]{qiu2023mathematical}. 

By \eqref{eq_G_pv_decomposition} and the definition of $
\mathbb{T}$, 
\begin{equation*}
\begin{aligned}
\int_{\Gamma}G(y,x;\lambda_* )\varphi(x)dx_2
=\mathbb{T}(\varphi)(y)-\frac{i}{2}\frac{v_{n_{*}}(y;0)\langle \varphi(x),\overline{v_{n_{*}}(x;0)}\rangle}{|\mu_{n_{*}}^{\prime}(0)|}
-\frac{i}{2}\frac{v_{m_{*}}(y;0)\langle \varphi(x),\overline{v_{m_{*}}(x;0)}\rangle}{|\mu_{m_{*}}^{\prime}(0)|}.
\end{aligned}
\end{equation*}
Formula \eqref{eq-repre-u} indicates that
\begin{equation*}
\begin{aligned}
\frac{1}{2}&v_{n_*}(y;0)
=\int_{\Gamma} G(y,x;\lambda_* )\frac{\partial v_{n_*}(x;0)}{\partial x_1}dx_2 \\
&=\mathbb{T}\Big(\frac{\partial v_{n_*}(x;0)}{\partial x_1}\Big|_{\Gamma}\Big)(y)
-\frac{i}{2}\frac{v_{n_{*}}(y;0)}{|\mu_{n_{*}}^{\prime}(0)|}\int_{\Gamma}\frac{\partial v_{n_*}(x;0)}{\partial x_1}\overline{v_{n_{*}}(x;0)}dx_2
+\frac{i}{2}\frac{v_{m_{*}}(y;0)}{|\mu_{m_{*}}^{\prime}(0)|}\int_{\Gamma}\frac{\partial v_{m_*}(x;0)}{\partial x_1}\overline{v_{m_{*}}(x;0)}dx_2.
\end{aligned}
\end{equation*}
By Lemma \ref{lemma_equivalence}, 
$v_{m_{*}}(x_1, x_2;0)\sim v_{n_{*}}(-x_1, x_2;0)$. Moreover, Assumption \ref{assump_dirac_points} states that $\mu_{n_{*}}^{\prime}(0)>0$. Since the operator $\mathcal{L}$ is reflection symmetric, so are its dispersion curves, thus yielding $\mu_{m_{*}}^{\prime}(0)= -\mu_{n_{*}}^{\prime}(0)$. Therefore, it follows that 
\begin{equation*}
\begin{aligned}
\frac{1}{2}v_{n_*}(y;0)
&=\mathbb{T}\Big(\frac{\partial v_{n_*}(x;0)}{\partial x_1}\Big|_{\Gamma}\Big)(y) 
-i\frac{v_{n_{*}}(y;0)}{\mu_{n_{*}}^{\prime}(0)}\cdot \int_{\Gamma}\frac{\partial v_{n_*}(x;0)}{\partial x_1}\overline{v_{n_{*}}(x;0)}dx_2.
\end{aligned}
\end{equation*}
Then, by Lemma \ref{lemma_energy_flux}, 
\begin{equation*}
\begin{aligned}
\mathbb{T}\Big(\frac{\partial v_{n_*}(x;0)}{\partial x_1}\Big|_{\Gamma}\Big)(y)
=\frac{1}{2}v_{n_*}(y;0) 
+i\frac{v_{n_{*}}(y;0)}{\mu_{n_{*}}^{\prime}(0)}\cdot
\frac{i}{2}\mu_{n_{*}}^{\prime}(0)=0,
\end{aligned}
\end{equation*}
which indicates that $\ker(\mathbb{T})\supset \text{span}\Big\{\frac{\partial v_{n_*}(x;0)}{\partial x_1}\Big|_{\Gamma}\Big\}$. 

We next show that $\ker(\mathbb{T})\subset \text{span}\Big\{\frac{\partial v_{n_*}(x;0)}{\partial x_1}\Big|_{\Gamma}\Big\}$. Using Lemma \ref{lemma_energy_flux}, the following decomposition holds for all $\varphi\in \tilde{H}^{-\frac{1}{2}}(\Gamma)$
\begin{equation*}
\varphi=\frac{\langle\varphi,\overline{v_{n_*}(x;0)}\rangle}{i\mu_{n_{*}}^{\prime}(0)/2}\frac{\partial v_{n_*}(x;0)}{\partial x_1}\Big|_{\Gamma}  +\tilde{\varphi}
\quad\text{with}\quad
\langle \tilde{\varphi},\overline{v_{n_*}(x;0)}\rangle=0.
\end{equation*}
Thus, to prove $\ker(\mathbb{T})\subset \text{span}\Big\{\frac{\partial v_{n_*}(x;0)}{\partial x_1}\Big|_{\Gamma}\Big\}$, it is sufficient to show that $\varphi=0$ if we have $\varphi\in \ker(\mathbb{T})$ and $\langle\varphi,\overline{v_{n_*}(x;0)}\rangle=0$. Indeed, suppose the latter holds. Then  \eqref{eq_G_pv_decomposition}-\eqref{eq_G_right_decomposition} and \eqref{eq-T} imply that
\begin{equation*}
\begin{aligned}
\int_{\Gamma} G_{0}^{+}(y,x;\lambda_* )\varphi(x)dx_2
&=i\frac{v_{m_{*}}(y;q_*)}{|\mu_{m_{*}}^{\prime}(q_*)|}\int_{\Gamma}\varphi(x)\overline{v_{m_{*}}(x;q_*)}dx_2 \\
&\quad +\frac{i}{2}\frac{v_{n_{*}}(y;0)}{|\mu_{n_{*}}^{\prime}(0)|}\int_{\Gamma}\varphi(x)\overline{v_{n_{*}}(x;0)}dx_2
-\frac{i}{2}\frac{v_{m_{*}}(y;0)}{|\mu_{m_{*}}^{\prime}(0)|}\int_{\Gamma}\varphi(x)\overline{v_{m_{*}}(x;0)}dx_2.
\end{aligned}
\end{equation*}
By a similar argument as in Case 1 of the proof of Proposition \ref{prop_charac_to_resonance}, we see that $\varphi\in \ker(\mathbb{T})$ implies that
\begin{equation*}
\langle\varphi,\overline{v_{n_*}(x;-q_*)}\rangle=
\langle\varphi,\overline{v_{m_*}(x;q_*)}\rangle=0.
\end{equation*}
In addition, since $\langle\varphi(\cdot),\overline{v_{n_{*}}(\cdot;0)}\rangle =0$, Lemma \ref{lemma_equivalence} gives $\langle\varphi(\cdot),\overline{v_{m_{*}}(\cdot;0)}\rangle=0$. Thus,
\begin{equation*}
\int_{\Gamma} G_{0}^{+}(y,x;\lambda_* )\varphi(x)dx_2=0,\quad y\in\Gamma.
\end{equation*}
Define $v(y):=\int_{\Gamma} G_{0}^{+}(y,x;\lambda_* )\varphi(x)dx_2$ for $y\in\Omega^{+}$. It is clear that $v\big|_{\Gamma}=0$. We claim that $v(x)\equiv 0$ for $x\in\Omega^{+}$. To see this, we consider the odd extension of $v(x)$ to $x\in \Omega$: 
\begin{equation*}
    \tilde{v}(x_1,x_2)=\left\{
    \begin{aligned}
    &v(x_1,x_2),\quad x_1\geq 0, \\
    &-v(-x_1,x_2),\quad x_1<0.
    \end{aligned}
    \right.
\end{equation*}
Note that both $\tilde{v}$ and $\nabla \tilde{v}$ are continuous across the interface $\Gamma$, and $\tilde{v}(x)$ decays exponentially as $|x_1|\to +\infty$. Thus $\|\tilde{v}\|_{L^2(\Omega)}<\infty$ and $(\mathcal{L}-\lambda_*)\tilde{v}=0$. This implies $\lambda_*$ is an eigenvalue of the periodic operator $\mathcal{L}$, which contradicts the absolute continuity of the spectrum $\sigma(\mathcal{L})$ \cite{kuchment2016overview}. This contradiction proves that $v(x)\equiv 0$. Moreover,  \eqref{eq_G_right_decomposition} and the identity $\langle\varphi(\cdot),\overline{v_{n_{*}}(\cdot;0)}\rangle=\langle\varphi(\cdot),\overline{v_{m_{*}}(\cdot;q_*)}\rangle=0$ imply that 
\begin{equation*}
\int_{\Gamma} G(y,x;\lambda_* )\varphi(x)dx_2=0,\quad y\in\Omega^{+}.
\end{equation*}
Applying the operator $\mathcal{L}-\lambda_*$ to both sides of the above equation,  we get $\varphi=0$.
This concludes the proof of $\ker(\mathbb{T})\subset \text{span}\Big\{\frac{\partial v_{n_*}(x;0)}{\partial x_1}\Big|_{\Gamma}\Big\}$, and hence 
$\ker(\mathbb{T})=\text{span}\Big\{\frac{\partial v_{n_*}(x;0)}{\partial x_1}\Big|_{\Gamma}\Big\}$ follows.
\end{proof}

\medskip
\textbf{Data availability statement}: No new data were created or analyzed during this study. Data sharing is not applicable to this article.

\bibliographystyle{plain}
\bibliography{ref}

\end{document}